\documentclass{article}
\usepackage{arxiv}
\usepackage[T1]{fontenc}
\usepackage{graphicx}
\usepackage{booktabs}
\usepackage{enumitem}
\usepackage{amsmath}
\usepackage{amsfonts}
\usepackage{amssymb}
\usepackage{mathtools}
\usepackage{color}
\usepackage{bbm}
\usepackage{fullpage}
\usepackage{subcaption}
\usepackage[round]{natbib}
\usepackage{dsfont}

\renewcommand{\l}{\ell}

\newcommand{\calA}{{\cal A}}
\newcommand{\calB}{{\cal B}}

\newcommand{\calE}{{\cal E}}
\newcommand{\calF}{{\cal F}}
\newcommand{\calG}{{\cal G}}
\newcommand{\calH}{{\cal H}}
\newcommand{\calI}{{\cal I}}

\newcommand{\calK}{{\cal K}}
\newcommand{\calL}{{\cal L}}

\newcommand{\calN}{{\cal N}}

\newcommand{\calR}{{\cal R}}
\newcommand{\calS}{{\cal S}}

\newcommand{\calX}{{\cal X}}
\newcommand{\calY}{{\cal Y}}
\newcommand{\calZ}{{\cal Z}}

\newcommand{\eqdef}{\stackrel{\text{\tiny def}}{=}} 
\newcommand{%
  \sjacobian}{\mathbf{SJ}}
\newcommand{%
  \jacobian}{\mathbf{J}}
\newcommand{%
  \hessian}{\mathbf{H}}

\usepackage[linktocpage=true]{hyperref}
\hypersetup{
  colorlinks=true,
  citecolor=blue,
  linkcolor=magenta,
}
\usepackage{comment}
\usepackage{stmaryrd}

\usepackage{tikz}
\DeclareRobustCommand{\mycircle}{%
  \tikz[baseline=-0.5ex]\draw[blue, line width=0.8pt] (0,0) circle (0.8ex);%
}
\DeclareRobustCommand{\fmycircle}{%
  \tikz[baseline=-0.5ex]\filldraw[fill=red, draw=black] (0,0) circle (0.8ex);%
}

\DeclareRobustCommand{\mydiamond}{%
  \tikz[baseline=-0.5ex]\draw[green!55!black, line width =0.8pt] 
    (0,0.8ex) -- (0.8ex,0) -- (0,-0.8ex) -- (-0.8ex,0) -- cycle;%
}
\usetikzlibrary{calc, positioning, fit, shapes.geometric,trees}
\usepackage{sectsty}
\sectionfont{\large}
\subsectionfont{\normalsize}
\subsubsectionfont{\normalsize}
\paragraphfont{\normalsize}

\RequirePackage{algorithm}
\RequirePackage{algpseudocode}
\usepackage[linktocpage=true]{hyperref}
\usepackage[capitalize,noabbrev]{cleveref}
\crefname{assumption}{Assumption}{Assumptions}
\usepackage{amsthm}
\usepackage{thmtools}
\usepackage{thm-restate}

\declaretheorem[name=Theorem, numberwithin=section]{theorem}
\declaretheorem[name=Proposition, numberlike=theorem]{proposition}

\declaretheorem[name=Lemma, numberlike=theorem]{lemma}

\declaretheorem[name=Definition, numberlike=theorem]{definition}

\declaretheorem[name=Corollary, numberlike=theorem]{corollary}
\declaretheorem[name=Remark, numberlike=theorem]{remark}

\declaretheorem[name=Example, numberlike=theorem]{example}
\declaretheorem[name=Assumption, numberlike=theorem]{assumption}

\usepackage{multicol}

\newcommand\todo[1]{\noindent\textcolor{red}{#1}}
\newcommand{\argmin}{\operatornamewithlimits{argmin}}
\newcommand{\argmax}{\operatornamewithlimits{argmax}}

\usepackage{footmisc}
\usepackage{authblk}

\begin{document}
\raggedbottom
\renewcommand{\headeright}{}
\renewcommand{\undertitle}{}
\newcommand\blfootnote[1]{%
  \begingroup
  \renewcommand\thefootnote{}\footnote{#1}%
  \addtocounter{footnote}{-1}%
  \endgroup
}
\title{Equilibrium Computation in Extensive-Form Games with Stochastic Action Sets}
%
%

\renewcommand*{\Affilfont}{\normalsize\normalfont}
\renewcommand*{\Authfont}{\bfseries\large}
\newcommand*\samethanks[1][\value{Joint first authors}]{\footnotemark[#1]}
\author{Thomas Schwarz\thanks{Joint first authors.}}
\author{Jiaru Li\protect\footnotemark[1]}
\author{Ryann Sim}
\author{Chun Kai Ling}
\affil{School of Computing, National University of Singapore}
\date{}
%
%
%
\begingroup
    \hypersetup{hidelinks}
  \maketitle
\endgroup

%
\begin{abstract}
\blfootnote{Contact: \texttt{tschwarz@comp.nus.edu.sg},
\texttt{jiaru.li@u.nus.edu},
\texttt{ryann.sim@nus.edu.sg}, \texttt{chunkail@nus.edu.sg}}%
Extensive-form games (EFGs) are a standard model for sequential decision-making in games. A fundamental and typically implicit assumption in EFGs is that players always have access to all of their actions at every decision point. However, in many settings, certain actions might be unavailable during game-play due to exogenous stochasticity.
Given a base EFG, we formalize a model that allows for actions to be stochastically restricted, leading to a corresponding Extensive-Form Game with Stochastic Action Sets (EFGSAS). In EFGSAS, we derive an expansion procedure that results in an equivalent EFG, thus showing that standard strategy formalisms could require exponentially large representations. However, under an appropriate independence assumption, we show that approximate compact strategy representations polynomial in the size of the base EFG exist. Computationally, we introduce an algorithm called CAT-CFR, which minimizes ranking regret and whose averaged induced sequence-form strategies converge to approximate Nash equilibria with high probability in two-player zero-sum EFGSAS. Finally, we utilize a stochastic approximation procedure to recover compact representations of approximate Nash equilibria, utilizing the iterates of CAT-CFR and additional action availability samples.

\end{abstract}

\section{Introduction}\label{sec:intro}
Extensive-form games (EFGs) model sequential decision-making under imperfect information, with many applications spanning economics and artificial intelligence~\citep{kuhn1953extensive,brown2018superhuman,meta2022human}. A fundamental assumption enabling efficient equilibrium analysis in EFGs is that at each decision point, players always have access to all of their actions. However, in many realistic settings, this assumption does not hold due to exogenous stochasticity.
For instance, consider a modified pursuit-evasion game~\citep{isaacs1965differential} where a defender pursues an attacker through a road network. Due to temporary road-access restrictions that occur independently across players, some exits might be unavailable, and each player observes their restrictions only upon reaching the corresponding intersection.  
Consider also the recreational two-player zero-sum game Backgammon, where at each round, independent dice rolls determine the available actions for each player, restricting their playable actions stochastically.
In these examples, the players' decision-making process is greatly complicated under the standard EFG model, since their strategy at each information set should now be conditioned on available actions and they might be unable to play optimal strategies in the base game, not to mention additional considerations about distributions over opponent actions.


We introduce a formal model of this setting, which we call Extensive-Form Games with Stochastic Action Sets (EFGSAS). An EFGSAS $\calG$ requires as part of its definition a base EFG $\calG^{\mathrm{orig}}$, which admits information sets (hereafter infosets) $I\in\mathcal{I}_i$ belonging to player $i$ and corresponding action sets $A_{I,i}$ at each infoset. Then, the sets $\mathcal{S}_{I,i} \subseteq 2^{ A_{I,i}} \backslash \{ \varnothing \}$ encode the possible action availability subsets in each infoset, and
$\rho_{I,i} \in \Delta(\mathcal{S}_{I,i})$ is a probability distribution over elements $\mathcal{S}_{I,i}$. Each time the game is played, Nature draws from $\rho_{I,i}$, which determines the available actions at each infoset in the game, and players are restricted to only playing available actions. 

In this work, we focus on understanding the computational properties of EFGSAS, studying efficient ways to represent and compute Nash equilibria (NE). 
As it turns out, even defining a strategy in an EFGSAS is non-trivial. In a normal-form Game with Stochastic Action Sets (GSAS)~\citep{schwarz2026computing}, the natural approach is to expand the game into a larger induced normal-form representation similar to Bayesian games~\citep{harsanyi1968games}, which then allows for strategies to na\"ively be represented as mappings from all possible action subsets to distributions over actions. However, the sequential nature of EFGSAS leads to a fresh but important modeling decision: though action availabilities might be sampled before game-play occurs, the time at which the information is \textit{disclosed} can significantly affect players' strategy sets. Here, there are two extremes, ex-interim disclosures, where available actions are revealed to a player only upon reaching an infoset, and ex-ante disclosures, where action availabilities of all infosets for each player are given before the game begins.

Formalizing the representation of a strategy in an EFGSAS thus requires care. In order for the standard definitions of pure, mixed, and behavioral strategies from EFGs to carry over to EFGSAS, the game can be converted into an exponentially larger expanded form that uses chance nodes to model Nature sampling action availabilities. The conversion needs to be consistent with both extremes of action disclosure, and implies that even specifying an NE in EFGSAS can require a prohibitively large representation. We seek conditions under which compact (i.e., polynomial in the size of the base EFG) representations of NE can be obtained. This type of result is in principle analogous to Kuhn's theorem~\citep{kuhn1953extensive}, which establishes that in EFGs with perfect recall, NE can be captured by \emph{behavioral}, as
opposed to exponentially sized normal-form \emph{mixed} strategies.

Beyond equilibrium representation, the canonical framework of counterfactual regret minimization (CFR) has seen significant success for equilibrium \emph{computation} in large-scale imperfect information EFGs~\citep{zinkevich2007regret,lanctot2009monte,brown2018superhuman}, particularly in the two-player zero-sum (2p0s) setting. However, in EFGSAS, the connection between external-regret minimization and NE cannot be applied directly, since the action sets for each player can vary every time the game is played. In the case that the solver is fully aware of the action availability distribution $\rho$, it is possible to perform the na\"ive expansion, which would allow for the application of CFR-type algorithms. However, these algorithms scale poorly, since the regret accrued depends on the (possibly exponential) size of the expanded game tree. Overall, this motivates the design of computational techniques that can efficiently compute compact representations of NE
in 2p0s-EFGSAS.


\paragraph{Our contributions.}
Our primary contribution is an end-to-end framework for computing compact representations of \(\epsilon\)-NE in ex-interim 2p0s-EFGSAS.
We first formalize EFGSAS and analyze its properties.
A key challenge is that EFGSAS strategies can depend on the observed history of action availabilities and disclosure regime.
Under an independence assumption, we show that it is possible to compress these histories while maintaining outcome equivalence with standard strategies, giving \emph{DAG-plex} representations across the spectrum of disclosure regimes from ex-interim to ex-ante. 
For ex-interim EFGSAS, we  extend results from the normal-form GSAS setting to the sequential setting, yielding approximate \emph{compact} representations of implementable sequence-form strategies whose size is polynomial in the size of the base game.
Subsequently, we combine regret minimizers developed for GSAS with the CFR framework to obtain CAT-CFR which minimizes \emph{ranking regret}. 
Our analysis accounts for both stochastic action availability and sampling error in the counterfactual feedback, yielding regret bounds and high-probability guarantees for the time-averaged marginal strategies from CAT-CFR.
Finally, 
we apply stochastic approximation techniques to the CAT-CFR outputs and additional action availability samples to construct compact strategies and show that they approximately implement NE via a finite-time duality-gap bound.

\section{Related Work}\label{appsec:related}
\paragraph{Equilibrium computation in EFGs.} 
CFR has been one of the most widely successful frameworks for learning equilibria in large-scale imperfect information EFGs via self-play~\citep{zinkevich2007regret,lanctot2009monte,farina2020stochastic}, leading to the development of state-of-the-art poker algorithms~\citep{bowling2015heads,brown2018superhuman}. More broadly, regret minimization techniques have been used to obtain Nash equilibria in zero-sum EFGs~\citep{hoda2010smoothing,farina2019optimistic,lee2021last}, as well as extensive-form correlated equilibria (EFCE) in general-sum EFGs~\citep{farina2019efficient,celli2020no,anagnostides2022faster}. 
\paragraph{Sleeping bandits.} In the multi-armed bandit literature, the \emph{sleeping bandit} setting studies regret minimization when arms are available stochastically or adversarially~\citep{auer2002finite,blum2007external,kanade2009sleeping,kleinberg2010regret,kanade2014learning,saha2020improved,gaillard2023one,nguyen2024near}. These techniques have been applied to online combinatorial optimization~\citep{neu2014online,kale2016hardness} and reinforcement learning~\citep{drago2025sleeping}. 
\paragraph{Normal-Form GSAS} The notion of games with stochastic action sets (GSAS) was introduced by~\cite{schwarz2026computing}, which focused on normal-form games. Given a base game, the corresponding GSAS includes in its specification the possible availability subsets and a distribution over the action availabilities. While some of our definitions and results are analogous to the normal-form setting, the sequential nature of EFGSAS requires additional technical insights. Indeed, to deal with the sequential structure of EFGSAS, our compact representation result requires both a DAG-plex compression and an extension of normal-form compactness for sequential action spaces. Moreover, our computational results utilize techniques from the EFG literature to extend the algorithm of~\cite{schwarz2026computing} to EFGSAS. Finally, while most prior work on CFR-type algorithms focus primarily on regret bounds, we additionally bound the duality gap of the induced compact vectors using our method.
For completeness, we provide the formal definitions and key properties of GSAS in~\cref{appsec:gsasprelims}.

\section{Preliminaries}
\label{sec:prelim}
\paragraph{Notation.} Denote the $n$-dimensional non-negative orthant by $\mathbb{R}^n_{\ge0}$. For a finite set $S$, $\Delta(S)$ is the associated probability simplex 
$\{ x \in \mathbb{R}^{|S|}_{\ge0} | \sum_{s\in S} x(s) = 1 \}$, such that if $y \in \Delta(S), s \in S, y(s)$ is the probability 
that item $s$ is selected. Denote \([r]^+\coloneqq\max\{r,0\}\). 

\paragraph{Extensive-Form Games.}  
An $n$-player \textit{extensive-form game} (\textit{EFG}) $\calG$ is a tuple $\calG \coloneqq ( \mathcal{H},A,\calN, \mathcal{Z}, u, \mathcal{I} )$ where
        the nonempty, finite set $\mathcal{H}$ denotes the \textit{states} of the game which form a tree rooted at an initial state $r \in \mathcal{H}$. We denote \textit{terminal nodes} in $\mathcal{H}$ by $\mathcal{Z}$. Each \textit{nonterminal state} $h \in \mathcal{H}\setminus\mathcal{Z}$ is associated with a set of possible \textit{actions} $A_{h}$.
         Given $\mathcal{N} = \{1,\dots,n\}$, the set $\mathcal{N}\cup\{c\}$ denotes  the $n+1$ \textit{players} of the game. Each state $h \in \mathcal{H}\setminus\mathcal{Z}$ admits a label $\mathrm{player}(h) \in \mathcal{N}\cup\{c\}$ which denotes the \textit{acting player} at state {$h$}. The letter $c$ denotes a \textit{chance player}, 
        representing exogenous stochasticity. $\mathcal{H}_i \subseteq \mathcal{H}\setminus\mathcal{Z}$ denotes the states $h \in \mathcal{H}\setminus\mathcal{Z}$ with $\mathrm{player}(h) = i$.
        Each chance node $h\in\mathcal{H}_c$ is associated with a fixed distribution $\mathbb{P}_c(\cdot\mid h)$ over $A_h$, denoting the distribution over actions chosen by the chance player at each node. 
        For each $i\in\mathcal{N}$, the \textit{payoff} function $u_i : \mathcal{Z}\to[-1,1]$ specifies the payoff that player $i$ receives if the game ends at terminal state $z\in\mathcal{Z}$.
        
        
        The players' decision points $\bigcup_{i\in \calN}\mathcal H_i$ 
        are partitioned into \textit{information sets} (\textit{infosets}) ascribed to each player, which are collected into sets $\{\mathcal{I}_1,\ldots,\mathcal{I}_n\}$. Each infoset $I \in \mathcal{I}_i$ contains nodes that the acting player $i$ cannot distinguish between, i.e.,
		$h_1, h_2 \in I$ implies $A_{h_1} = A_{h_2}$. We let $A_{I,i}$ denote the shared action set of infoset $I$ belonging to player $i$ and use $A_I$ when there is no player ambiguity.
        For notational convenience, we ascribe a singleton information set to each chance node and define $\mathcal{I}_c$ as the collection of these chance node infosets. For each nonterminal node {$h \in\mathcal{H} \setminus \mathcal{Z}$}, we thus define $I_h \in \bigcup_{i\in\mathcal{N}}\mathcal{I}_i \cup \mathcal{I}_c$ to be the infoset it belongs to.

In this paper, we make the standard assumption that EFGs exhibit {\em perfect recall}. That is, no player ever forgets their past history, the sequence of information sets visited, actions taken within those information sets, and any information acquired along the way.  
Formally, for any player $i$, any information set $I \in \mathcal{I}_i$ and any two nodes $h_1, h_2 \in I$, the sequence of player~$i$'s previously visited infosets and actions from the root $r$ to $h_1$ must coincide with the corresponding sequence from $r$ to $h_2$.


 \paragraph{Strategy Formalisms.} In EFGs, a \emph{pure strategy} specifies a deterministic action at every information set of a player. 
A (normal-form) \emph{mixed strategy} 
is a probability distribution over pure strategies.  
A \emph{behavioral strategy} specifies independent distributions at each infoset.
Formally, for any infoset $I \in \mathcal{I}_i$, $\Delta(A_I)$ denotes the probability simplex over the actions $A_I$.  
A behavioral strategy for player $i$ is a mapping $
\beta_i : \mathcal{I}_i \to \bigcup_{I \in \mathcal{I}_i} \Delta(A_I),
$
assigning to each infoset $I$ a distribution {$\beta_i(\cdot \mid I) \in \Delta(A_I)$}.  We denote the set of all behavioral strategies for each player by $\calB_i$ and let $\calB\coloneqq\prod_{i\in\mathcal N}\mathcal B_i$.
The \textit{joint behavioral strategy profile} for all players is denoted by $\beta \coloneqq (\beta_i)_{i \in \mathcal{N}} \in \calB$.  Kuhn's theorem~\citep{kuhn1953extensive} establishes the outcome equivalence of behavioral and normal-form mixed strategies in EFGs with perfect recall.
The expected payoff of player $i\in\mathcal{N}$ following (joint) behavioral strategy $\beta\in\calB$ is denoted $U_i(\beta)\coloneqq \sum_{z\in\mathcal{Z}} \mathbb{P}(z\mid \beta, r)\cdot u_i(z)$, where $ \mathbb{P}(z\mid \beta, r)$ is the probability that leaf $z\in\mathcal{Z}$ is reached from root $r$ following  $\beta$. We call $\beta^*$  an $\epsilon$-\emph{(behavioral) NE}
if for every player $i \in \mathcal{N}$, it holds that
$U_i(\beta^*) \geq \max_{\beta_i\in\calB_i} U_i(\beta_i, \beta^*_{-i}) - \epsilon$,
i.e., no player can improve their expected payoff by more than $\epsilon$ by deviating from $\beta^*$ to any other behavioral strategy.


\paragraph{Sequence-Form and Treeplexes.}
The set of \textit{sequences} of player $i$ is defined as $\Sigma_i \coloneqq  \{(I, a) : I \in \calI_i, a \in A_I \} \cup \{\varnothing\}$, where the special element $\varnothing$ is the empty sequence. Given an infoset $I \in \calI_i$,  $\sigma(I)$ denotes the \textit{parent sequence} of $I$, defined as the last pair $(I', a') \in \Sigma_i$ encountered on the path from the root to any node $h \in I$. If no such pair  exists (i.e., player $i$ never acts before any node $h \in I$), we let $\sigma(I) = \varnothing$.  We (recursively) define a sequence $\tau \in \Sigma_i$ to be a descendant of $\tau' \in \Sigma_i$, denoted by $\tau \succeq\tau'$, if $\tau = \tau'$ or if $\tau = (I, a)$ and $\sigma(I) \succeq \tau'$. 
In the \textit{sequence-form} representation~\citep{romanovskii1962reduction,von1996efficient}, players select strategies which are represented by a vector $x$ indexed by sequences $\sigma\in\Sigma_i$. For $\sigma=(I,a)$, the entry $x[\sigma]\ge 0$ captures the product of probabilities of player $i$'s actions from the root $r$ to $I$, including playing action $a$. By convention, $x[\varnothing] = 1$. This gives a set of linear constraints encoding the probability mass conservation for valid sequence-form strategies: for all $I\in\calI_i$, $\sum_{a\in A_I} x[(I,a)] = x[\sigma(I)]$. 
The set of sequence-form strategies forms a convex polytope commonly called a \emph{treeplex}~\citep{hoda2010smoothing}. 
Treeplexes are generalizations of simplices tailored towards sequential decision making. 
In two-player zero-sum (2p0s) EFGs where $n=2$ and $\forall z\in\calZ$, $u_1(z)=-u_2(z)$, let $(x,y)$ be the players' sequence-form strategies. \cite{von1996efficient} showed that
the expected payoffs of the players can be written concisely as $ U_1(x,y) = x^\top A y$, and Nash equilibria are saddle points of the function $U_1=-U_2$, i.e., solutions to $\max_x\min_y(x^\top A y)=\min_y\max_x(x^\top A y)$, where $A$ is the payoff matrix.

We give additional preliminaries on normal-form games with stochastic action sets (GSAS) and counterfactual regret minimization (CFR) in~\cref{appsec:gsasprelims} and~\cref{appsec:cfrprelims} respectively. We also provide a notation table for important symbols in~\cref{appsecs:table}. Finally, throughout the paper we work in the unit-cost real-arithmetic model, where exact arithmetic operations take constant time.


\section{EFGs with Stochastic Action Sets}\label{sec:efg_repr}
In this section, we define and analyze properties of extensive-form games with stochastic action sets (EFGSAS). Proofs of theoretical results in this section are deferred to~\cref{appsecs:efgproofs4}.  


\begin{definition}[EFGSAS]\label{def:efgsas}
	Given a base extensive-form game $\mathcal{G}^\mathrm{orig} = (\calH,A,\calN,\calZ,u,\calI)$, let $\mathcal{S}_{I,i} \subseteq 2^{ A_{I,i}} \backslash \{ \varnothing \}$ for all infosets $I\in\mathcal{I}_i$ belonging to player $i$. 
	Let $\rho_{I,i} \in \Delta(\mathcal{S}_{I,i})$ be a distribution over elements $\mathcal{S}_{I,i}$, such that $\rho_{I,i}(S_{I,i})$ gives the probability that action subset $S_{I,i}$ is observed in each $I\in\mathcal{I}_i$, for each $S_{I,i} \in \mathcal{S}_{I,i}$. Let $\calS\coloneqq\prod_{i\in\mathcal N}\prod_{I\in\mathcal I_i}\mathcal S_{I,i}$ and $\rho\in\Delta(\calS)$ denote the ensemble of action subsets and distributions over all infosets and players.
	An \textbf{EFG} with \textbf{S}tochastic \textbf{A}ction \textbf{S}ets (\textit{EFGSAS}) is given by the tuple $\mathcal{G} = (\mathcal{G}^\mathrm{orig}, \mathcal{S}, \rho)$.
\end{definition}
We additionally consider two-player zero-sum EFGSAS (\textit{2p0s-EFGSAS}), where \(\calG^\mathrm{orig}\) is 2p0s, i.e., \(n=2\) and \(\forall z \in \calZ , u_1(z) = -u_2(z)\).
Going forward, we make the following assumption that the availability sets at all infosets are mutually independent, both within and across players.
This assumption enables our strategy representation reduction results. Moreover, relaxing the assumption implies that a polynomial-bit representation of approximate equilibria might not exist, even in normal-form GSAS (see~\cref{appsec:bitcompact}
for details).
\begin{assumption}
	There exists $\rho(S) = \prod_{i\in\calN} \prod_{I\in\mathcal{I}_i} \rho_{I,i}(S_{I,i})$ for every $S\in\calS$, 
    i.e., action availabilities in the game are independent across all infosets and players. Moreover, the payoff $u_i(z)$ for reaching $z\in\calZ$ does not depend on the realized availability profile~$S$.
	\label{assumption:productefg}
\end{assumption}

\subsection{Information Disclosure and Strategy Representation}\label{sec:infodisclosure}

In standard EFG literature, equilibrium computation typically relies on \emph{ex-ante} strategy formulations, meaning that players fix a strategy before the game begins. This requires all information about the game (e.g., infosets and game structure) to be revealed upfront.
However, EFGSAS belong to a class of 
games where information disclosure (and in particular, action availability disclosure) can be done \emph{sequentially}.
Two natural regimes of action disclosures arise: (i) the \emph{ex-interim} regime, where players only observe their available actions upon reaching an infoset, and (ii) the \emph{ex-ante} regime, where players can observe all available actions at each infoset $I\in\calI_i$ before game-play. This introduces new challenges in terms of strategy formalisms in EFGSAS, 
since players need to take their action availabilities into account when selecting strategies. 
\begin{definition}[Ex-interim EFGSAS]
	In an ex-interim EFGSAS, the availability set \(S_{I,i} \in \calS_{I,i}\) for an infoset \(I \in \calI_i\) is sampled according to \(\rho_{I,i}\) and privately observed by player \(i\) only when infoset \(I\) is reached by that player.
\end{definition}

Ex-interim EFGSAS capture some interesting settings, such as (i) board games such as Backgammon, Dice Chess, and the ancient Mesopotamian game ``The Royal Game of Ur'', in which at the start of each turn, players roll dice to determine how pieces can be moved, and
(ii) pursuit-evasion games where Nature restricts actions randomly over time due to, for example, inclement weather, resulting in players only observing their available actions at every timestep. In the remainder of this paper, we primarily focus on the ex-interim case for notational clarity. However, we remark that generalizations of our results in Sections~\ref{sec:infodisclosure} and~\ref{sec:dagplex} can be obtained for the ex-ante regime and for \emph{intermediate} regimes where a player might learn about their action availabilities $k$ steps in advance. We give an overview of these generalizations in~\cref{rem:k-interim} and formal statements in~\cref{appsec:kexinterim}. 



For an ex-interim EFGSAS \(\calG\) meeting~\cref{assumption:productefg}, a na\"ive expansion procedure can be applied to obtain a strategically equivalent expanded EFG \(\calG^\dagger\) with infosets \(\calI^\dagger\). Informally, before decision nodes, the procedure inserts chance nodes encoding the distribution over all action subsets $\calS_{I,i}$, with corresponding probabilities following $\rho_{I,i}\in\Delta(\calS_{I,i})$, for each infoset $I\in\calI_i$. A detailed specification of this procedure is given in~\cref{appsubsec:exinterimexpansion}. 
We illustrate the expansion procedure using an example. 

\begin{example}[Running Example]
		\label{ex:updown}
		Consider an EFGSAS with base game $\calG^{\mathrm{orig}}$, which is based on a 2p0s game from~\cite{cerny2024layered}. $\calG^{\mathrm{orig}}$ has a \emph{defender} (player 1) and an \emph{attacker} (player 2), shown in~\cref{fig:updown} (Top). 
        The attacker receives a payoff of $-1$ if they ever choose the same action as the defender ($H$ or $T$), and otherwise receives 1. It is easy to see that the unique NE in \(\calG^{\mathrm{orig}}\) for the defender is $x^*_1(H) = x^*_1(T) = 0.5$, while the attacker plays $x^*_2(H_1H_2) = x^*_2(T_1T_2) = 0.5$ and $x^*_2(H_1T_2) = x^*_2(T_1H_2) = 0$; there is  no reason for the attacker to play the dominated strategies $H_1T_2$ or $T_1H_2$.

        In the EFGSAS $\calG$, at infoset $A$ the attacker observes an availability set of \(\{H_1\}\) with probability \(\alpha\), and \(\{H_1,T_1\}\) otherwise. At infosets $B$ and $C$, the attacker observes \(\{H_2\}\) with probability \(\lambda\) and \(\{H_2,T_2\}\) otherwise. In Figure~\ref{fig:updown} (Bottom), we show the na\"ive expansion when $\alpha=0$, $\lambda > 0$.
        Notice that defining a strategy in this expanded form requires a larger representation. Moreover, we can intuitively see that EFGSAS may have significantly different solutions to the base game. For example, consider the degenerate case where $\alpha=\lambda=1$. Then, since the attacker is forced to always play $H_1$ and $H_2$, the defender's optimal strategy is $H$.
\end{example}

	\begin{figure}[t]
	\centering
		\begin{minipage}{0.35\linewidth}
			\centering
            \resizebox{\linewidth}{!}{
			\begin{tikzpicture}[x=0.3cm, y=1cm, font=\scriptsize]
            \tikzstyle{efg root}   =[circle, fill=red!85!black, draw=red!55!black, minimum size=3.4mm, inner sep=0pt]
\tikzstyle{efg node}   =[circle, draw=blue, line width=.8pt, fill=white, minimum size=3.6mm, inner sep=0pt]
\tikzstyle{efg chance} =[diamond, draw=green!55!black, line width=.8pt, fill=white, minimum size=3.8mm, inner sep=0pt]
\tikzstyle{efg edge}   =[line width=.7pt]
\tikzstyle{efg lab}    =[midway, fill=white, inner sep=1pt]
\tikzstyle{efg alpha}  =[pos=0.26, fill=white, inner sep=1pt, text=green!45!black]
\tikzstyle{efg iset}   =[densely dotted, line width=.9pt]
				\def\hw{1}
				\node[efg root] (R) at (6,0) {};
				\node[efg node] (BH) at (2,-1.0) {$A$};
				\node[efg node] (BT) at (10,-1.0) {$A$};
				\node[efg node] (D1) at (0,-2.2) {$B$};   
				\node[efg node] (D2) at (4,-2.2) {$C$};   
				\node[efg node] (D3) at (8,-2.2) {$B$};   
				\node[efg node] (D4) at (12,-2.2) {$C$};   
                \foreach \d/\i/\ui/\j/\uj in {1/1/-1/2/-1, 2/3/-1/4/+1, 3/5/+1/6/-1, 4/7/-1/8/-1}{
					\node (t\i) at ($(D\d)+(-\hw,-1.2)$) {$\ui$};
					\node (t\j) at ($(D\d)+( \hw,-1.2)$) {$\uj$};}
				\draw[efg edge] (R)--(BH) node[efg lab]{$H$};   \draw[efg edge] (R)--(BT) node[efg lab]{$T$};
				\draw[efg edge] (BH)--(D1) node[efg lab]{$H_1$}; \draw[efg edge] (BH)--(D2) node[efg lab]{$T_1$};
				\draw[efg edge] (BT)--(D3) node[efg lab]{$H_1$}; \draw[efg edge] (BT)--(D4) node[efg lab]{$T_1$};
				\foreach \d/\h/\k in {1/1/2, 2/3/4, 3/5/6, 4/7/8}{
					\draw[efg edge] (D\d)--(t\h) node[efg lab]{$H_2$};
					\draw[efg edge] (D\d)--(t\k) node[efg lab]{$T_2$};}
				\draw[efg iset] (BH)--(BT);
				\draw[efg iset] (D1) to[bend left=26] (D3);   
				\draw[efg iset] (D2) to[bend left=26] (D4);   
			\end{tikzpicture}
			}
	\end{minipage}
 \hspace{3em}
	\begin{minipage}{.48\columnwidth}
		\centering
\resizebox{\linewidth}{!}{
			\begin{tikzpicture}[x=0.55cm, y=1cm, font=\scriptsize]
				\def\hw{0.6}   
                \tikzstyle{efg root}   =[circle, fill=red!85!black, draw=red!55!black, minimum size=3.4mm, inner sep=0pt]
\tikzstyle{efg node}   =[circle, draw=blue, line width=.8pt, fill=white, minimum size=3.6mm, inner sep=0pt]
\tikzstyle{efg chance} =[diamond, draw=green!55!black, line width=.8pt, fill=white, minimum size=3.8mm, inner sep=0pt]
\tikzstyle{efg edge}   =[line width=.7pt]
\tikzstyle{efg lab}    =[midway, fill=white, inner sep=1pt]
\tikzstyle{efg alpha}  =[pos=0.26, fill=white, inner sep=1pt, text=green!45!black]
\tikzstyle{efg iset}   =[densely dotted, line width=.9pt]
\tikzstyle{prob lab left}   =[midway, left]
\tikzstyle{prob lab right}   =[midway, right]
                
				\node[efg root] (R) at (5.25,0) {};                       
				\node[efg node] (BH) at (2.25,-0.8) {};                   
				\node[efg node] (BT) at (8.25,-0.8) {};                   
				\foreach \i/\x in {1/0.75, 2/3.75, 3/6.75, 4/9.75}        
				\node[efg chance] (C\i) at (\x,-1.6) {};
				\foreach \i/\x in {1/0, 2/1.5, 3/3, 4/4.5, 5/6, 6/7.5, 7/9, 8/10.5}  
				\node[efg node] (D\i) at (\x,-2.6) {};
				
				\foreach \d/\i/\u in {1/1/-1, 3/4/-1, 5/7/+1, 7/10/-1}    
				\node (t\i) at ($(D\d)+(0,-1.2)$) {$\u$};
				\foreach \d/\i/\ui/\j/\uj in {2/2/-1/3/-1, 4/5/-1/6/+1, 6/8/+1/9/-1, 8/11/-1/12/-1}{ 
					\node (t\i) at ($(D\d)+(-\hw,-1.2)$) {$\ui$};
					\node (t\j) at ($(D\d)+( \hw,-1.2)$) {$\uj$};}
				
				\draw[efg edge] (R)--(BH)  node[efg lab]{$H$};    \draw[efg edge] (R)--(BT)  node[efg lab]{$T$};
				\draw[efg edge] (BH)--(C1) node[efg lab]{$H_1$};  \draw[efg edge] (BH)--(C2) node[efg lab]{$T_1$};
				\draw[efg edge] (BT)--(C3) node[efg lab]{$H_1$};  \draw[efg edge] (BT)--(C4) node[efg lab]{$T_1$};
				\foreach \c/\l/\r in {1/1/2, 2/3/4, 3/5/6, 4/7/8}{        
					\draw[efg edge] (C\c)--(D\l) node[prob lab left ]{$\lambda$};
					\draw[efg edge] (C\c)--(D\r) node[prob lab right]{$1{-}\lambda$};}
				\foreach \d/\t in {1/1, 3/4, 5/7, 7/10}                   
				\draw[efg edge] (D\d)--(t\t) node[efg lab]{$H_2$};
				\foreach \d/\h/\k in {2/2/3, 4/5/6, 6/8/9, 8/11/12}{      
					\draw[efg edge] (D\d)--(t\h) node[efg lab]{$H_2$};
					\draw[efg edge] (D\d)--(t\k) node[efg lab]{$T_2$};}
				
				\draw[efg iset] (BH)--(BT);                       
				\draw[efg iset] (D1) to[bend left=18]  (D5);
				\draw[efg iset] (D3) to[bend left=18]  (D7);
				\draw[efg iset] (D2) to[bend right=18] (D6);
				\draw[efg iset] (D4) to[bend right=18] (D8);
			\end{tikzpicture}
            }
		\label{fig:efgsasgame1}
	\end{minipage}
\vspace{-1em}
	\caption{(Left) The game tree of $\calG^{\mathrm{orig}}$. (Right) The expanded game $\calG^\dagger$ for \(\alpha=0\), i.e., at infoset $A$, actions \(\{H_1,T_1\}\) are always available.
    Circular nodes $(\fmycircle, \mycircle)$
        denote decision points for defender (resp. attacker). Diamonds $(\mydiamond)$ denote chance nodes and dotted lines connect nodes belonging to the same infoset. 
        At infosets \(B\) and \(C\), the attacker observes availability set \(\{H_2\}\) with probability \(\lambda\) and \(\{H_2,T_2\}\) otherwise. }
    \label{fig:updown}
\end{figure}

Beyond ex-interim action disclosures, an EFGSAS may have \emph{ex-ante} information disclosure, with the availability sets for \emph{every} infoset \(I \in \calI_i\) sampled and revealed to player \(i\) prior to game-play. Such an EFGSAS can be expanded with a similar procedure, with a chance node added at the root encoding the distribution among \textit{all} possibly observed \(S_{I,i}\) across all infosets. 

Rather than being a purely conceptual separation, it turns out that strategies in the same EFGSAS \(\calG\) can differ significantly between the ex-interim and ex-ante settings. In the game in~\cref{ex:updown}, it is possible in the ex-ante setting for the attacker to avoid ever playing \(T_1H_2\), a strictly dominated strategy in \(\calG^{\mathrm{orig}}\). However, in the ex-interim setting, any attacker strategy that plays \(T_1\) with strictly positive probability might be forced to play \(T_1H_2\), depending on the sampled availability set.

For a $\calG^\dagger$ associated with EFGSAS $\calG$ in both ex-interim and ex-ante regimes, standard definitions of mixed, behavioral and sequence-form strategies hold (cf.~\cref{sec:prelim}). However, the \emph{representation size} of these strategies can grow exponentially in the size of \(\calG\), since all possible observed action subsets need to be encoded (see~\cref{appsubsec:expansion_size} for details). Hence, going forward, we seek to develop techniques that can reduce representation size in EFGSAS.

\subsection{Strategy Representation over DAG-plexes}
\label{sec:dagplex}

In the na\"ive expansion of an ex-interim EFGSAS, behavioral strategies need to be defined at each new infoset generated by the availability sets. This requires strategies to account for the history of availability sets that have been observed. We now show that for any behavioral strategy on \(\calG^{\dagger}\), there exists a behavioral strategy with equal payoff that considers only the observed action availabilities at the \emph{current} infoset, and not prior ones.
We define \(\mathrm{obs} : \calI_i^\dagger \to \bigcup_{I \in \calI_i} \calS_{I,i}\) such that \(\mathrm{obs}(I^\dagger)\) is the availability set \(S_{I,i}\) which \(I^\dagger\) corresponds to observing in the EFGSAS.

 
\begin{proposition}
    \label{prop:decisionmerging}
	For any behavioral strategy \(\beta_i \) in \(\calG^\dagger\),
	there exists some strategy \(b_i: \bigcup_{I\in\calI_i}(\{I\}\times\calS_{I,i}) \to \bigcup_{I \in \calI_i} \Delta(A_I)\), where we suppress the first argument when clear,
	such that \(\mathrm{supp}(b_i(S_{I,i})) \subseteq S_{I,i}\)
	and the strategy \(\hat{\beta_i}(I^\dagger) = b_i(\mathrm{obs}(I^\dagger))\) has
    \(U_j(\beta_i, \beta_{-i})=U_j(\hat{\beta_i}, \beta_{-i})\) for every $j\in\calN$ and every \(\beta_{-i}\).
\end{proposition}

\cref{prop:decisionmerging} allows us to consider strategies over the space where all
infosets \(I^\dagger\) that correspond to observing the same \(S_{I,i}\) are merged, shrinking the decision problem from a potentially exponentially large tree (cf.~\cref{appsubsec:expansion_size}) to a smaller, though possibly still exponential, directed acyclic graph (DAG).
Every node in the DAG uniquely corresponds to an observed availability set \(S_{I,i}\) in \(\calG\), which is denoted \(S^\dagger_{I,i}\). We abuse notation to treat the node \(S^\dagger_{I,i}\) as the corresponding availability set \(S_{I,i}\) in \(\calG\) and vice versa, denoting a particular node in the DAG as \(S^\dagger_{I,i} \in \calS_{I,i}\) or \(b_i(S_{I,i})\) as the strategy from \(b_i\) at \(S^\dagger_{I,i}\).

A strategy \(b_i\) as defined in~\cref{prop:decisionmerging} is henceforth referred to as an \textit{EFGSAS behavioral strategy}, and we let \(b_i(a\mid S^\dagger_{I,i})\) be the probability that \(b_i(S^\dagger_{I,i})\) assigns to action \(a\). 
Player \(i\)'s expected payoff from a joint behavioral strategy profile \(b = (b_i)_{i \in \calN}\) is given by the payoff of \(\hat{\beta}\) in the expanded EFG and is equal to
$U_i(b) = \mathbb{E}_{S \sim \rho}\left[\sum_{z \in \calZ} \mathbb{P}(z \mid b(S),r) u_i(z) \right]$,
where \(\mathbb{P}(z \mid b(S), r)\) is the probability that leaf \(z \in \calZ\) is reached from root \(r\) following \(b\) under the availability set ensemble \(S\).
In line with standard definitions, we call \(b^*\) an \textit{\(\epsilon\)-Nash equilibrium} (\textit{$\epsilon$-NE}) if no player can unilaterally deviate to increase their expected payoff by more than \(\epsilon\), i.e., \(\forall i \in \calN\text{ and } \forall b_i\), \(U_i(b^*) \geq U_i(b_i, b^*_{-i}) - \epsilon\).


The payoff equivalence established in~\cref{prop:decisionmerging} allows us to encode EFGSAS behavioral strategies using ``realization vectors'' on the merged DAG. These vectors belong to smaller, compressed decision polytopes (compared to the standard treeplex), which we call \emph{DAG-plexes}. In these polytopes, the \emph{branching} operation in treeplexes~\citep{hoda2010smoothing} is replaced with a \emph{recombination} operation.


\begin{definition}[DAG-plex]
	The class of DAG-plexes is defined recursively as follows:
	\begin{enumerate}
		\item \emph{Simplices}: Every simplex \(\Delta_m \coloneqq \{x \in [0,1]^m \mid \sum_{i=1}^m x_i = 1\}\) is a DAG-plex.
		\item \emph{Cartesian product}: If \(Q_1,\ldots, Q_k\) are DAG-plexes, then \(Q_1 \times \ldots \times Q_k\) is a DAG-plex.
		\item \emph{Recombination}: If \(P \subseteq\mathbb{R}_{\ge0}^p\) and \(Q \subseteq\mathbb{R}_{\ge0}^q\) are DAG-plexes and \(j \subseteq \{1,...,p\}, j\neq \varnothing\), then \(\{(x,y)\in \mathbb{R}^{p+q} \mid x \in P, y \in Q \cdot\sum_{i \in j} x_i\}\) is a DAG-plex.
	\end{enumerate}
\end{definition}


\begin{figure*}[ht]
\begin{minipage}{0.55\textwidth}
\centering
\begin{tikzpicture}[x=0.55cm, y=1cm, font=\scriptsize]
    \def\hw{0.7}   

    \tikzstyle{efg root}   =[circle, fill=red!85!black, draw=red!55!black, minimum size=3.4mm, inner sep=0pt]
    \tikzstyle{efg node}   =[circle, draw=blue, line width=.8pt, fill=white, minimum size=3.6mm, inner sep=0pt]
    \tikzstyle{efg chance} =[diamond, draw=green!55!black, line width=.8pt, fill=white, minimum size=3.8mm, inner sep=0pt]
    \tikzstyle{efg edge}   =[line width=.7pt]
    \tikzstyle{efg lab}    =[midway, fill=white, inner sep=1pt]
    \tikzstyle{efg alpha}  =[pos=0.26, fill=white, inner sep=1pt, text=green!45!black]
    \tikzstyle{efg iset}   =[densely dotted, line width=.9pt]
    \tikzstyle{prob lab left}   =[midway, left]
    \tikzstyle{prob lab right}  =[midway, right]

    \node[efg root] (R) at (5.25,0) {};

    \node[efg chance] (C1) at (2.25,-0.8) {};
    \node[efg chance] (C2) at (8.25,-0.8) {};

    \node[efg node] (F1) at (0.75,-1.6) {};   
    \node[efg node] (F2) at (3.75,-1.6) {};   
    \node[efg node] (F3) at (6.75,-1.6) {};   
    \node[efg node] (F4) at (9.75,-1.6) {};   

    \foreach \i/\x in {1/0.75, 2/2.75, 3/4.75, 4/6.75, 5/8.75, 6/10.75}
        \node[efg node] (D\i) at (\x,-2.6) {};

    \foreach \d/\i/\ui/\j/\uj in {
        1/1/-1/2/-1,
        2/3/-1/4/-1,
        3/5/-1/6/+1,
        4/7/+1/8/-1,
        5/9/+1/10/-1,
        6/11/-1/12/-1
    }{
        \node (t\i) at ($(D\d)+(-\hw,-1.2)$) {$\ui$};
        \node (t\j) at ($(D\d)+( \hw,-1.2)$) {$\uj$};
    }

    \draw[efg edge] (R)--(C1) node[efg lab]{$H$};
    \draw[efg edge] (R)--(C2) node[efg lab]{$T$};

    \draw[efg edge] (C1)--(F1) node[prob lab left ]{$\alpha$};
    \draw[efg edge] (C1)--(F2) node[prob lab right]{$1{-}\alpha$};
    \draw[efg edge] (C2)--(F3) node[prob lab left ]{$\alpha$};
    \draw[efg edge] (C2)--(F4) node[prob lab right]{$1{-}\alpha$};

    \draw[efg edge] (F1)--(D1) node[efg lab]{$H_1$};
    \draw[efg edge] (F2)--(D2) node[efg lab]{$H_1$};
    \draw[efg edge] (F2)--(D3) node[efg lab]{$T_1$};
    \draw[efg edge] (F3)--(D4) node[efg lab]{$H_1$};
    \draw[efg edge] (F4)--(D5) node[efg lab]{$H_1$};
    \draw[efg edge] (F4)--(D6) node[efg lab]{$T_1$};

    \foreach \d/\h/\k in {1/1/2, 2/3/4, 3/5/6, 4/7/8, 5/9/10, 6/11/12}{
        \draw[efg edge] (D\d)--(t\h) node[efg lab]{$H_2$};
        \draw[efg edge] (D\d)--(t\k) node[efg lab]{$T_2$};
    }

    \draw[efg iset] (F1) to[bend left=18]  (F3);   
    \draw[efg iset] (F2) to[bend right=18] (F4);   

    \draw[efg iset] (D1) to[bend left=18]  (D4);   
    \draw[efg iset] (D2) to[bend left=18]  (D5);   
    \draw[efg iset] (D3) to[bend right=18] (D6);   
\end{tikzpicture}
\end{minipage}
\begin{minipage}{0.21\textwidth}
\centering
\begin{tikzpicture}[x=0.4cm, y=0.975cm, font=\scriptsize]

    \tikzstyle{tp empty}  =[rectangle, draw=black, fill=white,
        line width=.6pt, minimum size=2.1mm, inner sep=0pt]
    \tikzstyle{tp filled} =[rectangle, draw=black, fill=black,
        line width=.6pt, minimum size=2.1mm, inner sep=0pt]
    \tikzstyle{tp edge}   =[line width=.6pt]
    \tikzstyle{tp dashed} =[line width=.6pt, dashed]
    \tikzstyle{tp lab}    =[midway, fill=white, inner sep=0.8pt]

    \coordinate (Top) at (0,0.9);
    \node[tp filled]  (R) at (0,0) {};

    \node[tp empty] (L) at (-1.5,-0.9) {};
    \node[tp empty] (M) at ( 1.5,-0.9) {};

    \node[tp empty] (A) at (-1.5,-1.9) {};
    \node[tp empty] (B) at ( 0.5,-1.9) {};
    \node[tp empty] (C) at ( 2.5,-1.9) {};

    \coordinate (A1) at (-2.25,-2.85);
    \coordinate (A2) at (-0.75,-2.85);

    \coordinate (B1) at (-0.25,-2.85);
    \coordinate (B2) at ( 1.25,-2.85);

    \coordinate (C1) at ( 1.75,-2.85);
    \coordinate (C2) at ( 3.25,-2.85);

    \draw[tp edge] (Top) -- (R) node[tp lab] {$\varnothing$};

    \draw[tp dashed] (R) -- (L) node[tp lab, above left] {$\{H_1\}$};
    \draw[tp dashed] (R) -- (M) node[tp lab, above right] {$\{H_1,T_1\}$};

    \draw[tp edge] (L) -- (A) node[tp lab] {$H_1$};

    \draw[tp edge] (M) -- (B) node[tp lab] {$H_1$};
    \draw[tp edge] (M) -- (C) node[tp lab] {$T_1$};

    \draw[tp edge] (A) -- (A1) node[tp lab] {$H_2$};
    \draw[tp edge] (A) -- (A2) node[tp lab] {$T_2$};

    \draw[tp edge] (B) -- (B1) node[tp lab] {$H_2$};
    \draw[tp edge] (B) -- (B2) node[tp lab] {$T_2$};

    \draw[tp edge] (C) -- (C1) node[tp lab] {$H_2$};
    \draw[tp edge] (C) -- (C2) node[tp lab] {$T_2$};
\end{tikzpicture}
\end{minipage}
\begin{minipage}{0.21\textwidth}
\centering
\begin{tikzpicture}[x=0.4cm, y=0.975cm, font=\scriptsize]

    \tikzstyle{tp empty}  =[rectangle, draw=black, fill=white,
        line width=.6pt, minimum size=2.1mm, inner sep=0pt]
    \tikzstyle{tp filled} =[rectangle, draw=black, fill=black,
        line width=.6pt, minimum size=2.1mm, inner sep=0pt]
    \tikzstyle{tp edge}   =[line width=.6pt]
    \tikzstyle{tp dashed} =[line width=.6pt, dashed]
    \tikzstyle{tp lab}    =[midway, fill=white, inner sep=0.8pt]

    \coordinate (Top) at (0,0.9);
    \coordinate (R)   at (0,0);
    \node[tp filled]  (R) at (0,0) {};

    \node[tp empty] (L) at (-1.5,-0.9) {};
    \node[tp empty] (M) at ( 1.5,-0.9) {};

    \node[tp empty] (A) at ( 0.0,-1.9) {};
    \node[tp empty] (B) at ( 2.5,-1.9) {};

    \coordinate (A1) at (-0.85,-2.85);
    \coordinate (A2) at ( 0.85,-2.85);

    \coordinate (B1) at ( 1.65,-2.85);
    \coordinate (B2) at ( 3.35,-2.85);

    \draw[tp edge] (Top) -- (R) node[tp lab] {$\varnothing$};

    \draw[tp dashed] (R) -- (L) node[tp lab, above left] {$\{H_1\}$};
    \draw[tp dashed] (R) -- (M) node[tp lab, above right] {$\{H_1,T_1\}$};

    \draw[tp edge] (L) -- (A) node[tp lab] {$H_1$};

    \draw[tp edge] (M) -- (A) node[tp lab] {$H_1$};
    \draw[tp edge] (M) -- (B) node[tp lab] {$T_1$};

    \draw[tp edge] (A) -- (A1) node[tp lab] {$H_2$};
    \draw[tp edge] (A) -- (A2) node[tp lab] {$T_2$};

    \draw[tp edge] (B) -- (B1) node[tp lab] {$H_2$};
    \draw[tp edge] (B) -- (B2) node[tp lab] {$T_2$};
\end{tikzpicture}
\end{minipage}
\caption{(Left) The expanded game tree for Example~\ref{ex:updown-v2}: refer to Figure~\ref{fig:updown} for the legend.  (Middle) The corresponding treeplex for the attacker. Empty vertices ($\square$) are infosets. Filled vertices ($\blacksquare$) are observation nodes that lead to parallel infosets. Solid edges are actions that can be taken. Dotted edges are possible observations. In this case, the only observations are the action availabilities, though in general these could instead be related to $\calG^{\mathrm{orig}}$.
The empty sequence is labeled $\varnothing$.
(Right) The DAG-plex representation, with nodes and edges defined the same as in the treeplex.} 
\label{fig:updown-v2-dagplex}
\end{figure*}
\begin{example}[DAG-plex for~\cref{ex:updown}]
    Consider an instance of the game in Example~\ref{ex:updown} where $\alpha > 0$ and $\lambda=0$, i.e., $H_2$ and $T_2$ are always available, but $\{H_1\}$ is available with probability $\alpha$ and $\{H_1,T_1\}$ is available with probability $1-\alpha$. The expanded game tree, together with the corresponding treeplex and DAG-plex, is shown in Figure~\ref{fig:updown-v2-dagplex}.  
    In the treeplex, there are two decision points after taking action $H_1$, corresponding to action sets $\{ H_1 \}$ and $\{ H_1, T_1 \}$. In the DAG-plex, these decision points are collapsed to one, 
    resulting in a more compact structure.
    \label{ex:updown-v2}
\end{example}

Next, we define DAG realization vectors over the DAG-plex. Let
\(\Xi_i = \{(S^\dagger_{I,i},a) \mid I\in\calI_i,S^\dagger_{I,i}\in\calS_{I,i},a\in S^\dagger_{I,i}\}\cup\{\varnothing\}\)
be the coordinate index set for player \(i\)'s DAG-plex. For a node \(S^\dagger_{I,i}\), let \(\xi_i(S^\dagger_{I,i})\) denote the set of predecessor pairs \((S^\dagger_{I,i},a')\) immediately preceding that node along paths in player \(i\)'s DAG.

\begin{definition}[DAG realization vector]
A \emph{DAG realization vector} for player \(i\) is a non-negative vector \(\chi_i\) indexed by \(\Xi_i\), satisfying \(\chi_i(\varnothing)=1\) and, for every \(I\in\calI_i\) and \(S^\dagger_{I,i}\in\calS_{I,i}\), \(\sum_{a\in S^\dagger_{I,i}}\chi_i(S^\dagger_{I,i},a)=\sum_{\xi'\in\xi_i(S^\dagger_{I,i})}\chi_i(\xi')\).
Here, \(\xi_i(S^\dagger_{I,i})=\{\varnothing\}\) if there is no preceding action of player \(i\), and the coordinate \(\chi_i(S^\dagger_{I,i},a)\) is the sum of the products of player \(i\)'s action probabilities along DAG paths reaching \(S^\dagger_{I,i}\) and then taking action \(a\).
\end{definition}


A DAG realization vector \(\chi_i\) recovers an EFGSAS
behavioral strategy \(b_i\) given by 
\(
b_i(a \mid S^\dagger_{I,i})\coloneqq \chi_i(S^\dagger_{I,i},a)/
\sum_{a' \in S^\dagger_{I,i}}\chi_i(S^\dagger_{I,i},a')
\)
when the denominator is positive, and by any distribution on \(S^\dagger_{I,i}\) otherwise. Unavailable actions have probability zero.
Given a profile of DAG realization vectors \(\chi=(\chi_i)_{i\in\calN}\), let \(b=(b_i)_{i\in\calN}\) be the recovered behavioral profile.
We define \(U_i(\chi)\coloneqq U_i(b)\).
Together with~\cref{prop:decisionmerging}, this construction ensures that every strategy in the expanded EFG \(\calG^\dagger\) has a payoff-equivalent behavioral strategy represented by a DAG realization vector.
Thus, DAG realization vectors suffice to represent equilibrium behavior in \(\calG\). For equilibrium computation, we subsequently study compact representations over \emph{sequence-form strategies}
on the sequence set of the base game.


\begin{remark}[On $k$-ex-interim EFGSAS]\label{rem:k-interim}
    The definitions of ex-ante and ex-interim action disclosures can be viewed as extrema of a spectrum of information disclosure. In particular, due to the sequential structure of an EFGSAS, players might be able to observe their action availabilities $k$ of their own decision points ahead. 
    Informally, selecting $k = 0$ recovers the ex-interim action disclosure setting, while selecting $k$ at least as large as every player's decision depth recovers the ex-ante setting. We give formal definitions of what is referred to as the \textit{$k$-ex-interim} setting in~\cref{appsec:kexinterim}. It turns out that under Assumption~\ref{assumption:productefg}, the results in Sections~\ref{sec:infodisclosure} and~\ref{sec:dagplex} can be extended to the $k$-ex-interim setting.
    
    In particular, suppose $B_i$ bounds the branching between consecutive decision points of player $i$. Then, the upper bound on its size contains a factor which is exponential in $\sum_{r=1}^{k}B_i^r$. Furthermore, we show that DAG-plexing can also be performed in the $k$-ex-interim setting, allowing strategically equivalent expanded infosets to be merged (Proposition~\ref{prop:mergingkinterim}). Nevertheless, the upper bound on the size of the resulting DAG-plex still depends exponentially on $\sum_{r=1}^{k}B_i^r$ (Proposition~\ref{prop:dagplexsizek}), though this does not rule out alternative compact representations which are beyond the scope of this paper. Hence, the ex-interim setting is a crucial special case, where the dependence on the branching of the players' DAG-plexes can be sufficiently controlled. As such, our analysis in the rest of the paper focuses on ex-interim EFGSAS.
\end{remark}

\subsection{Compact Representation}
Thus far, we have established that in the ex-interim setting, expressing a player's decision problem over the DAG-plex indexed by \(\Xi_i\) instead of the na\"ively expanded sequences \(\Sigma_i^\dagger\) mitigates the exponential dependence of a player's decision problem on the depth of the decision tree. However, since each chance node in the expanded game enumerates an availability set, there remains a possibly exponential dependence on the \emph{number of actions} at each infoset.


In this section, we consider \textit{compact} strategy representations for players in EFGSAS. In particular, our goal is to study conditions under which Nash equilibria can be represented by a vector of size comparable to $|\Sigma_i^{{\mathrm{orig}}}|$, where $\Sigma_i^{{\mathrm{orig}}}$ denotes player $i$'s set of sequences in $\calG^{\mathrm{orig}}$. 
We show under~\cref{assumption:productefg} that for ex-interim EFGSAS, it suffices to consider the restricted space of ``implementable'' sequence-form strategies, which can be constructed using marginal behavioral strategies at each infoset. Then, we establish the existence of a compact vector of size $|\Sigma_i^{{\mathrm{orig}}}|$ that can approximately implement any NE strategy.

\paragraph{Implementable strategies.}
For an EFGSAS behavioral strategy \(b_i\), due to Assumption~\ref{assumption:productefg}, its corresponding \emph{marginal behavioral strategy} $\mu_{I,i}(a)$ for each $a\in A_{I,i}$ at an infoset \(I \in \calI_i\) is 
	\(\mu_{I,i}(a) \coloneqq 
	\mathbb{P}[a ; \rho_{I,i} , b_i]
	= \sum_{S_{I,i}\in\calS_{I,i}}\rho_{I,i}(S_{I,i}) b_i(a\mid S_{I,i}) \mathds{1}\{a\in S_{I,i}\}\).
Let \(\mu_{i}\) be the ensemble over \(I\in\calI_i\) of marginal behavioral strategies for player $i$.
Then, the expected payoff of the joint EFGSAS behavioral strategy \(b=(b_i)_{i \in \calN}\) can be written as
$
    U_i(b) = \sum_{z \in \calZ} u_i(z)p_c(z) \prod_{j \in \calN}
    \mathbb{P}[z \mid \mu_{j}, r]$,
where \(p_c(z)\) is the product of base-game chance probabilities along the path from \(r\) to \(z\), and \(\mathrm P[z\mid\mu_j,r]\) is the product of player \(j\)’s marginal action probabilities along that path.
Hence, at an infoset $I\in\calI_i$, it suffices to work in the space of $\mu_{I,i} \in \Delta(A_{I,i})$, rather than the much larger $b_i(S_{I,i})$. Note that multiple behavioral strategies can share the same marginal behavioral strategy $\mu$. These marginal behavioral strategies induce a sequence-form strategy in the base game \(\calG^{\mathrm{orig}}\), and we next characterize the ``implementable'' base-game sequence-form strategies which can arise in this way.

\begin{definition}[Implementable sequence-form strategy]
\label{def:implementablesfs}
Consider a sequence-form strategy \(\omega_i\) for player \(i\) in \(\calG^{\mathrm{orig}}\).
We call \(\omega_i\) an \emph{implementable sequence-form strategy} if there exists an EFGSAS behavioral strategy \(b_i\), with induced marginal behavioral strategies \(\mu_{I,i}\), such that \(\omega_i[(I,a)]=\omega_i[\sigma(I)]\mu_{I,i}(a)\) for every \(I\in\calI_i,\ a\in A_{I,i}\).
The set of implementable sequence-form strategies is denoted \(\Omega_i\) and is a subset of player \(i\)'s sequence-form strategy polytope in \(\calG^{\mathrm{orig}}\).
We say that a DAG realization vector \(\chi_i\) \emph{implements} \(\omega_i\) if the behavioral strategy recovered from \(\chi_i\) induces marginals satisfying the preceding equality.
\end{definition}


It turns out that under~\cref{assumption:productefg}, each player's \emph{effective} strategy space is the compact, convex set of implementable sequence-form strategies, and the expected payoffs are bilinear in two-player games. The following results ensure that (i) implementable sequence-form strategies suffice to capture all Nash equilibria of an EFGSAS, and (ii) a version of the minimax theorem holds for implementable strategies in 2p0s-EFGSAS.

\begin{proposition}
[NE equivalence]
\label{prop:NE-equiv}
    Consider a profile of DAG realization vectors \(\chi=(\chi_1,\ldots,\chi_n)\) that implements a sequence-form strategy profile \(\omega=(\omega_1,\ldots,\omega_n)\).
    The behavioral profile represented by \(\chi\) is an \(\epsilon\)-Nash equilibrium if and only if \(\forall i\in[n]\),
    $
		U_i(\omega) \geq \max_{\omega'_i\in \Omega_i} U_i(\omega'_i,\omega_{-i}) - \epsilon.
        $
\end{proposition}

\begin{proposition}
[Minimax theorem for EFGSAS]
\label{prop:minimaxEFGSAS}
    For a 2p0s-EFGSAS, let \((\chi_1^*,\chi_2^*)\) be DAG realization vectors implementing \((\omega_1^*,\omega_2^*)\). The behavioral profile represented by \((\chi_1^*,\chi_2^*)\) is a Nash equilibrium if and only if \((\omega_1^*,\omega_2^*)\) is a saddle point of \(U=U_1=-U_2\), i.e.,
	\(\omega_1^* \in \mathrm{argmax}_{\omega_1 \in \Omega_1} \min_{\omega_2 \in \Omega_2} U(\omega_1, \omega_2)\)
	and
	\(\omega_2^* \in \mathrm{argmin}_{\omega_2 \in \Omega_2}\) \(\max_{\omega_1 \in \Omega_1} U(\omega_1, \omega_2)\),
	where
	\begin{equation*}
	    \max_{\omega_1 \in \Omega_1} \min_{\omega_2 \in \Omega_2} U(\omega_1, \omega_2) =  \min_{\omega_2 \in \Omega_2} \max_{\omega_1 \in \Omega_1} U(\omega_1, \omega_2).
	\end{equation*}
\end{proposition}


Despite the above, it remains unclear how one can recover $\chi_i$ from $\omega_i$ efficiently.
To deal with this, we show that there exists a compact vector $W_i\in \mathbb{R}^{\lvert \Sigma_i^\mathrm{orig}\rvert}_{> 0}$ that \emph{approximately} implements each \(\omega_i \in \Omega_i\) and can be used to obtain valid EFGSAS strategies for any observed availability set. Let \(\chi_i(S_i)\) be the realization vector in \(\calG^{\mathrm{orig}}\) induced by the behavioral strategy associated with \(\chi_i\) under the action availability realization \(S_i=(S_{I,i})_{I\in\mathcal I_i}\).


\begin{theorem}\label{thm:efgsascompact}
	Let $\chi_i$ implement $\omega_i\in \Omega_i$. Then, $\forall \delta\in(0,1)$, there exists a vector \(W_i \in \mathbb{R}^{\lvert \Sigma_i^\mathrm{orig}\rvert}_{> 0}\) that implicitly represents a DAG realization vector \(\chi_i'\).
    Moreover, $\chi'_i(S_i)$ can be constructed in time linear in the size of \(\calG^{\mathrm{orig}}\), and $\chi'_i$ induces a marginal $\omega'_{i}$ satisfying $\|\omega_{i} - \omega'_{i}\|_1 < \delta$.
\end{theorem}


\begin{example}[Compact Representation]
    \label{example:lsg_compact}
    We focus on the special case of the running example (\cref{fig:updown}) with $\alpha=0$, $\lambda=0.5$, i.e., at infoset $A$, both $H_1$ and $T_1$ are always available. %
    We fix the indexing of the sequences as
    $[\varnothing, H, T]$ for the defender (player 1), and as $[\varnothing, H_1, T_1, H_1H_2, H_1T_2, T_1H_2, T_1T_2]$ for the attacker (player 2). By~\cref{thm:efgsascompact}, there exists a  compact vector which is of the size of sequence-form strategies in the base game. It can be shown that when $\lambda=0.5$, for any $\epsilon\in(0,1)$,
    $W^*_1 = [1, 2/3, 1/3]$ and $W^*_2 = [1, 1/3, 2/3, (1-\epsilon)/3, \epsilon/3, 2\epsilon/3, 2(1-\epsilon)/3]$ 
    give a compact representation of an $\epsilon$-NE of the game (see~\cref{appsec:examplecompactderiv} for the derivation).
Utilizing the procedure given in~\cref{alg:renorm}, we can obtain the appropriate conditional realization vector given any possible action availabilities. For instance, if for an availability realization \(S_2\), player 2 observes availability sets $\{H_2\}$ at infoset $B$ and  $\{H_2, T_2\}$ at infoset $C$,~\cref{alg:renorm} 
outputs 
$\chi_2(S_2) = [1, 1/3, 2/3, 1/3, 0, 2\epsilon/3, 2(1-\epsilon)/3]$, which is the conditional realization vector associated with \(W_2^*\) for these availabilities since player 2 cannot play $T_2$ at $B$. 

\end{example}

\section{Computing Compact Equilibria in EFGSAS}\label{sec:efg_compute}
In this section, we focus on the problem of NE computation in ex-interim 2p0s-EFGSAS.
From the perspective of the game solver, if $\rho$ were known explicitly, one could construct the expanded game and apply off-the-shelf solvers. However, explicitly tabulating $\rho$ can be impractical since an infoset could have exponentially many availability sets.
In light of this, we focus on the more challenging regime where the solver only has \emph{sample access} to $\rho$. 
More precisely, we assume sample access to the action availability distributions and payoff feedback only for currently available actions at each original infoset. This feedback is obtained through counterfactual evaluations, including at off-path infosets, rather than solely from a realized play trajectory. Our method consists of two stages: (i) we apply a regret minimizer that achieves sublinear ranking regret over the game tree, producing $\epsilon$-NE time-averaged marginals, then (ii) we compute via stochastic approximation a compact weight vector which approximately implements these marginals and specifies how each player should play at each infoset.
Proofs of theoretical results in this section are deferred to~\cref{appsecs:efgproofs5}.



\subsection{Ranking Regret Minimization}

Going forward, we focus on the online learning in games setting, where at each timestep \(t\leq T\), each player \(i\in\calN\) selects a strategy \(x_i^t\in\calX_i\) from a compact, convex strategy set \(\calX_i\). Let \(x^t=(x_i^t)_{i\in\calN}\). The standard performance metric in this setting is (cumulative) external regret, defined for player \(i\) over \(\calX_i\) as
\begin{equation}
R^{\mathsf{EXT}}_{T,i}(\calX_i)\coloneqq\max_{x_i'\in\calX_i}\sum_{t=1}^T\left(U_i(x_i',x_{-i}^t)-U_i(x^t)\right).
\end{equation}

Intuitively, low external regret means that the player is not outperformed by any single fixed strategy in hindsight. The folk result of no-regret learning in 2p0s-games states that if both players incur sublinear external regret (i.e., $R^{\mathsf{EXT}}_{T,i} = o(T)$), then the time average over strategies played is an approximate Nash equilibrium.

In EFGs, the sets \(\calX_i\) are typically the sequence-form strategy sets. \cite{zinkevich2007regret} showed that external regret over the game tree can be upper-bounded by individual per-infoset counterfactual regrets, leading to the \emph{counterfactual regret minimization} (CFR) framework for computing Nash equilibria in 2p0s-EFGs under self-play.

In a 2p0s-EFGSAS, however, comparison with a fixed action is unsuitable because that action may be unavailable at a given infoset. In this setting, the relevant comparator at infoset \(I\in\calI_i\) is instead the set \(M_{I,i}\subseteq\Delta(A_{I,i})\) of implementable marginal behavioral strategies. \(M_{I,i}\) contains precisely the marginals \(\mu_{I,i}\) induced by some choice of \(b_i(\cdot\mid S)\in\Delta(S)\) for every \(S\in\calS_{I,i}\). This motivates the study of techniques that can minimize external regret over the set of implementable sequence-form strategies $\Omega_i$.



In light of the above, we utilize the notion of \emph{ranking regret}, introduced in the \emph{sleeping bandits} literature~\citep{kanade2009sleeping}. A \textit{ranking} \(\kappa\) of \(A_{I,i}\) is a strict total ordering of its actions, and \(\kappa(S)\) denotes the highest-ranked action in \(S\). Let \(\mu_{I,i}^{\kappa}\in M_{I,i}\) be the marginal induced by this ranking, i.e.,
$\mu_{I,i}^{\kappa}(a)\coloneqq\sum_{S\in\calS_{I,i},\kappa(S)=a}\rho_{I,i}(S)$.

For each timestep \(t\), let \(u^t_{I,i}\coloneqq(u^t_{I,i}(a))_{a\in A_{I,i}}\) denote the expected \textit{counterfactual} payoff vector at \(I\), where \(u_{I,i}^t(a)\) is the expected counterfactual payoff obtained by playing action \(a\) and otherwise following their strategy at timestep \(t\). This payoff includes the reach probabilities of chance and the other player, but excludes player \(i\)'s own probability of reaching \(I\). We first give a per-infoset definition of ranking regret.

\begin{definition}[Ranking Regret]\label{def:rankingregret}
At an infoset \(I\), let \(\mu_{I,i}^t\in M_{I,i}\) be the marginal played at timestep \(t\). Then, over all rankings \(\kappa\) of \(A_{I,i}\), the ranking regret for player $i$ in $T$ timesteps is
\begin{equation}
R^{\mathsf{RANK}}_{T,I,i}\coloneqq\max_{\kappa}\sum_{t=1}^T\left\langle u_{I,i}^t,\mu_{I,i}^{\kappa}-\mu_{I,i}^t\right\rangle.
\end{equation}
\end{definition}

Intuitively, ranking regret compares the learner with the best fixed ranking in hindsight. Over \(M_{I,i}\), it coincides with external regret $R^{\mathsf{EXT}}_{T,I,i}$, as is shown in~\cref{lem:rankingexternal}. This allows us to bound $R^{\mathsf{EXT}}_{T,i}$ over the set of implementable sequence-form strategies $\Omega_i$ via ranking regrets, as shown in~\cref{cor:rankingcfr}.

We minimize these local ranking regrets using CAT-FTPL, a follow-the-perturbed-leader method for stochastic action availability introduced in~\cite{schwarz2026computing}. CAT-FTPL requires counterfactual payoffs only for currently available actions and uses independent auxiliary samples from \(\rho_{I,i}\) to construct its loss estimates. 
Running one CAT-FTPL instance at every infoset and passing counterfactual payoffs to these instances gives a procedure that achieves sublinear ranking regret, which we call CAT-CFR. Details and pseudocode of CAT-CFR are given in~\cref{appsec:catcfr} as~\cref{alg:catcfr}.

\begin{restatable}{theorem}{catcfrregret}
\label{thm:catcfrregret}
\label{thm:siregretboundprob}
Consider a 2p0s-EFGSAS. With an appropriate choice of parameters, a player using CAT-CFR has external regret bounded as
\(
\mathbb{E}[[R^{\mathsf{EXT}}_{T,i}(\Omega_i)]^+]
\le16\lvert\Sigma_i^{\mathrm{orig}}\rvert\sqrt{T}.
\)
Moreover, for every \(p\in(0,1)\), with probability at least \(1-p\),
\(
R^{\mathsf{EXT}}_{T,i}(\Omega_i)
\le\frac{16}p\lvert\Sigma_i^{\mathrm{orig}}\rvert\sqrt{T}.
\)

\end{restatable}
Crucially, the regret bound grows with the number of sequences in the \emph{base game}, rather than the na\"ive expanded strategy space.
We next connect the external-regret guarantee to equilibrium convergence by averaging the induced implementable sequence-form strategies. 


\begin{proposition}\label{prop:catNEconvergeEFG}
\label{prop:NEconvergeEFG}
\label{prop:SINash}
Consider a 2p0s-EFGSAS where players incur external regrets \(R^{\mathsf{EXT}}_{T,1}(\Omega_1)\) and \(R^{\mathsf{EXT}}_{T,2}(\Omega_2)\) after \(T\) timesteps. Let \(\bar{\omega}_i^T\coloneqq\frac{1}{T}\sum_{t=1}^T\omega_i^t\) be player \(i\)'s average marginal sequence-form strategy. Then any strategy profile \((\chi_1,\chi_2)\) implementing \((\bar{\omega}_1^T,\bar{\omega}_2^T)\) is a \((R^{\mathsf{EXT}}_{T,1}(\Omega_1)+R^{\mathsf{EXT}}_{T,2}(\Omega_2))/T\)-approximate NE.
\end{proposition}

This ensures that with high probability, the strategies implementing the average marginals produced by CAT-CFR converge to an approximate NE.

\subsection{Extracting Compact Equilibria via Stochastic Approximation}\label{sec:SA}
By~\cref{thm:siregretboundprob,prop:NEconvergeEFG}, the average \emph{marginal} sequence-form strategies have a small equilibrium gap with high probability. However, these marginals do not directly specify how a player should act at each infoset. We therefore seek compact weight vectors defining valid strategies that approximately implement these marginals, whose existence is guaranteed by~\cref{thm:efgsascompact}.

In the remainder of this section, we develop a stochastic approximation (SA) procedure to recover these weights, with full details given in~\cref{appsecs:add_SA}.
At a high level, the SA procedure (\cref{alg:compute_w_efg})  seeks a compact weight vector $W_i\in\mathbb{R}^{|\Sigma^{\mathrm{orig}}_i|}_{ > 0}$ such that its induced implementable sequence-form strategy $\hat\omega_i(W_i)$ is approximately $\omega_i^*$, the marginal sequence-form strategy induced by an (approximate) Nash equilibrium $\chi^*$. 
The following result shows that, whenever the empirical averages converge, the induced marginals of the SA output converge to the same limit.

\begin{restatable}{proposition}{stochasticapproximation}\label{prop:stochasticapproxasymptotic}
    \label{prop:compute_w} Let $W^T_i$ be the vector produced by~\cref{alg:compute_w_efg}.
    Assume that $\frac{1}{T}\sum_{t=1}^T \chi_i^t(S^t_i)\to \omega_i^*$ as $T\to \infty$ almost surely, and that $\sum_{t=1}^{\infty}\eta_t = \infty$ and $\sum_{t=1}^{\infty}\eta^2_t < \infty$. Then, almost surely, $\|\hat\omega_i(W_i^T)-\omega_i^*\|_1\to0$ as $T\to\infty$. 
\end{restatable}

Beyond this conditional asymptotic convergence, we are also interested in 
finite-time bounds on the error in the induced implementable strategies and on their
equilibrium gap. To this end, we utilize a version of the \emph{robust stochastic approximation} (RSA) approach introduced by~\cite{nemirovski1978cezari,nemirovski2009robust}, which modifies~\cref{alg:compute_w_efg} by adding a quadratic regularizer, taking a stepsize-weighted average and utilizing the suffix averaging scheme proposed by~\cite{rakhlin2012making} 
(details in~\cref{appsec:RSAdetails}).
Under the RSA procedure (\cref{alg:compute_w_efg_rsa}), we obtain a finite-time convergence result. Crucially, our bound is given in terms of the squared \emph{duality gap} $\gamma^2$ (cf. \cref{def:dualitygap}) of the joint strategy profile induced by the compact weight vectors $W_i$. 


\begin{restatable}{theorem}{mainthm}\label{thm:finitetimeRSAEFG}
Consider a 2p0s-EFGSAS $\calG$, and run CAT-CFR
for $T$ rounds with the same parameters as in~\cref{thm:catcfrregret}. For each player $i$, apply Algorithm~\ref{alg:compute_w_efg_rsa} to $\tilde\omega_i^T=\frac1T\sum_{t=1}^T\chi_i^t(S_i^t)$, and let $\tilde{W}$ be the returned compact vectors.
Then, for all $p\in(0,1)$, with probability at least $1-p$, we have 
$\gamma^2(\tilde{W}) \leq  O\left(\frac{1}{p\sqrt T}
\right)$.
\end{restatable}
Here $O(\cdot)$ hides factors polynomial in $|\Sigma^\mathrm{orig}_i|$ and $|\calI_i|$, the number of player $i$'s sequences and infosets in the base game, respectively. 
Note that these factors depend only on $\calG^{\mathrm{orig}}$, and not the possibly exponential size of the expanded game $\calG^\dagger$.
Combining the statements of Theorems~\ref{thm:siregretboundprob} and~\ref{thm:finitetimeRSAEFG} ensures that with high probability, running RSA after CAT-CFR leads to a compact vector $W$ that encodes a strategy with low duality gap in 2p0s-EFGSAS.


\section{Experimental Results}
We seek to corroborate the theoretical results for our proposed computational procedure with empirical observations. All figures show the mean and central 95\% interval over 100 runs, and additional figures and details are given in~\cref{appsecs:add_example}.

\paragraph{Experiment 1: Running Example}
We run CAT-CFR and~\cref{alg:compute_w_efg_rsa} on the game from~\cref{ex:updown} with $\alpha=0$, $\lambda=0.5$. We repeat for 100 runs with a time horizon of \(T=10,000\), sampling the stochastic action availabilities with different random seeds each run. In~\cref{fig:updown_compact}, we show the mean and central 95\% interval across all runs of the entries of the resulting compact vector, and compare against \(\hat W_i^* = \lim_{\epsilon \to 0} W_i^*\) computed analytically in~\cref{example:lsg_compact}. Since the RSA procedure converges quickly, we show only the first 1000 iterates here, and show a plot containing all iterates in \cref{fig:updown_compact_full}. We observe that the entries of the computed \(W_i\) converge to the corresponding values of \(\hat W_i^*\). 

\paragraph{Experiment 2: Convergence in larger EFGSAS}
We additionally investigate the convergence in terms of duality gap \(\gamma\) of CAT-CFR and~\cref{alg:compute_w_efg_rsa} on variants of Kuhn poker (SAS-Kuhn) and stochastic binary pursuit (SBP) EFGSAS, which are larger EFGSAS than \cref{ex:updown}. Both EFGSAS are defined formally in~\cref{appsubsec:game_definitions}. In~\cref{fig:spr_marginal}, we observe convergence in \(\gamma\) of the time-averaged marginals played by CAT-CFR, corroborating~\cref{prop:SINash}.
To allow comparisons of~\cref{alg:compute_w_efg_rsa} across runs, all runs of a game use the same marginal from CAT-CFR as input.
In~\cref{fig:spr_compact}, we observe convergence in \(\gamma\) of the strategy induced by the compact vector \(\tilde W_i\) produced by~\cref{alg:compute_w_efg_rsa} for various time horizons \(T\), corroborating~\cref{thm:finitetimeRSAEFG}.

\begin{figure}[ht]
  \centering
    \includegraphics[width=0.6\linewidth]{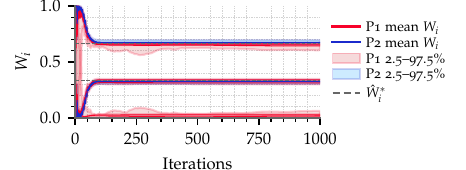}
    \caption{\(W^t_i\) entries obtained via CAT-CFR and~\cref{alg:compute_w_efg_rsa}. The entries converge to \(\hat W^*_i = \lim_{\epsilon \to 0} W_i^*\) from~\cref{example:lsg_compact}.}
    \label{fig:updown_compact}
\end{figure}
\begin{figure}[ht]
    \centering
\begin{subfigure}[t]{0.49\linewidth}
        \vspace{0pt}
    \includegraphics[width=\linewidth]{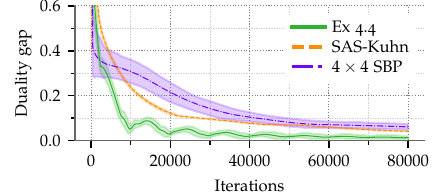}
    \caption{Duality gap of time-averaged empirical marginals \(\tilde \omega_i^t\) at intermediate iterates, played by CAT-CFR.}
    \label{fig:spr_marginal}
\end{subfigure}
\hfill
\begin{subfigure}[t]{0.49\linewidth}
        \vspace{0pt}
    \includegraphics[width=\linewidth]{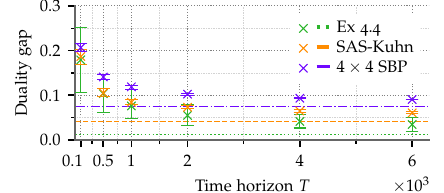}
    \caption{Duality gap of strategies induced by the \(\tilde W_i\) computed by \cref{alg:compute_w_efg_rsa} for various time horizons \(T\). 
    We observe convergence to the duality gap of the CAT-CFR output, which is shown using the horizontal lines.}
    \label{fig:spr_compact}
\end{subfigure}
    \caption{Results for~\cref{example:lsg_compact}, SAS-Kuhn, and \(4 \times 4\) SBP.}
\end{figure}

\section{Discussion and Future Work}\label{sec:discussion}

In this paper, we have formalized and studied the properties of strategy representation in EFGSAS, 
and provided a new algorithm for compact Nash equilibrium computation in 2p0s-EFGSAS. 
Our analysis leaves open several research directions. First, in this paper we focused on the ex-interim setting of action availabilities. However, it is important to further study the computational properties of the $k$-ex-interim setting and conditions under which similar ``compactification'' of players' strategies can be derived.
Second, the duality gap bound for computing compact Nash equilibria is relatively loose due to the additional quadratic factor introduced by the RSA procedure. Improving upon this bound  using tailored stochastic approximation techniques is an interesting open direction. Finally, we have utilized CAT-FTPL as a ranking regret minimizer for CAT-CFR, rather than a ``standard'' regret matching algorithm. 
Deriving a suitable modification of regret matching that obtains sublinear ranking regret in EFGSAS would be a significant contribution.


\subsection*{Acknowledgements}
This project is supported by the Ministry of Education, Singapore, under the Academic Research Fund Tier 1 (FY2025) and by the National University of Singapore, under the Start-Up Grant Scheme.

\subsection*{LLM Usage Statement}
In this work, we used generative AI tools to assist in cleaning up and refining author-written drafts of proofs, to aid in writing code to implement our proposed method, and to sparingly identify related literature. We have reviewed all AI-assisted work: all suggested edits to proofs were verified by hand, a manual literature survey was done to ensure that the identified research was valid, and any LLM-generated code was verified and tested for correctness. We take full responsibility for the final content of this work, including text, claims or artifacts produced with the aid of generative AI. 

\bibliographystyle{plainnat} 
\bibliography{ref}

\newpage
\appendix
\setcounter{tocdepth}{1}
\renewcommand{\theHsection}{A\arabic{section}}
\section*{Appendix}
This supplementary material contains 
additional preliminaries and algorithmic details for normal-form GSAS in~\cref{appsecs:addprelims}, additional details on CFR and CAT-CFR in~\cref{appsecs:cfr-sicfr}, proofs omitted from the main paper for space considerations in~\cref{appsecs:efgproofs4} and~\cref{appsecs:efgproofs5}, and further experimental results and details in~\cref{appsecs:add_example}.
Finally, we also provide a notation table in~\cref{appsecs:table}.


\section{Additional Preliminaries and Algorithms}\label{appsecs:addprelims}
\subsection{Normal-Form GSAS Setting and Properties}\label{appsec:gsasprelims}
In this section, we provide preliminaries and relevant properties of normal-form GSAS, as introduced and analyzed by~\cite{schwarz2026computing}.

Consider an $n$-player normal/strategic-form game $\mathcal{G}^\text{orig}$ (not to be confused with extensive-form games studied throughout the rest of this paper) with finite action sets $A_i$, joint action space $A = A_1 \times \dots \times A_n$, and payoff functions $u_i: A \rightarrow [-1, 1]$. 
In line with prevailing conventions, we denote $a=(a_1,\dots,a_n)$ to be an action profile, and as shorthand $u_i(a) = u_i(a_1,\dots,a_n)$. We also denote by $-i$ the set of players other than $i$, such that $u_j(a'_i, a_{-i}) = u_j(a_1,\dots, a_i', \dots, a_n)$.

\begin{definition}[GSAS]
    Given a game $\mathcal{G}^\text{orig} = (A_1, \dots, A_n, u_1, \dots, u_n)$, let $\mathcal{S}_i \subseteq 2^{A_i} \backslash \{ \varnothing \}$ such that $\mathcal{S} = \mathcal{S}_1 \times \dots \times \mathcal{S}_n$, and $\rho \in \Delta(\mathcal{S})$ be a distribution over elements $\mathcal{S}$, such that $\rho(S), S \in \mathcal{S}$ gives the probability that stochastic action set $S$ is observed.
    A normal form \textbf{G}ame with \textit{\textbf{S}tochastic \textbf{A}ction \textbf{S}ets} is given by the tuple $\mathcal{G} = (\mathcal{G}^\text{orig}, \mathcal{S}, \rho)$.
\end{definition}

\begin{definition}[2p0s-GSAS]
    A two-player zero-sum GSAS (2p0s-GSAS) $\mathcal{G}$ is one where $\mathcal{G}^\text{orig}$ is two-player zero-sum, i.e., $n=2, u_1(a)=-u_2(a)$ for all action profiles $a \in A$.
\end{definition}

A GSAS proceeds as follows. At the start of the game, each player privately receives their action set $S_i \in \mathcal{S}_i$ from Nature based on $\rho$. Each then plays an action $a_i \in S_i$ simultaneously and receives a reward $u_i(a)$ based on the action profile $a \in A$. While each player $i$ observes their availability set $S_i$ and might have full knowledge of $\rho$,
they do not observe their opponents' availability sets $S_{-i}$ at any point. 
The analogous independence assumption made for GSAS is:


\begin{assumption}
     $\rho(S) = \prod_{i=1}^{n} \rho_i(S_i)$ for some probability distributions $\rho_i: \mathcal{S}_i \rightarrow [0, 1]$, i.e., the availability of actions is independent across players.
    \label{assumption:product}
\end{assumption}
Even under this assumption, a GSAS remains large since $|\mathcal{S}_i|$ remains exponential in the size of $A_i$.
A pure strategy in a GSAS is a deterministic mapping $\pi_i: \mathcal{S}_i \rightarrow A_i$ where $\pi_i(S_i) \in S_i$. A \textit{mixed strategy} (or simply strategy) for player $i$ is a mapping $\pi_i: \mathcal{S}_i \to \Delta(A_i)$ such that $\operatorname{supp}(\pi_i(S_i)) \subseteq S_i$. 
A player's strategy gives, for every possible subset of actions they could observe, a distribution of actions corresponding to the observed action subset. 

Notice that this definition is compatible with strategies in \emph{Bayesian games}~\citep{harsanyi1968games}. In particular, a GSAS can be represented as a Bayesian game in which action
availabilities are modeled as types. Hence, strategies can be viewed as mappings from types to distributions over actions. We further clarify that while this connection is intuitive and immediate in normal-form GSAS, the nuance of the EFGSAS setting is that \emph{when} the players observe the realization of their action availabilities has a significant effect on the strategy representation.


Given a joint action set $S\in\mathcal{S}$, $\pi$ denotes the joint strategy of all players and $\pi(a \mid S) = \prod_{i \in [n]} \pi_i(a_i \mid S_i)$ is the probability of an action profile $a$ for every $a\in A$.
The expected payoff to player $i$ is then given by
    $U_i(\pi) = \mathbb{E}_{S \sim \rho}\left[\mathbb{E}_{a \sim \pi(S)}[u_i(a)]\right]$,
where the inner expectation is over the actions sampled independently according to each player's strategy given their available action sets, and the outer expectation is over the stochastic action set $S$ drawn from $\rho$. The expected payoff of a specific action $a_i$ of player $i$ w.r.t. the ensemble of the opponents' strategies $\pi_{-i}$ is given by 
    $U_i(a_i; \pi_{-i}) = \mathbb{E}_{S \sim \rho}\left[\mathbb{E}_{a_{-i} \sim \pi_{-i}(S_{-i})}[u_i(a_i,a_{-i})]\right]$.

Similarly to the EFG setting, an $\epsilon$-NE of a GSAS is a strategy profile from which no player can increase their expected payoff by more than $\epsilon$ by unilaterally deviating, i.e., $\forall i \in [n],\ \forall \pi_i'$, we have $U_i(\pi_i, \pi_{-i}) \ge U_i(\pi_i', \pi_{-i}) - \epsilon$.
These definitions extend the classical normal-form Nash equilibrium to GSAS and are consistent with Bayesian games and Bayesian-Nash equilibria~\citep{harsanyi1968games}.

Under~\cref{assumption:product}, GSAS have several key properties which were shown in~\cite{schwarz2026computing}. For the sake of brevity, we review the most relevant properties to our work. First, a player's expected payoff can be written in terms of the marginal distribution over actions $a_i\in A_i$ induced by strategy $\pi_i$. In particular, let
\begin{equation}
    \mathbb{P} \left[ a_i ; \rho_i, \pi_i \right] =
    \sum_{S_i \in \mathcal{S}_i}  \left( \rho_i(S_i)  
     \pi_i(a_i | S_i) \mathds{1}\{ a_i \in S_i \}\right).
\end{equation}

Then, the expected payoff of a player for joint strategy profile $\pi$ is
\begin{align}
    U_i(\pi) = \sum_{a \in A} u_i(a) \prod_{j \in [n]} \mathbb{P} \left[ a_j ; \rho_j, \pi_j \right].
    \label{eq:depend-only-marginals}
\end{align}


The normal-form analogue of an implementable strategy is defined as follows.
\begin{definition}\label{def:implementable}
Let $\mu_i \in \Delta(A_i)$ be a probability distribution over player $i$'s possible actions. $\mu_i$ is called \textit{implementable} if there exists a strategy $\pi_i:\mathcal{S}_i \to \Delta(A_i)$ such that, for every action $a_i \in A_i$, $\mu_i(a_i) = \mathbb{P}[a_i; \rho_i, \pi_i]$. 
In this case, we also say that $\pi_i$ \textit{implements} $\mu_i$, or that $\pi = (\pi_1,\dots, \pi_n)$ implements $\mu = (\mu_1, \dots, \mu_n)$ if $\pi_i$ implements $\mu_i$ for all $i$. The set of implementable strategies for player $i$ is denoted by $M_i \subseteq \Delta(A_i)$.
\end{definition}

Given an implementable marginal profile $\mu$, the expected payoff to player $i$ following $\mu$ is $U_i(\mu) = \mathbb{E}_{S \sim \rho}\left[\mathbb{E}_{a \sim \mu}[u_i(a)]\right]$. Then, 
let $U_i(a_i; \mu_{-i}) = \mathbb{E}_{S \sim \rho}\left[\mathbb{E}_{a_{-i} \sim \mu_{-i}}[u_i(a_i,a_{-i})]\right]$.
By definition of $\mu$, $\mathbb{E}_{S \sim \rho}\left[\mathbb{E}_{a \sim \pi(S)}[u_i(a)]\right] = \mathbb{E}_{S \sim \rho}\left[\mathbb{E}_{a \sim \mu}[u_i(a)]\right]$.
Similarly to EFGSAS, implementable strategies suffice to capture NE in GSAS.


\begin{proposition}
    \label{thm:ne-implementable-GSAS}
    Consider a GSAS where $\pi = (\pi_1, \dots, \pi_n)$ is a strategy profile that implements $\mu = (\mu_1, \dots, \mu_n)$. Then $\pi$ is an $\epsilon$-Nash equilibrium if and only if for all $i \in [n]$,
        $U_i(\mu) \geq \max_{\mu_i' \in M_i} U_i(\mu_i', \mu_{-i}) - \epsilon$.
\end{proposition}

A version of the minimax theorem also holds for implementable strategies in 2p0s-GSAS under~\cref{assumption:product}.

\begin{proposition}
    \label{thm:ne-2p0s}
    For a 2p0s-GSAS, a strategy profile $(\pi_1^*, \pi^*_2)$ is an NE if and only if their associated $(\mu^*_1, \mu^*_2)$ is a saddle point of the function $U=U_1=-U_2$, i.e.,
    \begin{align*}
        \mu^*_1 \in \argmax_{\mu_1 \in M_1} \min_{\mu_2 \in M_2} U(\mu_1, \mu_2), \quad \mathrm{and} \quad
        \mu^*_2 \in \argmin_{\mu_2 \in M_2} \max_{\mu_1 \in M_1} U(\mu_1, \mu_2),
    \end{align*}
        where
        $\max_{\mu_1 \in M_1} \min_{\mu_2 \in M_2} U(\mu_1, \mu_2) = \min_{\mu_2 \in M_2} \max_{\mu_1 \in M_1} U(\mu_1, \mu_2)$.
\end{proposition}

Moreover, it holds that every implementable NE of $\calG^{\mathrm{orig}}$ corresponds to an NE of $\calG$. Note that even in the 2p0s-case, \cref{thm:gsas-contains-ne} requires the NE in $\mathcal{G}^\text{orig}$, $x^*$, to be implementable for both players, i.e., $x^*_i$ being implementable does not imply a solution to the max-min problem (or optimal strategy for player $i$).
\begin{proposition}
Consider GSAS $\mathcal{G} = (\mathcal{G}^\text{orig}, \mathcal{S}, \rho)$. Let $x^*=(x^*_1,\dots,x^*_n)$ be an $\epsilon$-NE of $\mathcal{G}^\text{orig}$, where $x^*_i \in \Delta(A_i)$. 
 If $\mu^*=x^*$ is implementable in $\mathcal{G}$ by $\pi^*=(\pi_1, \dots, \pi_n)$, then $\pi^*$ is an $\epsilon$-NE in $\mathcal{G}$. 
\label{thm:gsas-contains-ne}
\end{proposition}

The above discussion indicates that for the purposes of equilibrium representation, and analogous to the EFGSAS setting, it suffices to work over the space of $M_i$ rather than the larger set of possible $\pi_i$. A key result is that in a GSAS satisfying~\cref{assumption:product}, every $\mu_i\in M_i$ can be approximately implemented by a (possibly non-unique) compact, polynomially sized vector $w_i$. This is formally stated in~\cref{lem:compact}.


\subsection{CAT-FTPL}\label{appsec:catftpl}

CAT-FTPL~\citep{schwarz2026computing} is a follow-the-perturbed-leader algorithm that minimizes ranking regret in normal-form GSAS. It requires only sample access to the availability distribution~$\rho_i$. On round~$t$, player~$i$ observes $S_i^t\sim\rho_i$. After the players act, consider the complete loss
vector denoted by
\[
\ell_i^t(a_i)\coloneqq\frac{1-u_i(a_i,a_{-i}^t)}{2}\in[0,1],\qquad a_i\in A_i.
\]

Due to the presence of stochastic action availabilities, player~$i$ observes only the coordinates with $a_i\in S_i^t$. Conditional on everything revealed before round~$t$, the complete loss vector must be independent of $S_i^t$, player~$i$'s current perturbations, and the auxiliary samples described below. This condition holds under simultaneous play and Assumption~\ref{assumption:product} when the players randomize independently.

Let $\hat L_i^{t-1}$ be the cumulative estimated-loss vector before round~$t$. CAT-FTPL draws an independent standard Gumbel random variable $Z_i^t(a_i)$ for every $a_i\in A_i$, where $\mathbb{P}\bigl[Z_i^t(a_i)\leq z\bigr]=\exp(-\exp(-z))$,
and selects
\begin{equation}\label{eq:cat-choice}
a_i^t\in\operatorname*{argmin}_{a_i\in S_i^t}\left\{\eta_i\hat L_i^{t-1}(a_i)-Z_i^t(a_i)\right\}.
\end{equation}

Equivalently, the strategy induced before observing the current set is
\begin{equation}\label{eq:cat-policy}
\pi_i^t(a_i\mid S_i)=\frac{\mathds{1}\{a_i\in S_i\}\exp\left(-\eta_i\hat L_i^{t-1}(a_i)\right)}{\sum_{b_i\in S_i}\exp\left(-\eta_i\hat L_i^{t-1}(b_i)\right)}.
\end{equation}

CAT-FTPL uses geometric resampling to correct for the selective observation of action losses.
In particular, CAT-FTPL independently generates, for every $a_i\in S_i^t$, the sequences
$\tilde S_i^{t,1}(a_i),\tilde S_i^{t,2}(a_i),\ldots\overset{\mathrm{i.i.d.}}{\sim}\rho_i$
and independent stopping indicators
\[
J_i^{t,k}(a_i)\sim\operatorname{Bernoulli}\left(\frac{\zeta_i}{1+\zeta_i}\right).
\]

All the above auxiliary random variables are independent of the current round's play. For every $a_i\in A_i$, define
\begin{equation}\label{eq:cat-waiting}
N_i^t(a_i)\coloneqq\min\left\{k\geq1:a_i\in\tilde S_i^{t,k}(a_i)
\text{ or }J_i^{t,k}(a_i)=1\right\},
\end{equation}
which is intuitively the number of trials until the first time an action is observed.
The stopping indicator increases the success probability of each trial and controls the tail of the loss estimate. Then, the complete procedure of CAT-FTPL is given in~\cref{alg:catftpl}.

\begin{algorithm}[ht]
\caption{CAT-FTPL for player~$i$}
\label{alg:catftpl}
\begin{algorithmic}[1]
\renewcommand{\algorithmicrequire}{\textbf{Input:}}
\Require $\eta_i,\zeta_i>0$
\State $\hat L_i^0\gets
\mathbf{0}\in\mathbb{R}_{\geq0}^{|A_i|}$
\For{$t=1,2,\ldots,T$}
    \State Observe $S_i^t$, draw the Gumbel perturbations, and select $a_i^t$ according to~\eqref{eq:cat-choice}
    \State Observe $\ell_i^t(a_i)$ for every $a_i\in S_i^t$
    \State $\hat\ell_i^t\gets\mathbf{0}\in\mathbb{R}_{\geq0}^{|A_i|}$
    \For{$a_i\in S_i^t$}
        \State Generate $N_i^t(a_i)$ according to~\eqref{eq:cat-waiting}
        \State $\hat\ell_i^t(a_i)\gets\ell_i^t(a_i)N_i^t(a_i)/(1+\zeta_i)$
    \EndFor
    \State $\hat L_i^t\gets\hat L_i^{t-1}+\hat\ell_i^t$
\EndFor
\end{algorithmic}
\end{algorithm}

To see the role of the resampling step, consider an action $a_i$ and let $q_i(a_i)\coloneqq\mathbb{P}_{S_i\sim\rho_i}[a_i\in S_i]$. Then
\[
\mathbb{E}\bigl[N_i^t(a_i)\bigr]=\frac{1+\zeta_i}{q_i(a_i)+\zeta_i}.
\]

If $\bar\ell_i^t(a_i)$ denotes the conditional mean of
$\ell_i^t(a_i)$ given the preceding rounds, the independence condition
above gives
\[
\mathbb{E}\bigl[\hat\ell_i^t(a_i)\mid
\text{preceding rounds}\bigr]=\frac{q_i(a_i)}{q_i(a_i)+\zeta_i}
\bar\ell_i^t(a_i).
\]

The estimate is intentionally biased toward zero, since the offset $\zeta_i$
keeps the correction controlled even for rarely available actions.
For a ranking $\kappa$ of $A_i$, let $\kappa(S_i)$ be its highest-ranked action in $S_i$, and let $\pi_i^\kappa$ be the strategy that selects $\kappa(S_i)$ for every $S_i$. The ranking regret of a sequence of strategy profiles $\pi^1,\ldots,\pi^T$ is
\[
R^{\mathsf{RANK}}_{T,i}\coloneqq\max_{\kappa\text{ ranking of }A_i}\sum_{t=1}^T\left(U_i(\pi_i^\kappa,\pi_{-i}^t)-U_i(\pi^t)\right).
\]

\begin{proposition}[CAT-FTPL regret bound]
\label{prop:catftplregret}
For every
$\eta_i,\zeta_i>0$, CAT-FTPL satisfies
\[
\mathbb{E}\left[R^{\mathsf{RANK}}_{T,i}\right]\leq2\left[\log|A_i|\left(\frac{1}{\eta_i}+\frac{1}{\zeta_i}\right)+|A_i|T(\eta_i+\zeta_i)\right].
\]

In particular, setting the parameters
$
\eta_i=\zeta_i=\sqrt{\frac{\log|A_i|}{|A_i|T}}
$
gives
\[
\mathbb{E}\left[R^{\mathsf{RANK}}_{T,i}\right]\leq8\sqrt{|A_i|T\log|A_i|},\qquad\mathbb{E}\left[\left[R^{\mathsf{RANK}}_{T,i}\right]^+\right]\leq16\sqrt{|A_i|T\log|A_i|}.
\]

Consequently, for every $p\in(0,1)$, with probability at least $1-p$,
\[
R^{\mathsf{RANK}}_{T,i}\leq\frac{16}{p}\sqrt{|A_i|T\log|A_i|}.
\]
\end{proposition}

\begin{proof}
The ranking regret bound was shown in \cite[Theorem~5.5]{schwarz2026computing}. The bound on the positive component of the regret was derived in the proof of \cite[Proposition~5.6]{schwarz2026computing}. The probability bound then follows from Markov's inequality applied to $[R^{\mathsf{RANK}}_{T,i}]^+$.
\end{proof}

\subsection{On Bit-Compactness in GSAS}\label{appsec:bitcompact}
\cref{thm:efgsascompact} establishes approximate compact representations of EFGSAS strategies under~\cref{assumption:productefg}, where compactness here refers to the number of real-valued coordinates of $W_i$. Throughout the paper, this is taken as the meaning of the term ``compact''. A stronger requirement is \emph{bit-compactness}, which is intuitively a binary strategy description of length polynomial in the size of $\calG^\mathrm{orig}$. It turns out that for 2p0s-GSAS that do not meet an analogous independence assumption (cf.~\cref{assumption:product}),~\cite[Theorem D.5]{schwarz2026computing} rules out such a uniform polynomial-bit representation of $\epsilon$-NE. Since normal-form GSAS are a special case of EFGSAS, 2p0s-EFGSAS that do not meet~\cref{assumption:productefg} also exhibit this obstruction. 

It is worth noting, however, that this lower bound holds for the \emph{bit-representation} and thus does not rule out representations with polynomially many real numbers, nor the necessity of independence for any compact real-valued representation. Nevertheless, our results in~\cref{sec:efg_repr} and~\cref{sec:efg_compute} show that~\cref{assumption:productefg} is \emph{sufficient} for our purposes of developing an end-to-end computational technique that uses a polynomial number of real numbers to play strategies implementing an approximate NE in EFGSAS.

\section{CFR and CAT-CFR}\label{appsecs:cfr-sicfr}
\subsection{CFR Preliminaries}\label{appsec:cfrprelims}

We first recall the regret-minimization interpretation of CFR through scaled extensions, which were introduced in~\cite{farina2019efficient}. Let $n$ and $m$ be positive integers, let $\calX\subseteq\mathbb{R}_{\geq0}^n$ and $\calY\subseteq\mathbb{R}_{\geq0}^m$ be nonempty, compact, convex sets, and let $f:\calX\to\mathbb{R}_{\ge0}$ be affine. Their scaled extension consists of the vectors $(x,f(x)y)$ with $x\in\calX$ and $y\in\calY$. Formally, we give the definition below.

\begin{definition}[Scaled Extension]\label{def:scaledextension}
        Let $\calX$ and $\calY$ be nonempty, compact and convex sets, and let $f:\calX\to \mathbb{R}_{\ge0}$ be a non-negative affine function. The scaled extension of $\calX$ with $\calY$ via $f$, where $f(x) = \langle a, x\rangle+b$, is defined by the set
        \(
            \calX \overset{f}{\vartriangleleft} \calY \coloneqq \{(x,y):x\in\calX, y\in f(x)\calY\}
        \).
    \end{definition}

\cref{alg:cfr-basic} can be composed recursively to construct an external regret minimizer for any set that is expressed via a chain of scaled extensions, such as the polytope of sequence-form strategies in a standard EFG. When applied to the polytope of sequence-form strategies, this gives the standard CFR algorithm of~\cite{zinkevich2007regret} if all external regret minimizers for the individual simplices in the chain of scaled extensions are implemented using the regret matching (RM) algorithm~\citep{hart2000simple}. 
The resulting procedure has sublinear external regret whenever the component procedures do.

\begin{algorithm}[ht]
\caption{External-regret minimization on a scaled extension}
\label{alg:cfr-basic}
\begin{algorithmic}[1]
\Require External-regret minimizers
$\calR_{\calX}^{\mathsf{EXT}}$ and $\calR_{\calY}^{\mathsf{EXT}}$
\Require An affine function $f:\calX\to\mathbb{R}_{\ge0}$
\For{$t=1,2,\ldots,T$}
    \State Obtain $x^t\in\calX$ from $\calR_{\calX}^{\mathsf{EXT}}$
    \State Obtain $y^t\in\calY$ from $\calR_{\calY}^{\mathsf{EXT}}$
    \State Play $(x^t,f(x^t)y^t)$
    \State Receive a linear payoff vector
    $u^t=(u_{\calX}^t,u_{\calY}^t)$
    \State Pass $u_{\calY}^t$ to $\calR_{\calY}^{\mathsf{EXT}}$
    \State Pass
    $u_{\calX}^t+\langle u_{\calY}^t,y^t\rangle\nabla f$
    to $\calR_{\calX}^{\mathsf{EXT}}$
\EndFor
\end{algorithmic}
\end{algorithm}

\subsection{Detailed Description and Discussion of CAT-CFR}
\label{appsec:catcfr}

A natural idea for solving 2p0s-EFGSAS is to use a sample-efficient CFR-type algorithm such as external-sampling MCCFR~\citep{lanctot2009monte} to obtain sublinear external-regret bounds on the expanded game. Unfortunately, several issues arise when applying MCCFR to the expanded game tree: (i) the size of an explicit sequence-form strategy representation is still exponential, and (ii) the convergence bound of MCCFR contains a factor that is upper bounded by the number of infosets in the expanded game $\calG^\dagger$, which is exponential in our case. 

We propose a modified version of CFR which we call CAT-CFR (pseudocode in \cref{alg:catcfr}) that exploits the additional structure afforded by ex-interim EFGSAS, allowing it to be run on the base game $\calG^{\mathrm{orig}}$ directly.
CAT-CFR runs one CAT-FTPL instance at every infoset of the base game. The instance at $I\in\calI_i$ stores one cumulative estimated loss $\hat L_{I,i}^t(a)$ for every $a\in A_{I,i}$. Let $\calF_{t-1}$ contain everything revealed through the end of timestep~$t-1$. For every $I\in\calI_i$, fix a (constant) learning rate $\eta_{I,i}>0$. For every $S\in\calS_{I,i}$, the behavioral strategy at $I$ is
\begin{equation}\label{eq:catcfr-local-policy}
b_{I,i}^t(a\mid S)\coloneqq\frac{
\mathds{1}\{a\in S\}\exp\left(-\eta_{I,i}\hat L_{I,i}^{t-1}(a)\right)}{\sum_{a'\in S}\exp\left(-\eta_{I,i}\hat L_{I,i}^{t-1}(a')\right)}.
\end{equation}

Notice that this is precisely the distribution induced by the Gumbel-perturbed update rule in~\cref{alg:catftpl} (Equation~\eqref{eq:cat-choice}). The marginal induced by this behavioral strategy is
\begin{equation}\label{eq:catcfr-local-marginal}
\mu_{I,i}^t(a)
\coloneqq
\sum_{S\in\calS_{I,i}}
\rho_{I,i}(S)b_{I,i}^t(a\mid S).
\end{equation}

CAT-CFR performs traversals of the game tree similarly to standard CFR, with some modifications due to the stochastic action sets.
At each timestep~$t$, $S_{I,i}^t\sim\rho_{I,i}$ is drawn for every $i\in\calN$ and $I\in\calI_i$, including infosets that are not reached under the sampled availability sets.
These draws are independent under Assumption~\ref{assumption:productefg}.
The strategy mappings in~\eqref{eq:catcfr-local-policy} are fixed before these draws and remain unchanged until all counterfactual feedback has been collected.
CAT-FTPL also uses independent auxiliary availability samples for geometric resampling. Neither the availability distributions nor the marginals in~\eqref{eq:catcfr-local-marginal} need to be computed.



We next describe a CAT-CFR iteration on the base game.
In~\cref{alg:catcfr}, the \textsc{NextStrategy} function invokes ranking-regret minimizers (i.e., CAT-FTPL) at each infoset and returns one valid sequence-form strategy per player in the base-game treeplex, with support restricted to the sampled action sets.
Then, the \textsc{ObserveUtility} function receives counterfactual payoffs for the available actions at each infoset and updates the local ranking-regret minimizers.
For convenience, we denote $x_i^t\coloneqq\chi_i^t(S_i^t)$ and
$S_i^t\coloneqq(S^t_{I,i})_{I\in\calI_i}$.

Formally, if $h$ is a nonterminal history and $a\in A_h$, let $ha$ be its successor after action~$a$ is selected. With the sampled sets and current policies held fixed, let $V_i^t(h)$ be player $i$'s expected terminal payoff starting at $h$. These values satisfy the recursion: 
\begin{equation}\label{eq:catcfr-continuation}
V_i^t(h)=
\begin{cases}
u_i(h),
& h\in\calZ,\\
\displaystyle\sum_{a\in A_h}\mathbb{P}_c(a\mid h)V_i^t(ha),
& \mathrm{player}(h)=c,\\
\displaystyle\sum_{a\in S_{I_h,j}^t}b_{I_h,j}^t(a\mid S_{I_h,j}^t)V_i^t(ha),
& \mathrm{player}(h)=j\in\calN.
\end{cases}
\end{equation}

For $h\in I\in\calI_i$, let $\mathbb{P}_{-i}^t(h)$ be the product, along the path from the root to $h$, of the original chance probabilities and the probabilities assigned by $b_{I',j}^t(\cdot\mid S_{I',j}^t)$ at infosets of players $j\neq i$. All probabilities belonging to player~$i$ are omitted. For every base action $a\in A_{I,i}$, we define the counterfactual payoff as
\begin{equation}\label{eq:catcfr-sampled-value}
\tilde u_{I,i}^t(a)\coloneqq\sum_{h\in I}\mathbb{P}_{-i}^t(h)V_i^t(ha).
\end{equation}

We assume that the solver receives the counterfactual payoffs in~\eqref{eq:catcfr-sampled-value} for every $I\in\calI_i$ and $a\in S_{I,i}^t$.
The complete vector is defined for the analysis but payoffs for unavailable actions need not be observed.
This feedback is available also at infosets that are not reached under the sampled availability sets.
In particular, an unavailable earlier action of player~$i$ does not remove the feedback requirement at a later infoset, since $\mathbb{P}_{-i}^t(h)$ omits player~$i$'s own reach probabilities.

Under the above definitions, we have the following lemma, which is required to prove our CAT-CFR convergence guarantee in~\cref{thm:catcfrregret}.
\begin{lemma}
\label{lem:catcfr-feedback}
For every $i\in\calN$, $I\in\calI_i$, and $a\in A_{I,i}$,
\(
\tilde u_{I,i}^t(a)\in[-1,1]\) and \(\mathbb{E}\left[\tilde u_{I,i}^t(a)\mid\calF_{t-1}\right]=u_{I,i}^t(a)
\).
Conditional on $\calF_{t-1}$, the vector $\tilde u_{I,i}^t$ is independent of the availability set $S_{I,i}^t$, the current local behavioral strategy, and the auxiliary samples used by its CAT-FTPL update.
\end{lemma}

\begin{proof}


By perfect recall, neither the path preceding $I$ nor a path starting at $ha$, with $h\in I$, visits $I$ again. Therefore~\eqref{eq:catcfr-sampled-value} does not use the strategy at $I$ or the set $S_{I,i}^t$. Conditional on $\calF_{t-1}$, it depends only on the action availabilities at other infosets. Drawing these action availabilities is independent of $S_{I,i}^t$ and of the auxiliary randomness induced by CAT-FTPL.

Equation~\eqref{eq:catcfr-sampled-value} can be expanded as a sum over terminal histories. Each term in this summation is a terminal payoff multiplied by chance probabilities and sampled action probabilities. Each infoset occurs at most once on a terminal history, so independence allows the conditional expectation of this product to be computed by replacing each sampled action probability by its marginal in~\eqref{eq:catcfr-local-marginal}. Note that by~\cref{assumption:productefg}, terminal payoffs do not depend on availability. The resulting sum is hence precisely $u_{I,i}^t(a)$.

Finally, terminal payoffs lie in $[-1,1]$, so the same is true of $V_i^t(h)$. By perfect recall, all histories in $I$ share the same preceding sequence of player $i$'s actions. Fixing these actions with probability one in the base game makes $\mathbb{P}_{-i}^t(h)$ the probability of reaching $h$. The histories in $I$ are mutually exclusive, so $\sum_{h\in I}\mathbb{P}_{-i}^t(h)\leq1$. This proves $|\tilde u_{I,i}^t(a)|\leq1$.
\end{proof}

\begin{algorithm}[ht]
\caption{CAT-CFR}
\label{alg:catcfr}
\begin{algorithmic}[1]
\renewcommand{\algorithmicrequire}{\textbf{Input:}}
\Require One ranking regret minimizer $\calR^{\mathsf{RANK}}_{I,i}$ per infoset
\Function{NextStrategy}{$S_{1}^t, S_{2}^t$}
    \For{each $i\in\{1,2\}$}
        \State $x_i^t\gets\mathbf{0}\in\mathbb{R}^{|\Sigma_i|}$; $x_i^t[\varnothing]\gets1$
        \For{each infoset $I\in\calI_i$ in top-down order}
            \State $b_{I,i}^t(\cdot\mid S_{I,i}^t)\gets
            \calR^{\mathsf{RANK}}_{I,i}$.\textsc{NextStrategy}($S_{I,i}^t$)
            \For{each $a\in A_{I,i}$}
                \State $x_i^t[(I,a)]\gets x_i^t[\sigma(I)]\cdot b_{I,i}^t(a\mid S_{I,i}^t)$
            \EndFor
        \EndFor
    \EndFor
    \State \Return $(x_1^t,x_2^t)$
\EndFunction
\Statex
\Function{ObserveUtility}
{$(\tilde u_{I,i}^t(a))_{i,I,a\in S_{I,i}^t}$}
    \For{each $i\in\{1,2\}$ and $I\in\calI_i$}
        \State $\calR^{\mathsf{RANK}}_{I,i}$.\textsc{ObserveUtility}
        ($S_{I,i}^t,
        (\tilde u_{I,i}^t(a))_{a\in S_{I,i}^t}$)
    \EndFor
\EndFunction
\end{algorithmic}
\end{algorithm}

\section{Omitted Proofs and Details from Section~\ref{sec:efg_repr}}\label{appsecs:efgproofs4}

\subsection{Ex-Interim Expansion Details}
\label{appsubsec:exinterimexpansion}
Recall that in an ex-interim EFGSAS \(\calG\), for each information set \(I \in \calI_i\) assigned to player \(i\), the action set available at that infoset is given by \(A_{I,i}\). Moreover, \(\rho_{I,i}\in \Delta(\calS_{I,i})\) is the probability distribution of observing each availability subset \(S_{I,i}\in\calS_{I,i}\subseteq2^{A_{I,i}}\backslash\{\varnothing\}\).
Any ex-interim EFGSAS \(\calG\) meeting~\cref{assumption:productefg} can be na\"ively expanded into an EFG \(\calG^\dagger\) using the following procedure in a top-down traversal:
\begin{enumerate}
	\item For an infoset $I\in\calI_i$, insert a chance node with children encoding all subsets $\calS_{I,i}$ before each decision node. Each edge is associated with the corresponding probability given by $\rho_{I,i}\in\Delta(\calS_{I,i})$.
	\item Each child of the chance node then incorporates the realizations of availability sets $S_{I,i}\in\calS_{I,i}$. In particular, each action which is not available is removed from the corresponding action set according to the realized availability set $S_{I,i}$.
    For all nodes \(h\) within the subtree of \(S_{I,i}\) and are a decision point of \(i\), move from infoset \(I_h\) to a new infoset \((I_h,S_{I,i})\) (i.e., preserving perfect recall). All nodes within the subtree of \(S_{I,i}\) that are not decision points of \(i\) stay in the same infoset (i.e., do not observe \(S_{I,i}\)).
\end{enumerate}


\subsection{Ex-Interim Expansion Size Analysis}
\label{appsubsec:expansion_size}
We analyze the size of a player’s decision problem in an EFGSAS. Consider the case where the base game tree is a full \(a\)-ary tree with \(d\) decision layers, where \(a\geq 2\) and \(d\geq 1\). Every nonterminal node belongs to player \(i\) and forms a singleton infoset, and all leaves occur after exactly \(d\) decisions. At each decision point, every non-empty subset of the \(a\) actions has positive availability probability, independently across infosets.

In the ex-interim expansion, each decision point is preceded by a chance node enumerating all non-empty subsets of actions. If a child of the chance node has \(k\) available actions, that decision point contributes \(k\) action sequences. Thus, the number of pairs generated at an infoset \(I\) is
\(\sum_{S\subseteq A_{I,i}} \lvert S \rvert = a2^{a-1}\) where $S\ne\varnothing$ and $|A_{I,i}|=a$.

Performing this expansion at every decision point implies that the \(l\)-th decision layer contributes \((a2^{a-1})^l\) sequences. Thus, for this family, the total number of sequences satisfies
\begin{align}
    \lvert \Sigma^\dagger_i \rvert &= 1+\sum_{l=1}^d (a2^{a-1})^l= 1+\frac{(a{2}^{a-1})^{d+1}-a 2^{a-1}}{a{2}^{a-1} - 1}\geq (a 2^{a-1})^{d}.
\end{align}

In particular, when \(d=1\), the base game has \(a+1\) nodes, whereas the expanded sequence set has \(1+a2^{a-1}\) elements, demonstrating a worst-case exponential increase in representation size.

\subsection{Proof of Proposition~\ref{prop:decisionmerging}}
We use an exchange argument, showing we can incrementally modify any strategy \(\beta_i\) to one which, at all infosets that correspond to observing the same \(S_{I,i}\), plays the same distribution.
In a bottom-up traversal in the EFGSAS game tree of player \(i\)'s infosets, for an infoset \(I \in \calI_i\) and availability set \(S_{I,i}\), we consider the set \(M(I,S_{I,i}) = \{I^\dagger \in \calI^\dagger_i \mid I^\dagger \text{ an expanded infoset for observing } S_{I,i} \text{ at } I \in \calI_i\}\).
We  construct a new behavioral strategy \(\beta'_i\), which (i) at an infoset \(I^\dagger \in M(I,S_{I,i})\) outputs the average strategy played by \(\beta_i\) across \(M(I,S_{I,i})\), and (ii) at all other infosets plays identically to \(\beta_i\).
We show for any joint strategy \(\beta_{-i}\), \(U_i(\beta_i, \beta_{-i})=U_i(\beta'_i, \beta_{-i})\).
Repeatedly performing this procedure over all the infosets gives us the desired strategy from which we can extract \(b_i\).



We start by making a key definition and proving a key lemma.
We let \(M(I,S_{I,i})\) be the infosets in \(\calI^\dagger\) which represent observing \(S_{I,i}\) at \(I\). Similarly, let \(M^\dagger(I^\dagger)\) be the infosets that represent observing the same availability set at the same infoset as \(I^\dagger\).
Formally, for an infoset \(I \in \calI_i\) and an availability set \(S_{I,i} \in \calS_{I,i}\), let
\begin{align}
    M(I,S_{I,i}) &\coloneqq \{I^\dagger \in {\calI^\dagger}_i \mid I^\dagger \text{ an infoset for observing } S_{I,i} \text{ at } I \in \calI_i\},\\
    M^\dagger(I^\dagger) &\coloneqq M(I, \mathrm{obs}(I^\dagger)),
\end{align}
where \(\mathrm{obs}(I^\dagger)\) is the availability set \(I^\dagger\) corresponds to observing.

\begin{lemma}
    \label{lemma:merging_exchange}
    Let \(\beta_i\) be a behavioral strategy for \(\calG^\dagger\),
    \(I \in \calI_i\) an infoset in \(\calG\) and \(S_{I,i} \in \calS_{I,i}\). Let \(I^\dagger \in \calI^\dagger_i\) be an infoset in \(\calG^\dagger\) corresponding to observing \(S_{I,i}\) at \(I\).
    Suppose that for every strict descendant $J$ of $I$ belonging to player $i$ and every availability set $S_{J,i}\in\calS_{J,i}$, the strategy is constant on $M(J, S_{J,i})$, that is,
    \begin{equation}\label{eqn:merging_exchange_lemma_condition}
        \forall J \text{ strictly succeeding } I, \forall S_{J,i}\in\calS_{J,i}, \forall J^\dagger, \tilde{J}^\dagger \in M(J,S_{J,i}), \quad \beta_i(J^\dagger) = \beta_i(\tilde{J}^\dagger).
    \end{equation}
    
    Then, we can construct a new behavioral strategy \(\beta'_i\) for \(\calG^\dagger\) that coincides with $\beta_i$ outside of $M(I,S_{I,i})$ and is constant on $M(I,S_{I,i})$ 
    such that
    \begin{equation}
    \forall \beta_{-i},\quad U_i(\beta_i, \beta_{-i}) = U_i(\beta'_i, \beta_{-i}).\end{equation}
\end{lemma}
\begin{proof}
Let \(Q_i(I^\dagger)\) be the set of action availabilities observed from the root up to and including infoset \(I^\dagger\), and let $\mathrm{Pred}_i(I^\dagger)$ denote the infoset-action pairs for player $i$ on the path to $I^\dagger$ strictly before the action at $I^\dagger$. We define the contribution from player $i$'s earlier actions and availabilities before choosing the current action as:
\begin{equation}p_i(I^\dagger)\coloneqq\prod_{(J^\dagger,a')\in\operatorname{Pred}_i(I^\dagger)}\beta_i(a'\mid J^\dagger) \prod_{(J,S_{J,i})\in Q_i(I^\dagger)}\rho_{J,i}(S_{J,i}).
\end{equation}


Next, let  $P_M = \sum_{I^\dagger\in M(I,S_{I,i})} p_i(I^\dagger)$. In particular, $P_M$ is independent of the action $a$ at $I$. For $P_M > 0$, let
\begin{align}
\hat{b}_i(a) &\coloneqq \frac{1}{P_M}\sum_{I^\dagger \in M(I, S_{I,i})} \beta_i(a \mid I^\dagger) p_i( I^\dagger), \quad a\in A_{I,i}. 
\end{align}

For $P_M=0$, select any $\hat{b}_i \in \Delta(S_{I,i})$. Moreover, note that for $P_M>0$, $\hat{b}_i \in \Delta(S_{I,i})$ as well, since $\frac{1}{P_M}\sum_{I^\dagger \in M(I, S_{I,i})}  p_i( I^\dagger)  \sum_{a\in S_{I,i}} \beta_i(a \mid I^\dagger)= 1$.
Using the above, we can define
\begin{align}
\beta'_i({I^\dagger}')&\coloneqq \begin{cases} [\hat{b}_i(a_1), \ldots, \hat{b}_i(a_{\lvert A_{I,i} \rvert})] & \text{ if } {I^\dagger}' \in M(I,S_{I,i}),\label{eqn:merging_beta'_def}\\
\beta_i({I^\dagger}') & \text { otherwise}.\end{cases}
\end{align}


For any joint strategy \(\beta_{-i}\), with \(\mathrm{hist}(z)\) the set of nonterminal nodes on the path from the root to node \(z\) (which is unique from perfect recall), with \(a_{h,z}\) the unique action at \(h\) that leads towards node \(z\), we have:
\begin{align}
    U_i(\beta_i, \beta_{-i}) &= \sum_{z \in \calZ^\dagger} u_i(z) \prod_{h \in \mathrm{hist}(z)}\beta_{\mathrm{player}(h)}(a_{h,z} \mid I_h) \\
    &= \sum_{z \in \calZ^\dagger} u_i(z) \left[ \prod_{\substack{h \in \mathrm{hist}(z)\\\mathrm{player}(h) \neq i}}\beta_{\mathrm{player}(h)}(a_{h,z} \mid I_h)
    \prod_{\substack{h \in \mathrm{hist}(z)\\\mathrm{player}(h) = i}}\beta_i (a_{h,z} \mid I_h) \right].\label{eqn:merging_lemma_finalpt1}
\end{align}

We split \(\calZ^\dagger\) into two disjoint sets \(\calZ^\dagger_M\) with terminals whose history includes a node whose infoset lies in \(M(I, S_{I, i})\) and \(\calZ^\dagger_{\neg M}\) with their history never intersecting \(M(I,S_{I,i})\). Formally,
\begin{align*}
    \calZ^\dagger_M &\coloneqq \{z \in \calZ^\dagger \mid \exists h \in \mathrm{hist}(z) \text{ s.t. } I^\dagger_h \in M(I,S_{I,i})\},\\
    \calZ^\dagger_{\neg M} &\coloneqq \{z \in \calZ^\dagger \mid \forall h \in \mathrm{hist}(z),  I^\dagger_h \notin M(I,S_{I,i})\} =\calZ^\dagger \setminus \calZ^\dagger_M.
\end{align*}

As these sets are disjoint and include all terminal nodes,

\begin{align}
    U_i(\beta_i, \beta_{-i}) =\sum_{z \in \calZ^\dagger_M} &u_i(z) \left[ \prod_{\substack{h \in \mathrm{hist}(z) \\ \mathrm{player}(h) \neq i}} \beta_{\mathrm{player}(h)}(a_{h,z} \mid I_h) \prod_{\substack{h \in \mathrm{hist}(z)\\\mathrm{player}(h) = i}}\beta_i (a_{h,z} \mid I_h)\right]\notag\\
    +\sum_{z \in \calZ^\dagger_{\neg M}} &u_i(z) \left[ \prod_{\substack{h \in \mathrm{hist}(z)\\ \mathrm{player}(h) \neq i}}\beta_{\mathrm{player}(h)}(a_{h,z} \mid I_h) \prod_{\substack{h \in \mathrm{hist}(z)\\\mathrm{player}(h) = i}}\beta_i (a_{h,z} \mid I_h) \right]. \label{eqn:merging_split_eqn}
\end{align}

We analyze these two sums individually. For the sum over \(\calZ^\dagger_{\neg M}\), by construction of \(\beta'_i\) and \(\calZ^\dagger_{\neg M}\), \(\beta'_i\) plays identically to \(\beta_i\), i.e., we have \(\forall z \in \calZ^\dagger_{\neg M}  \forall h \in \mathrm{hist}(z) \text { with } \textrm{player}(h)=i, \beta_i(I_h) = \beta'_i(I_h)\). Therefore,

\begin{align}
\sum_{z \in \calZ^\dagger_{\neg M}} &u_i(z) \left[ \prod_{\substack{h \in \mathrm{hist}(z)\\ \mathrm{player}(h) \neq i}}\beta_{\mathrm{player}(h)}(a_{h,z} \mid I_h) \prod_{\substack{h \in \mathrm{hist}(z)\\\mathrm{player}(h) = i}}\beta_i(a_{h,z} \mid I_h)\right]\notag\\
= \sum_{z \in \calZ^\dagger_{\neg M}} &u_i(z) \left[ \prod_{\substack{h \in \mathrm{hist}(z)\\ \mathrm{player}(h) \neq i}}\beta_{\mathrm{player}(h)}(a_{h,z} \mid I_h) \prod_{\substack{h \in \mathrm{hist}(z)\\\mathrm{player}(h) = i}}\beta'_i(a_{h,z} \mid I_h)\right].\label{eqn:20}
\end{align}

For the sum over $\calZ_M^\dagger$, we need to group terminal histories that differ only in the earlier action availability realizations of player $i$ that distinguish the infosets in $M(I,S_{I,i})$. We require some additional notation. Let $\pi_{\calZ}:\calZ^\dagger\to\calZ$ map each expanded terminal history to its corresponding terminal history in the original EFG. For $z\in\calZ_M^\dagger$, let
$\eta_I(z)$ encode all action availability realizations on the path to $z$, except those belonging to player $i$ and observed strictly before $I$. For each $\hat z\in\calZ$, define
\(
\calE_I(\hat z) \coloneqq
\left\{\eta_I(z): z\in\calZ_M^\dagger,\ \pi_{\calZ}(z)=\hat z \right\}
\),
and for every $\eta\in\calE_I(\hat z)$, define
\(
\calZ_M^\dagger(\hat z,\eta)\coloneqq\left\{z\in\calZ_M^\dagger:\pi_{\calZ}(z)=\hat z,\ \eta_I(z)=\eta\right\}
\).

Note that $\calZ_M^\dagger=\bigsqcup_{\hat z\in\calZ} \bigsqcup_{\eta\in\calE_I(\hat z)} \calZ_M^\dagger(\hat z,\eta)$.

For each $z\in\mathcal Z_M^\dagger$, let $h_I(z)$ be the unique node in $\mathrm{hist}(z)$ such that $I_{h_I(z)}^\dagger\in M(I,S_{I,i})$, and let $I^\dagger(z)\coloneqq I_{h_I(z)}^\dagger$. We use $h\prec g$ to denote strict ancestry in the game tree, and use $h\preceq g$ for ancestry allowing equality. 
We also split the product further into whether the node is an action availability chance node, and whether the nodes for player \(i\) are a child of \(I^\dagger\) or not. To this end, let \(H^\dagger_{\rho_i}\) be the set of action availability chance nodes for player \(i\) and let \(\beta_c\) denote the behavioral strategy at chance nodes. Then, the sum over $\calZ_M^\dagger$ can be written as:

\begin{align}\label{eqn:33}
    \sum_{\hat{z} \in \calZ}u_i(\hat{z}) \sum_{\eta \in \calE_I(\hat z)} \sum_{z \in \calZ^\dagger_M(\hat{z},\eta)} \Bigg[
    \prod_{\substack{h \in \mathrm{hist}(z) \setminus H^\dagger_{\rho_i}\\ \mathrm{player}(h) \neq i}} \beta_{\mathrm{player}(h)}(a_{h,z} \mid I_h)
    \prod_{\substack{h \in \mathrm{hist}(z)\\\mathrm{player}(h) = i\\h_I(z) \prec h}}&\beta_i (a_{h,z} \mid I_h) \notag\\\prod_{\substack{h \in \mathrm{hist}(z) \cap H^\dagger_{\rho_i}\\}}\beta_c (a_{h,z} \mid I_h)
    \prod_{\substack{h \in \mathrm{hist}(z)\\\mathrm{player}(h) = i\\h \preceq h_I(z)}}&\beta_i (a_{h,z} \mid I_h)\Bigg].
\end{align}

For a fixed $\hat z\in\mathcal Z$ and $\eta \in \calE_I(\hat z)$, the terminal histories in $\calZ_M^\dagger(\hat z,\eta)$ are equivalent, except for player $i$'s observed action availabilities (strictly) before $I$.

As all other players' infosets are such that they cannot distinguish between player \(i\)'s received availability sets, \[\prod_{\substack{h \in \mathrm{hist}(z) \setminus H^\dagger_{\rho_i} \\ \mathrm{player}(h) \neq i}} \beta_{\mathrm{player}(h)}(a_{h,z} \mid I_h)\] is equal for all \(z \in \calZ^\dagger_M(\hat{z},\eta)\).
Similarly, by the condition in~\cref{eqn:merging_exchange_lemma_condition}, 
\[\prod_{\substack{h \in \mathrm{hist}(z)\\\mathrm{player}(h) = i\\h_I(z) \prec h}}\beta_i (a_{h,z} \mid I_h)\]
is also equal for all \(z \in \calZ^\dagger_M(\hat{z},\eta)\).

Thus, for each $(\hat z,\eta)$, any choice of a representative $z'=z'(\hat z,\eta)\in\mathcal Z_M^\dagger(\hat z,\eta)$ then gives, by~\cref{eqn:33},


\begin{align}
    = \sum_{\hat{z} \in \calZ}u_i(\hat{z})  \sum_{\eta \in \calE_I(\hat z)}\Bigg(\prod_{\substack{h \in \mathrm{hist}(z') \setminus H^\dagger_{\rho_i} \\ \mathrm{player}(h) \neq i}} \beta_{\mathrm{player}(h)}(a_{h,z'} \mid I_h)
    \prod_{\substack{h \in \mathrm{hist}(z')\\\mathrm{player}(h) = i\\h_I(z')\prec h}}&\beta_i (a_{h,z'} \mid I_h)\notag\\ 
    \sum_{z \in \calZ^\dagger_M(\hat{z},\eta)}\Bigg[
    \prod_{h \in \mathrm{hist}(z) \cap H^\dagger_{\rho_i}}\beta_c (a_{h,z} \mid I_h)
    \prod_{\substack{h \in \mathrm{hist}(z)\\\mathrm{player}(h) = i\\h \preceq h_I(z)}}&\beta_i (a_{h,z} \mid I_h)\Bigg]\Bigg).
\end{align}

Fix $(\hat{z},\eta)$ in the equality above and let $a_{\hat{z}}$ be the action at $I$ on $\hat{z}$. All availability realizations to player $i$ strictly after $I$ in the inner sum above are fixed by $\eta$. Let us denote their product by $c(\hat{z},\eta)$. 

For each $z\in\mathcal Z_M^\dagger(\hat z,\eta)$,
the action availability probability factors (belonging to player $i$) through the current disclosure, along with player $i$'s behavioral strategy probabilities strictly before $h_I(z)$, contribute exactly $p_i(I^\dagger(z))$. For fixed $(\hat z,\eta)$, the mapping $z\mapsto I^\dagger(z)$ is thus a bijection from $\mathcal Z_M^\dagger(\hat z,\eta)$
onto $M(I,S_{I,i})$.
Hence, the inner sum is equal to $c(\hat{z},\eta) \sum_{I^\dagger \in M(I,S_{I,i})} p_i(I^\dagger) \beta_i(a_{\hat{z}} | I^\dagger) = c(\hat{z},\eta) P_M \hat{b}_i(a_{\hat{z}})$.

Under $\beta'_i$, $p_i(I^\dagger)$ are unchanged since the latest action is not taken into account, so the inner sum again gives $c(\hat{z},\eta) P_M \hat{b}_i(a_{\hat{z}})$. Note that this holds even when $P_M=0$, since both sums vanish. Hence, due to this equivalence we can write


\begin{align}
    &= \begin{aligned}[t]\sum_{\hat{z} \in \calZ}u_i(\hat{z}) \sum_{\eta \in \calE_I(\hat z)}\Bigg(\prod_{\substack{h \in \mathrm{hist}(z') \setminus H^\dagger_{\rho_i} \\ \mathrm{player}(h) \neq i}} \beta_{\mathrm{player}(h)}(a_{h,z'} \mid I_h)
    \prod_{\substack{h \in \mathrm{hist}(z')\\\mathrm{player}(h) = i\\h_I(z') \prec h}}&\beta'_i (a_{h,z'} \mid I_h)\\ 
    \sum_{z \in \calZ^\dagger_M(\hat{z},\eta)}\Bigg[
    \prod_{h \in \mathrm{hist}(z) \cap H^\dagger_{\rho_i}}\beta_c (a_{h,z} \mid I_h)
    \prod_{\substack{h \in \mathrm{hist}(z)\\\mathrm{player}(h) = i\\h \preceq h_I(z)}}&\beta'_i (a_{h,z} \mid I_h)\Bigg]\Bigg)\end{aligned}\notag\\
&= \sum_{z \in {\calZ^\dagger_M}} u_i(z) \left[ \prod_{\substack{h \in \mathrm{hist}(z)\\ \mathrm{player}(h) \neq i}}\beta_{\mathrm{player}(h)}(a_{h,z} \mid I_h) \prod_{\substack{h \in \mathrm{hist}(z)\\\mathrm{player}(h) = i}}\beta'_i(a_{h,z} \mid I_h)\right].
\label{eqn:merging_lemma_finalpt2}
\end{align}

Substituting~\cref{eqn:20} and~\cref{eqn:merging_lemma_finalpt2} into~\cref{eqn:merging_split_eqn} gives
\begin{align}
    U_i(\beta_i, \beta_{-i}) &=\begin{aligned}[t]\sum_{z \in \calZ^\dagger_M} u_i(z) \left[ \prod_{\substack{h \in \mathrm{hist}(z) \\ \mathrm{player}(h) \neq i}} \beta_{\mathrm{player}(h)}(a_{h,z} \mid I_h) \prod_{\substack{h \in \mathrm{hist}(z)\\\mathrm{player}(h) = i}}\beta'_i (a_{h,z} \mid I_h)\right]\\ +\sum_{z \in \calZ^\dagger_{\neg M}} u_i(z) \left[ \prod_{\substack{h \in \mathrm{hist}(z)\\ \mathrm{player}(h) \neq i}}\beta_{\mathrm{player}(h)}(a_{h,z} \mid I_h) \prod_{\substack{h \in \mathrm{hist}(z)\\\mathrm{player}(h) = i}}\beta'_i (a_{h,z} \mid I_h) \right]\end{aligned}\notag\\
    &= \sum_{z \in \calZ^\dagger} u_i(z) \left[ \prod_{\substack{h \in \mathrm{hist}(z)\\ \mathrm{player}(h) \neq i}}\beta_{\mathrm{player}(h)}(a_{h,z} \mid I_h) \prod_{\substack{h \in \mathrm{hist}(z)\\\mathrm{player}(h) = i}}\beta'_i (a_{h,z} \mid I_h) \right]\\
    &= U_i(\beta'_i, \beta_{-i})
\end{align}
as required.
\end{proof}

We are now ready to give the proof of~\cref{prop:decisionmerging}.
\begin{proof}
Perform a bottom-up traversal of the EFGSAS decision tree, and for each infoset \(I \in \calI_i\) and for each availability set \(S_{I,i} \in \calS_{I,i}\), consider \(M(I, S_{I,i})\). Note that since we perform a bottom-up traversal, the condition in~\cref{eqn:merging_exchange_lemma_condition} is met.
If there exist infosets \(I^\dagger, {I^\dagger}' \in M(I, S_{I,i})\) such that \(I^\dagger \neq {I^\dagger}'\) and \(\beta_i(I^\dagger) \neq \beta_i({I^\dagger}')\), apply ~\cref{lemma:merging_exchange} to construct a new behavioral strategy \(\beta'_i\) with \(\forall \beta_{-i}, U_i(\beta_i, \beta_{-i}) = U_i(\beta'_i, \beta_{-i})\), and continue the traversal with \(\beta'_i\).

At the end of the bottom-up traversal, we have a behavioral strategy \(\beta'_i\) such that
\begin{equation}
\label{eqn:decisionmerging_beta'_equality}
    \forall I \in \calI_i,  \forall S_{I,i} \in \calS_{I,i},  \forall I^\dagger, {I^\dagger}' \in M(I,S_{I,i}),  \beta'_i(I^\dagger) = \beta'_i({I^\dagger}').
\end{equation}

For each $(I,S_{I,i})$ with $M(I,S_{I,i}) \neq \varnothing$, select an arbitrary element $M'(I,S_{I,i}) \in M(I,S_{I,i})$, and set $b_i(I,S_{I,i}) = \beta'_i(M'(I,S_{I,i}))$. Note that by~\cref{eqn:decisionmerging_beta'_equality}, any arbitrary choice of $M'$ results in the same expected payoff. Moreover, if $M(I,S_{I,i}) = \varnothing$, any choice of distribution supported on $S_{I,i}$ is valid. By construction, note that $\mathrm{supp} b_i(I,S_{I,i}) \subseteq S_{I,i}$.  Furthermore, by the transitive property of equality, which was maintained throughout the exchange argument, \(\forall \beta_{-i}, U_i(\beta_i, \beta_{-i}) = U_i(\beta'_i, \beta_{-i})\). Finally, we can define \(\hat{\beta}_i(I^\dagger) \coloneqq b_i(I,\mathrm{obs}(I^\dagger))\) for the original infoset $I$ corresponding to $I^\dagger$. As argued above, it holds that \(U_i(\beta_i, \beta_{-i}) = U_i(\hat{\beta}_i, \beta_{-i})\). The construction of \(\hat\beta_i\) is independent of the payoff function, so the same argument establishes payoff equality for every player \(j\in\mathcal N\). In the main text, we drop the input $I$ from $b_i(I,S_{I,i})$ for simplicity, writing instead $b_i(S_{I,i})$ since it is clear that the behavioral strategy corresponds to infoset $I$.

\end{proof}

\subsection{Proof of Proposition~\ref{prop:NE-equiv}}\label{appsec:ne-equivproof}
\begin{proof}
    $(\impliedby)$ Consider a strategy profile $\omega$ implemented by $\chi$.
    
    We have that $U_i(\omega) \geq \max_{\omega_i' \in \Omega_i} U_i(\omega_i', \omega_{-i}) - \epsilon$. Expanding the expression for expected payoff of $\chi$ and substituting the equivalent expected payoff for its corresponding behavioral strategy $b$, we get:
    \begin{align}
        U_i(\chi) &= \mathbb{E}_{S \sim \rho} \left[\sum_{z\in\mathcal{Z}} \mathbb{P}(z\mid b(S), r)\cdot u_i(z)\right]\\
    &=\sum_{z\in\mathcal Z}u_i(z)p_c(z)\prod_{j\in\mathcal N}p_j(z;\mu_j)\\
        &= U_i(\omega)\\ 
        &\ge \max_{\omega_i' \in \Omega_i} U_i(\omega_i', \omega_{-i}) - \epsilon\\
    &=\max_{b_i'}\sum_{z\in\mathcal Z}u_i(z)p_c(z)p_i(z;\mu_i'(b_i'))\prod_{j\ne i}p_j(z;\mu_j)-\epsilon\\
        &= \max_{\chi_i'} U_i(\chi_i', \chi_{-i}) - \epsilon,
    \end{align}
    where $p_c(z)$ is the product of base-game chance probabilities along $z$ and $p_i(z;\mu_i)$ is the product of player $i$'s marginal action probabilities along that path. Here we also utilized the fact that $\mu_{I,i}(a) = \mathbb{P}[a; \rho_{I,i}, b_i]$ and the payoff equivalence between EFGSAS behavioral and DAG realization vectors. The proof for the forward direction is similar.
    
\end{proof}

\subsection{Proof of Proposition~\ref{prop:minimaxEFGSAS}}\label{appsec:minimaxEFGSASproof}
\begin{proof}
Consider $(\chi_1^*,\chi_2^*)$ that implements $(\omega_1^*,\omega_2^*)$. By definition of implementable strategies and~\cref{prop:decisionmerging}, it follows that $U(\chi_1^*,\chi_2^*) = U(\omega_1^*,\omega_2^*)$ when $\chi$ implements $\omega$. Note that each local marginal set
\[
M_{I,i}:=\sum_{S\in\mathcal S_{I,i}}\rho_{I,i}(S)\Delta(S)
\]
is nonempty, compact, and convex, where distributions on \(S\) are extended by zero outside \(S\). By~\cref{def:implementablesfs}, \(\Omega_i\) consists precisely of the vectors satisfying
\(\omega_i[\varnothing]=1,\omega_i[(I,\cdot)]\in\omega_i[\sigma(I)]M_{I,i}\) for every \(I\in\mathcal I_i\).
These constraints define a nonempty, closed, convex set and imply that every coordinate lies in \([0,1]\). Thus \(\Omega_i\) is compact and convex. Since $U$ is bilinear, von Neumann’s minimax theorem applies on $\Omega_1\times\Omega_2$. Then, invoking~\cref{prop:NE-equiv} ensures that the statement holds.
    
\end{proof}

\subsection{Proof of Theorem~\ref{thm:efgsascompact}}\label{appsec:compactEFGSASderiv}
\begin{proof}
First, we consider an arbitrary infoset $I$ of the EFGSAS. Let $I$ have action set $A_{I,i}$ and action availability distribution $\rho_{I,i}$. Then, we will require the following lemma adapted from~\cite{schwarz2026computing}, which establishes that, under Assumption~\ref{assumption:product}, the behavioral strategy at the infoset can be approximately represented compactly. In particular, the lemma was originally written for a mixed strategy $\pi : S\to\Delta(A)$ in the normal-form GSAS setting. Here, $\mu$ denotes implementable strategies in the normal-form sense (as given in~\cref{def:implementable}).

\begin{restatable}{lemma}{gsasCompact}
 Consider a normal-form GSAS $\calG = (\calG^{\mathrm{orig}}, \calS, \rho)$. Let $\pi_i$ be a normal-form mixed strategy that implements $\mu_i \in M_i$. 
    Then, \(\forall \epsilon \in (0,1)\),
    there exists 
    $w_i \in \Delta(A_i)\cap\mathbb{R}^{|A_i|}_{>0}$
    prescribing a strategy \(\pi_i'\) where
		$\pi'_i(a_i \mid S_i) 
		\coloneqq \frac{w_i(a_i)\mathds{1}\{a_i \in S_i\}}
		{\sum_{a'_i \in A_i} w_i(a'_i)\mathds{1}\{a'_i \in S_i\}}$.
 Furthermore, \(\pi_i'\) induces a marginal \(\mu'_i\) with total variation distance \(d_{\mathrm{TV}}(\mu_i, \mu'_i) < \epsilon\).

    \label{lem:compact}
\end{restatable}
    We will adapt the above in order to show the existence of an approximate compact implementable sequence-form strategy in EFGSAS satisfying~\cref{assumption:productefg}.
    First, consider $\epsilon \coloneqq \delta/|\Sigma_i^\mathrm{orig}|^2 \in (0,1)$. Let $b_i$ be the EFGSAS behavioral strategy associated with $\chi_i$. For every original information set $I\in\calI_i$, let $A_{I,i}$ be its action set and define the local marginal behavioral strategy
\begin{equation}
    \mu_{I,i}(a)
    \coloneqq
    \sum_{S\in\mathcal S_{I,i}}
    \rho_{I,i}(S)b_{I,i}(a\mid S),
    \quad a\in A_{I,i}.
    \label{eq:local-original-marginal}
\end{equation}

Since $\chi_i$ implements $\omega_i$, the collection $(\mu_{I,i})_{I\in\mathcal I_i}$ induces the sequence-form strategy $\omega_i$.
%
Crucially, the existence of an implementable behavioral strategy over all infosets holds because at each infoset, by Lemma~\ref{lem:compact} there exists an EFGSAS behavioral strategy $b'_{i}$ that approximately implements $\omega_{i}$ and admits compact $w_{I,i}$ so that
    \begin{equation}
        b'_{i} (a \mid S_{I,i}) = \frac{w_{I,i}(a) \mathds{1}\{a\in S_{I,i}\}}{\sum_{a'\in A_{I,i}} w_{I,i}(a') \mathds{1}\{a'\in S_{I,i}\}}.
    \end{equation}

Let
\[
    \mu'_{I,i}(a)\coloneqq    \sum_{S\in\calS_{I,i}}
    \rho_{I,i}(S)b'_{I,i}(a\mid S)
\]
be the induced marginal associated with $b'_i$. Then, it holds that
\begin{equation}
    d_{\mathrm{TV}}(\mu_{I,i},\mu'_{I,i})
    <\varepsilon
    \quad
    \text{for every }I\in\mathcal I_i.
    \label{eq:local-tv-error}
\end{equation}

The existence of an implementable sequence-form strategy follows trivially from~\cite{von1996efficient}, since the EFGSAS has perfect recall. We construct a compact weight vector encoding $b_i'$ as follows. In topological order over infosets belonging to player $i$, consider
\begin{equation}
    W_i[\varnothing]\coloneqq1,
    \quad
    W_i[(I,a)] \coloneqq     W_i[\sigma_i(I)]w_{I,i}(a),
    \label{eq:global-compact-vector}
\end{equation}
where $\sigma_i(I)$ is the parent sequence of $I$. 

Given an availability realization $S_i=(S_{I,i})_{I\in\mathcal I_i}$, $\chi'_i(S_i)$ can be constructed in topological order by setting
$\chi'_i[\varnothing\mid S_i]=1$. If
$\chi'_i[\sigma_i(I)\mid S_i]=0$, set
$\chi'_i[(I,a)\mid S_i]=0$ for every $a\in A_{I,i}$. Otherwise, set
\begin{equation}
    \chi'_i[(I,a)\mid S_i] = \chi'_i[\sigma_i(I)\mid S_i]
    \frac{
        W_i[(I,a)]\mathds{1}\{a\in S_{I,i}\}
    }{
        \sum_{a'\in A_{I,i}}
        W_i[(I,a')]\mathds{1}\{a'\in S_{I,i}\}
    }.
    \label{eq:global-compact-decoding}
\end{equation}

Since it holds that $\sum_{a\in A_{I,i}}\chi'_i[(I,a)\mid S_i] = \chi'_i[\sigma_i(I)\mid S_i]$, $\chi'_i(S_i)$ satisfies the sequence-form probability conservation constraints and is a valid realization vector.

Next, we bound the error in the induced marginal sequence-form strategy. Let $\omega'_i$ denote the marginal induced by $\chi'_i$, i.e.,
\[
\omega'_i[\sigma] \coloneqq \mathbb E_{S_i\sim\prod_{I\in\calI_i}\rho_{I,i}} \left[\chi'_i[\sigma\mid S_i]\right].
\]

Moreover, the lemma ensures that for any action $a\in A_{I,i}$, $|\mu_{I,i}(a) - \mu'_{I,i}(a) | \leq d_{\mathrm{TV}(\mu_{I,i}, \mu'_{I,i})} < \epsilon$. Then, for a non-empty sequence $\sigma$ we can let $(I_1,a_1), \dots, (I_d, a_d)$ denote its unique sequence of infoset-action pairs for player $i$. Consequently, 
$\omega_i[\sigma] = \prod_{r=1}^{d}\mu_{I_r,i}(a_r)$ (and analogously for $\omega'_i[\sigma]$), so we can bound 
\begin{align}
    \left|\omega_i[\sigma]-\omega'_i[\sigma]\right|
    &=
    \left|
        \prod_{r=1}^{d}\mu_{I_r,i}(a_r)
        -
        \prod_{r=1}^{d}\mu'_{I_r,i}(a_r)
    \right|                                                   \notag\\
    &\leq
    \sum_{r=1}^{d} \left(\prod_{s<r} \mu_{I_s,i}(a_s) \right)
    \left|
        \mu_{I_r,i}(a_r)-\mu'_{I_r,i}(a_r)
    \right|  \left(\prod_{s>r} \mu'_{I_s,i}(a_s) \right)                                                 \notag\\
    &\leq
    \sum_{r=1}^{d}
    \left|
        \mu_{I_r,i}(a_r)-\mu'_{I_r,i}(a_r)
    \right|                                                   \notag\\
    &<d\varepsilon.
    \label{eq:sequence-error}
\end{align}

The first inequality follows by telescoping and the triangle inequality; the second uses that all marginal probabilities lie in $[0,1]$. Moreover, $d$ is the number of player $i$’s infoset-action pairs along the unique sequence $\sigma$, giving $d\le |\Sigma^{\mathrm{orig}}_i|$. Hence, it follows that
\begin{align*}
    \|\omega_i-\omega'_i\|_1
    &=
    \sum_{\sigma\in
        \Sigma_i^{\mathrm{orig}}}
    \left|\omega_i[\sigma]-\omega'_i[\sigma]\right|\\
    &<
    \sum_{\sigma\in
        \Sigma_i^{\mathrm{orig}}}
    |\Sigma^{\mathrm{orig}}_i|\varepsilon\\
    &=
    |\Sigma^{\mathrm{orig}}_i|^2\varepsilon\\
    &
    =
    \delta.
\end{align*}

Thus $\chi'_i$ induces an implementable marginal sequence-form
strategy $\omega'_i$ satisfying
$\|\omega_i-\omega'_i\|_1<\delta$. The construction of \(\chi_i'(S_i)\) from \(W_i\) and \(S_i\) takes time linear in \(\sum_{I\in\mathcal I_i}|A_{I,i}|\), and the full procedure is given in~\cref{alg:renorm}.
\end{proof}

To this end, any player's sequence-form strategy in an EFGSAS can be (approximately) represented as a collection of vectors $w_{I,i}$, each of size $\vert A_{I,i}\vert$, from which the conditional realization vector for any given $S_i$ can be constructed in time linear in $\sum_{I\in\mathcal{I}_i}\vert A_{I,i}\vert$. However, for representational purposes, it is often more practical to work in the space of sequence-form strategies, and so we focus on the compact weight vectors $W_i$. This leads to a possibly exponential representation improvement compared to the expanded game  \(\calG^\dagger\), which is required to explicitly define strategies in $\calG$ (cf.~\cref{sec:infodisclosure}).


\subsection{Calculation of NE Compact Vectors in~\cref{example:lsg_compact}}\label{appsec:examplecompactderiv}
We focus on the setting \(\alpha=0\) (i.e., at the first decision point, \(\{H_1,T_1\}\) are both always available) and  \(\lambda \in [0,1)\). Let the attacker choose \(H_1\) at \(A\) with probability \(p\), and the defender choose \(H\) with probability \(q\). 

As \(H_1T_2\) and \(T_1H_2\) are strictly dominated by a mixture of strategies \(H_1H_2\) and \(T_1T_2\), the attacker will never choose these actions unless forced to by the stochastic action availability. Therefore, in the NE the defender wins either when (i) the defender chooses \(H\) and the attacker either also chooses \(H_1\) or chooses \(T_1\) and is forced at the second decision point to choose \(H_2\) or (ii) both players choose \(T\):
\begin{equation}
    \mathbb{P}[\text{defender wins}] = pq + q(1-p)\lambda + (1-p)(1-q).
\end{equation}

Given that the defender seeks to maximize the chance of intersection, and the attacker seeks to minimize it, the above can be written as a min-max problem, and applying the minimax theorem yields:
\begin{align}
    \max_q \min_p \left[pq + q(1-p)\lambda + (1-p)(1-q)\right]=\min_p \max_q \left[pq + q(1-p)\lambda + (1-p)(1-q)\right].
\end{align}



Solving this gives \(p^* = 1-\frac{1}{2-\lambda}\) and \(q^* =\frac{1}{2-\lambda}\), with defender equilibrium payoff of $v=2q^*-1 = 1-2p^*$.
Since the defender has no stochastic action availabilities, we can write their sequence-form NE strategy directly.

For the attacker, since \(\alpha=0\), the sequence-form entries for the first decision point can also be written directly. For the second decision point, set some $\epsilon\in(0,1)$ and let the behavioral strategies for $(H_2, T_2)$ at $B$ and $C$ be $(1-\epsilon, \epsilon)$ and $(\epsilon, 1-\epsilon)$ respectively. Hence, when both actions are available,~\cref{alg:renorm} selects the preferred action with probability $1-\epsilon$, and otherwise when only $H_2$ is available, it selects $H_2$ with probability 1. 
The remaining sequences can be set so as to conserve probability mass:
\begin{align}
    W^*_1 &= \left[1, \frac{1}{2-\lambda}, 1-\frac{1}{2-\lambda}\right], \\ W^*_2 &= \left[1, 1-\frac{1}{2-\lambda}, \frac{1}{2-\lambda}, (1-\epsilon)\cdot \left(1-\frac{1}{2-\lambda}\right), \epsilon\cdot\left(1-\frac{1}{2-\lambda}\right), \frac{\epsilon}{2-\lambda}, \frac{1-\epsilon}{2-\lambda}\right].
\end{align}

We next verify that these strategies played by~\cref{alg:renorm} indeed form an $\epsilon$-NE of the game. Let $\hat\chi_i$ denote the strategies obtained by applying~\cref{alg:renorm} to $W_i^*$ as defined above. When played against $\hat\chi_2$, the defender's payoffs from their actions are:
\begin{equation}
    U_1(H, \hat\chi_2) = 1-2 (1-p^*)(1-\lambda) (1-\epsilon) = v+ 2\epsilon p^*,
\end{equation}
and
\begin{equation}
    U_1(T, \hat\chi_2) = 1-2 p^* [1-(1-\lambda)\epsilon] = v+ 2\epsilon (1-\lambda) p^*,
\end{equation}
where we have used the fact that $(1-p^*)(1-\lambda) = p^*$. Thus, it holds that $v\leq U_1(\hat\chi_1,\hat\chi_2) \leq v+2\epsilon p^*$. The defender's maximal deviation therefore earns them at most $2\epsilon p^*$.
A similar argument holds for the attacker: indeed, their best-response payoff is $-v$. Since the game is zero-sum, the attacker's maximal payoff gain is upper bounded by $2\epsilon p^*$.

Finally, note that $p^*\leq 1/2$, so each player’s payoff gain from deviating is at most $\epsilon$, ensuring that the compact strategy $(W_1^*,W_2^*)$ is an $\epsilon$-NE. Concretely, setting \(\lambda = 0.5\) gives the compact vectors given in~\cref{example:lsg_compact}. Moreover, $U_1(\hat\chi_1,\hat\chi_2) = 1/3 + 5\epsilon/9$, so the defender and attacker can gain at most $\epsilon/9$ and $5\epsilon/9$ for deviating, respectively.

\subsection{On $k$-Ex-Interim EFGSAS}
\label{appsec:kexinterim}
Let $\calG^{\mathrm{orig}}$ be a finite perfect-recall EFG. 
Recall that $\sigma(I)$ denotes the parent sequence of an infoset $I\in\calI$. Moreover, for $I,J\in\calI_i$, we use $I\preceq_i J$ to denote if $J=I$ or if $\sigma(J)$ is a descendent of a sequence $(I,a)$ for some $a\in A_{I,i}$. Moreover, we call $J$ an immediate successor of $I$ if $\sigma(J) = (I,a)$.
Next, we define the notion of decision depth required for our analysis.

\begin{definition}[Decision depth]
    For each player $i$, let $\tau_i(I)$ denote the number of player $i$ decisions strictly before any history in $I$, which is well defined by perfect recall. Let $d_i \coloneqq 1 +\max_{I\in\calI_i} \tau_i(I)$ be the decision depth for player $i$. In particular, $\tau_i(I)=0$ when $\sigma(I)=\varnothing$ and $\tau_i(J)=\tau_i(I)+1$ whenever $J$ is an immediate successor of $I$.
\end{definition}

    
In the $k$-ex-interim setup, upon reaching an infoset $I\in\mathcal I_i$, player $i$ privately learns their action availabilities up to $k$ of their own decision points ahead.
For $0\le r \le k$, the set of infosets precisely $r$ decision points after an infoset $I$ is given by the set
\begin{equation}
    \calI^r_i(I) \coloneqq \{J\in\calI_i : I\preceq_i J, \tau_i(J) = \tau_i(I)+r\}, \quad \calI^0_i(I) \coloneqq \{I\}.
\end{equation}

Then, the set of infosets whose availabilities are revealed to the player in the $k$-step lookahead is $\calL_{i}^k (I) \coloneqq \bigcup_{r=0}^k \calI_i^r (I)$. Finally, the set of action availabilities associated to the revealed infosets starting from $I$ is $\calS_i^k(I) = \prod_{J\in\calL_i^k(I)} \calS_{J,i}$. Note, however, that there may be several initial infosets for a player $i$. To account for this, we introduce a ``virtual'' disclosure point before the base game starts, which is not a player decision point and serves only to reveal the action availabilities $k$ steps ahead. In particular, at the virtual point, the $k$-step lookahead reveals action availabilities at
\begin{equation}
    P_i^k \coloneqq \{J\in\calI_i : \tau_i(J) < k\}.
\end{equation}

In particular, $P_i^0=\varnothing$ and $P_i^k=\calI_i$ whenever $k\ge d_i$.
Using the above, we arrive at the formal definition of $k$-ex-interim EFGSAS.
\begin{definition}[$k$-ex-interim EFGSAS]
    Prior to game-play, player $i$ privately observes $(S_{J,i})_{J\in P_i^k}$. Upon reaching $I\in\calI_i$, player $i$ privately observes every component of $\mathbf{S}_{I,i}^k\coloneqq (S_{J,i})_{J\in\calL_i^k(I)}$ which has not yet been disclosed, remembering all prior disclosures.
    
\end{definition}
As described before, the disclosed action availability profile includes the availabilities at every infoset in 
$\calL_i^k(I)$,
including those that are \emph{not} reachable from the realized history at $I$. Hence, the disclosure corresponding to an infoset $I$ depends \emph{only} on $I$ and does not identify which of these future infosets is reachable from the realized history. This ensures that the possibly imperfect information of the player is maintained even in the $k$-ex-interim setting. 
Note that by construction, setting $k=0$ recovers the ex-interim disclosure regime, and setting $k\geq d_i$ recovers the ex-ante regime.

Let $\calG^{\dagger,k}$ denote the expanded EFG obtained as follows: when a player reaches infoset $I$, every component of $\mathbf{S}_{I,i}^k$ that has not been previously revealed is disclosed to the player. The corresponding chance probabilities (defined for the ex-interim case in~\cref{appsubsec:exinterimexpansion}) are now constructed according to the product of the relevant distributions $\rho_{J,i}$ for $J\in\calL_i^k(I)$. We denote the set of expanded infosets belonging to player $i$ by $\calI_i^{\dagger,k}$, and define a mapping $\phi : \calI_i^{\dagger,k} \to \calI_i$ from each expanded infoset to its original infoset. Finally, for $I^\dagger$ in $\calI_i^{\dagger,k}$, let $\mathrm{obs}_k(I^\dagger) \coloneqq \mathbf{S}^k_{\phi(I^\dagger),i}$, which intuitively is the set of availability sets in the expanded game that correspond to the observed action availabilities.


We now consider the representation size of the expanded treeplex and DAG-plex.

\paragraph{Treeplex.}
First, we need to define some additional notation to represent the size of each player's decision space in a $k$-ex-interim EFGSAS. Let $\calI_i(t) = \{I\in\calI_i:\tau_i(I)=t\}$ be the set of infosets at decision depth $t$ for player $i$ in the game tree, and $N_i \coloneqq |\{I\in\calI_i : \sigma_i(I)=\varnothing\}|$ be the number of initial infosets for player $i$. Let $\mathrm{succ}_i(I) \coloneqq \left\{J\in\calI_i: \sigma(J)=(I,a)\text{ for some }a\in A_{I,i}\right\}$, and let $B_i = \max_I |\mathrm{succ}_i(I)|$. Intuitively, this is the branching factor of player $i$'s decision problem: it captures the maximum number of possible subsequent infosets that follow a current infoset, across all current actions and opponent play/chance. 
In our subsequent analysis, we  will require the observation that 
$|\calI_i(t)|\leq N_i B_i^t$ for $0\le t <d_i$.
Next, we define 3 game-dependent constants, labeled $c_1$, $c_2$ and $c_3$ (for simplicity, we drop the player-specific notation, but each of the below is to be understood as belonging to player $i$).
\begin{equation}
    c_1 \coloneqq \max_{I\in\calI_i} |\calS_{I,i}|,\quad c_2 \coloneqq \max_{I\in\calI_i, a\in A_{I,i}} |\{S_{I,i}\in\calS_{I,i} : a\in S_{I,i}\}|,\quad c_3 \coloneqq \max_{I\in\calI_i} \sum_{S_{I,i}\in\calS_{I,i}}|S_{I,i}|.
\end{equation}

Intuitively, $c_1$ is the maximum number of possible availability sets at one infoset. $c_2$ is the maximum number of availability sets that are compatible with a given action. Suppose an infoset is reached using action $a$, then the availability set at the preceding infoset must have contained $a$. Then, at most $c_2$ availability realizations are compatible with this action. Finally, $c_3$ is the maximum number of local expanded sequences generated from one infoset. 

Next, assume that each infoset in the game has $a$ actions, and that each non-empty availability subset is possible. Then, it is possible to simplify the above constants to
\begin{equation}
    c_1 = 2^a-1, \quad c_2 = 2^{a-1}, \quad c_3 = a 2^{a-1}.
\end{equation}

The first constant counts the nonempty subsets of an action set with \(a\) elements. For any fixed action, the remaining \(a-1\) actions may be included or excluded without restriction, giving the second constant. Finally, there are \(2^a-1\) nonempty availability subsets, and each action belongs to \(2^{a-1}\) of them. Each subset \(S\) contributes \(|S|\) local expanded sequences, so the third constant \(c_3=\sum_{\varnothing\ne S\subseteq A}|S|=a2^{a-1}\) holds.

Next, we consider a particular $I\in \calI_i(t)$. By perfect recall, $I$ is uniquely determined by the sequence $I_0,\dots,I_t$. When $I$ is reached,
the availability sets indexed by
\begin{equation}
    D_i^k(I) \coloneqq P_i^k \cup \bigcup_{s=0}^t \calL_i^k(I_s)
\end{equation}
have been disclosed to player $i$. Our goal is to bound $|D_i^k(I)|$.

First, let $k^{\min}_i = \min\{k,d_i-1\}$. Then, the initial action availability disclosure contains at most $\sum_{r=0}^{k^{\min}_i} B_i^r$ availability sets. Then, in subsequent stages, only one new layer of action availabilities is revealed, containing at most $B^{k^{\min}_i}_i$ availability sets. Thus, up to some layer $t$ in the game, for a sequence beginning at \emph{one} initial infoset, we have
\begin{equation}
    M_{i}^{\mathrm{loc}}(t,k) \coloneqq \sum_{r=0}^{k^{\min}_i} B_i^r + B_i^{k^{\min}_i}\min\{t, d_i-1-k_i^{\min}\}.
\end{equation}

Here, the superscript $\mathrm{loc}$ denotes that the counting argument is for a local (i.e., from a single initial infoset) lookahead. We also need to include the contribution of the virtual disclosure for the initial infosets. This is given by
\begin{equation}
    |P_i^k| = \sum_{s=0}^{\min\{k-1,d_i-1\}} |\calI_i(s)|.
\end{equation}

Together, for every $I\in \calI_i(t)$, we can write
\begin{equation}
    M_i(t,k) \coloneqq \min\{|\calI_i|, |P_i^k| + M_i^\mathrm{loc}(t,k)\} \geq |D_i^k(I)|.
\end{equation}

Using the above, we have the following:
\begin{proposition}
    Consider a player $i$ in an EFGSAS. The size of their $k$-ex-interim treeplex as a function of $k$ is bounded by
    \[|\Sigma_i^{\dagger}(k)| \leq    N_i\sum_{t=0}^{d_i-1}B_i^t \cdot c_1^{M_i(t,k)-t-1} \cdot c_2^t \cdot c_3+1.\]
\end{proposition}

\begin{proof}
Consider an infoset $I\in\calI_i(t)$. In $\calG^\mathrm{orig}$, there are at most $N_i B_i^t$ infosets at layer $t$, and reaching any infoset fixes $t$ preceding player actions. Each such action is contained in at most $c_2$ availability sets, giving at most $c_2^t$ compatible realizations for the availability sets at its ancestors. Moreover, at the current infoset, the number of action sequences over all possible current infosets is at most $c_3$. Each of the remaining $M_i(t,k)-t-1$ observed sets admits at most $c_1$ realizations. Thus, at most $c_1^{M_i(t,k)-t-1} \cdot c_2^t \cdot c_3$ expanded sequences correspond to this $I$.

Next, we sum over the $|\calI_i(t)|$ original infosets at each layer $t$ (including the empty sequence). Then, it holds that 
\begin{equation}
    |\Sigma_i^{\dagger}(k)| \leq   \sum_{t=0}^{d_i-1} |\calI_i(t)| \cdot c_1^{M_i(t,k)-t-1} \cdot c_2^t \cdot c_3+1.
\end{equation}

Recalling that $|\calI_i(t)|\leq N_i B_i^t$ for $0\le t <d_i$, we obtain
\begin{equation}
    |\Sigma_i^{\dagger}(k)| \leq  N_i \sum_{t=0}^{d_i-1} B_i^t \cdot c_1^{M_i(t,k)-t-1} \cdot c_2^t \cdot c_3+1,
\end{equation}
as required. Note that if $k\geq d_i$, the virtual disclosure reveals the entire action availability profile, which gives $M_i(t,k) = |\calI_i|$, which recovers the ex-ante expansion argument.


\end{proof}

\paragraph{DAG-plex.}
The statement of Proposition~\ref{prop:decisionmerging} can be extended to the $k$-ex-interim setting. 
\begin{proposition}\label{prop:mergingkinterim}
   For every behavioral strategy $\beta_i$ in $\calG^{\dagger,k}$, there exists a mapping
   \[b_i: \bigcup_{I \in \calI_i}(\{I\} \times \calS_i^k(I)) \to \bigcup_{I \in \calI_i} \Delta(A_I)\]
   such that $\mathrm{supp}(b_i(I,\mathbf{S}_{I,i}^k)) \subseteq S_{I,i}$ and the behavioral strategy $\hat{\beta}_i(I^\dagger) \coloneqq b_i(\phi(I^\dagger), \mathrm{obs}_k(I^\dagger))$ satisfies $U_i(\beta_i,\beta_{-i}) = U_i(\hat\beta_i, \beta_{-i})$ for any $\beta_{-i}$.
\end{proposition}
\begin{proof}
We give an extension of the exchange argument made for the ex-interim case. First, we require an extended definition for the set of infosets that correspond to observing $\mathbf{S}^k_{I,i}$ at $I$:
\begin{equation}
    M_k(I,\mathbf{S}^k_{I,i}) \coloneqq \left\{ I^\dagger \in \calI_i^{\dagger,k} : \phi(I^\dagger) = I, \mathrm{obs}_{k}(I^\dagger) = \mathbf{S}^k_{I,i}\right\}.
\end{equation}

For an expanded infoset \(I^\dagger\in\mathcal I_i^{\dagger,k}\), we also define
\begin{equation}
    M_k^\dagger(I^\dagger) \coloneqq M_k\left(\phi(I^\dagger), \mathrm{obs}_{k}(I^\dagger)\right).
\end{equation}

Using these definitions, we prove the following lemma.

\begin{lemma}
\label{lemma:merging_exchange_k}
    Let \(\beta_i\) be a behavioral strategy of player $i$ in \(\calG^{\dagger,k}\), let $I\in\calI_i$ be an infoset in $\calG$ and let $\mathbf{S}_{I,i}^k\in\calS_i^k(I)$ be a feasible $k$-ex-interim observation profile at $I$. 
    Consider $I^\dagger\in M_k(I,\mathbf S^k_{I,i})$.
    If for any expanded infoset $J^\dagger \in \calI^{\dagger,k}_i$ that can be reached strictly after an infoset in $M_k^\dagger(I^\dagger)$ we have 
    \begin{equation}\label{eq:conditionk}
        \forall \tilde J^\dagger\in M^\dagger_k(J^\dagger), \quad \beta_i(J^\dagger) = \beta_i(\tilde J^\dagger),
    \end{equation}
    then we can construct a new behavioral strategy \(\beta'_i\) for \(\calG^{\dagger,k}\) such that
    \[\forall \beta_{-i},\quad U_i(\beta_i, \beta_{-i}) = U_i(\beta'_i, \beta_{-i}).\]
    
    Moreover, $\beta_i'=\beta_i$ outside $M_k(I,\mathbf S^k_{I,i})$, and $\beta_i'(I^\dagger)$ is constant over $I^\dagger\in M_k(I,\mathbf S^k_{I,i})$.
\end{lemma}
\begin{proof}[Proof of Lemma~\ref{lemma:merging_exchange_k}]
The proof of the lemma is similar to that of the ex-interim case, with a few key modifications. 
First, we define \(p_i^k(I^\dagger)\) as the contribution of player \(i\)'s earlier actions and availability set realizations before the action at \(I^\dagger\).
Let \(\mathrm{Pred}_i(I^\dagger)\) be the prior infoset-action pairs belonging to a player $i$ on the path to $I^\dagger$ but strictly before the action at $I^\dagger$, as defined in the proof of~\cref{lemma:merging_exchange}.
Let $Q_i(I^\dagger)$ be the set of action availabilities observed by player $i$ from the root to infoset $I^\dagger$. Note that since 
$\mathrm{obs}_{k}(I^\dagger)=\mathbf S_{I,i}^k=\left(S_{J,i}:J\in\calL_i^k(I)\right)$, 
we can define
\begin{equation}\label{eq:qminusdef}
    Q^-_{i,k} (I^\dagger; I) = \left\{(J, S_{J,i})\in Q_i(I^\dagger) : J \notin \calL_i^k(I) \right\}.
\end{equation}

Intuitively, $Q^-_{i,k} (I^\dagger; I)$ contains the previously observed action availabilities that could distinguish $I^\dagger$ from other elements of $M_k(I, \mathbf{S}_{I,i}^k)$, but are not components of the current $k$-step observation profile at $I$.

In analogy to the ex-interim case, define 

\begin{equation}
    p_i^k(I^\dagger)\coloneqq \prod_{(J^\dagger,a')\in\operatorname{Pred}_i(I^\dagger)}\beta_i(a'\mid J^\dagger) \prod_{(J,S_{J,i})\in Q^-_{i,k}(I^\dagger;I)}\rho_{J,i}(S_{J,i}).
\end{equation}

Let $P_M^k =\sum_{I^\dagger\in M_k(I,\mathbf S^k_{I,i})}p_i^k(I^\dagger)$.
If $P_M^k>0$, define $\hat b_i^k$ by
\begin{align}
\hat{b}_i^k(a) &\coloneqq \frac{1}{P^k_M}\sum_{I^\dagger \in M_k(I, \mathbf{S}_{I,i}^k)} \beta_i(a \mid I^\dagger) p^k_i(I^\dagger).
\end{align}

If $P_M^k=0$, as in the ex-interim proof, we can select any $\hat b_i^k\in\Delta(S_{I,i})$. Otherwise, $\sum_{a\in S_{I,i}}\hat b_i^k(a)=1$, so again $\hat b_i^k\in\Delta(S_{I,i})$. Using this, we can define
\begin{align}\label{eq:betaprimek}
    \beta_i'( K^\dagger):=
\begin{cases}\hat b_i^k,& \text{if }K^\dagger\in M_k(I,\mathbf S^k_{I,i}),\\
\beta_i( K^\dagger),&\text{otherwise}.\end{cases}
\end{align}


For any joint opponent strategy $\beta_{-i}$, let $\mathrm{hist}(z)$ be the unique set of nonterminal nodes on the path from the root to a terminal history $z$, and let $a_{h,z}$ be the action taken at $h$ on this path. Then, 
\begin{equation}
    U_i(\beta_i,\beta_{-i}) = \sum_{z\in\calZ^{\dagger,k}} u_i(z) \prod_{h\in\mathrm{hist}(z)} \beta_{\mathrm{player}(h)} (a_{h,z} \mid I_h).
\end{equation}

Splitting $\calZ^{\dagger,k}$ into two disjoint sets, we get
\begin{equation}
    \calZ_M^{\dagger,k} \coloneqq \left\{ z\in\calZ^{\dagger,k} : \exists h\in\mathrm{hist}(z) \text{ such that } I^\dagger_h\in M_k(I, \mathbf{S}_{I,i}^k)\right\},
\end{equation}
and
\begin{equation}
    \calZ_{\neg M}^{\dagger,k} = \calZ^{\dagger,k} \setminus \calZ_M^{\dagger,k}.
\end{equation}

First consider the sum over $\calZ_{\neg M}^{\dagger,k}$: by construction of $\beta'_i$, $\beta_i$ and $\beta'_i$ coincide at each of player $i$'s infosets on the path to $z$. Hence,
\begin{equation}\label{eq:neg-sum-k}
    \sum_{z\in \calZ_{\neg M}^{\dagger,k}} u_i(z) \Pr_{\beta_i,\beta_{-i}}(z) = \sum_{z\in \calZ_{\neg M}^{\dagger,k}} u_i(z) \Pr_{\beta'_i,\beta_{-i}}(z).
\end{equation}

Next, consider the sum over $\calZ_M^{\dagger,k}$. Again in analogy to the ex-interim case, we group terminal histories that differ only in the availability observations contained in $Q^-_{i,k} (I^\dagger; I)$. 
Let $\pi^k_\calZ : \calZ^{\dagger,k} \to \calZ$ be a mapping from each expanded terminal history to its corresponding terminal history in the original EFG. For $z\in\calZ_M^{\dagger,k}$, let $\eta_{I}^k(z)$ encode all action availabilities observed on the path to $z$ except those contained in $Q^-_{i,k} (I^\dagger; I)$, where $I^\dagger$ is the unique member of $M_k(I, \mathbf{S}_{I,i}^k)$ on the path to $z$. For each $\hat{z}\in\calZ$, we define
\begin{equation}
    \calE_{I}^k(\hat{z}) \coloneqq \left\{ \eta_{I}^k(z) : z\in\calZ_{M}^{\dagger,k} , \pi_\calZ^k(z) = \hat{z}\right\},
\end{equation}
and for each $\eta\in \calE_{I}^k(\hat{z})$ we consider
\begin{equation}
    \calZ_M^{\dagger,k} (\hat z,\eta) = \left\{ z\in\calZ_M^{\dagger,k} : \pi^k_\calZ(z) = \hat{z},\eta_{I}^k(z) = \eta\right\}.
\end{equation}

Note that by construction, $\calZ_M^{\dagger,k} = \bigsqcup_{\hat{z}\in\calZ} \bigsqcup_{\eta\in\calE_{I}^k(\hat{z})} \calZ_M^{\dagger,k} (\hat z,\eta)$. All terminal histories in $\calZ_M^{\dagger,k} (\hat z,\eta)$ correspond to the same original leaf node $\hat{z}$, the same action availability profile $\mathbf{S}_{I,i}^k$, and take the same action at $I$. As notational shorthand, we denote this action by $a_{\hat{z}}$. 

Next, we argue that the probabilities corresponding to shared observations can be combined.
Indeed, no availability set in $Q^-_{i,k}(I^\dagger;I)$ belongs to the $k$-step observation profile of a strict successor of $I$: any such descendant availability disclosed by $I$ already belongs to the profile at $I$ and was removed in the definition of $Q^-_{i,k}(I^\dagger;I)$ (cf.~\cref{eq:qminusdef}).
Moreover, for a fixed $(\hat z,\eta)$, each $I^\dagger\in M_k(I,\mathbf S^k_{I,i})$ specifies one compatible expanded path, and its earlier probability factors (belonging to player $i$) are precisely $p_i^k(I^\dagger)$.
Consider the players' behavioral strategies across all histories in $\calZ_M^{\dagger,k} (\hat z,\eta)$. Player $i$'s behavioral strategy strictly after $I$ is equivalent across the histories in $\calZ_M^{\dagger,k} (\hat z,\eta)$, while all other players' infosets are such that they cannot distinguish between player $i$'s observed availability sets. 

Hence, we can collect all remaining probabilities, including the fixed current and later availability sets probabilities, opponent action probabilities, and player $i$'s probability factors strictly after $I$. We use a notational shorthand $\gamma^k_{\hat z,\eta}$ to denote the collection of these factors. By~\cref{eq:conditionk} and the independence assumption, this factor is independent of $I^\dagger$ and is unchanged by~\cref{eq:betaprimek}. Thus, we have


\begin{equation}
    \sum_{z\in \calZ_M^{\dagger,k} (\hat z,\eta)} \Pr_{\beta_i,\beta_{-i}} (z) = \gamma^k_{\hat z,\eta} \sum_{I^{\dagger}\in M_k(I, \mathbf{S}_{I,i}^k)} p^k_i(I^\dagger) \beta_i(a_{\hat{z}} \mid I^\dagger).
\end{equation}

By definition of $\hat{b}^k_i$, if $P^k_M >0$, it holds that 
\begin{equation}
    \sum_{I^{\dagger}\in M_k(I, \mathbf{S}_{I,i}^k)} p^k_i(I^\dagger) \beta_i(a_{\hat{z}} \mid I^\dagger) = P^k_M \hat{b}^k_i(a_{\hat{z}}).
\end{equation}

Note that the above equality also holds when $P_M^k=0$, since both sides vanish.
Therefore, under Equation~\eqref{eq:conditionk} and for any $\beta_i'$ constructed in~\cref{eq:betaprimek}, we obtain 
\begin{align}
     \sum_{z\in \calZ_M^{\dagger,k} (\hat z,\eta)} \Pr_{\beta'_i,\beta_{-i}} (z) &= \gamma^k_{\hat z,\eta} \sum_{I^{\dagger}\in M_k(I, \mathbf{S}_{I,i}^k)} p^k_i(I^\dagger) \beta'_i(a_{\hat{z}} \mid I^\dagger)\\
     &= \gamma^k_{\hat z,\eta}  P^k_M \hat{b}^k_i(a_{\hat{z}})\\
     &= \sum_{z\in \calZ_M^{\dagger,k} (\hat z,\eta)} \Pr_{\beta_i,\beta_{-i}} (z) \label{eq:equalityk}
\end{align}

Multiplying Equation~\eqref{eq:equalityk} by $u_i(\hat{z})$, summing over all $\eta\in\calE_{I}^k(\hat{z})$ and $\hat{z}\in\calZ$, then combining with~\eqref{eq:neg-sum-k} gives $U_i(\beta_i,\beta_{-i}) = U_i(\beta'_i,\beta_{-i})$ for any $\beta_{-i}$, as required.
\end{proof}

Similarly to the ex-interim case, we perform a bottom-up traversal of player $i$'s original infosets. Starting with a given behavioral strategy $\beta_i$, for each $I\in\calI_i$ in the ordering, and for each $k$-ex-interim action availability observation profile $\mathbf{S}^k_{I,i}$, consider $M_k(I,\mathbf{S}^k_{I,i})$ as defined before. When infoset $I$ is reached, each $J\in\calI_i$ where $I$ precedes $J$ has already been processed. 
Thus, for every strict successor $J$ of $I$ and every $\mathbf S^k_{J,i}\in\calS_i^k(J)$, we have

\[\beta_i(J^\dagger)=\beta_i(\widetilde J^\dagger) \quad \forall J^\dagger,\widetilde J^\dagger\in M_k(J,\mathbf S^k_{J,i}).\]


If there exist some $I^\dagger$, $\tilde I^\dagger \in M_k(I,\mathbf S_{I,i}^k)$ such that $\beta_i(I^\dagger)\neq\beta_i(\tilde I^\dagger)$, then we can apply Lemma~\ref{lemma:merging_exchange_k}, producing a new behavioral strategy $\beta'_i$ satisfying 
\[
\beta'_i(I^\dagger) = \beta'_i(\tilde I^\dagger)\quad \forall I^\dagger,\tilde I^\dagger \in M_k(I,\mathbf{S}_{I,i}^k),
\]
and moreover, $\forall \beta_{-i}$, $U_i(\beta_i,\beta_{-i}) = U_i(\beta'_i,\beta_{-i})$. 

Performing this procedure over the whole tree results in a strategy $\bar \beta_i$ such that 
\begin{equation}
\label{eqn:decisionmerging_beta'_equalityk}
    \forall I\in\calI_i, \quad \forall \mathbf{S}_{I,i}^k\in \calS_{i}^k(I), \forall I^\dagger,\tilde I^\dagger\in M_k(I,\mathbf{S}_{I,i}^k),\quad  \bar\beta_i(I^\dagger) = \bar\beta_i(\tilde I^\dagger).
\end{equation}

Repeatedly applying the lemma ensures that $U_i(\beta_i,\beta_{-i}) = U_i(\bar\beta_i,\beta_{-i})$.
For each $(I,\mathbf S^k_{I,i})$ with $M_k(I,\mathbf S^k_{I,i})\neq\varnothing$, we select an arbitrary element \(M'(I,\mathbf S^k_{I,i})\in M_k(I,\mathbf S^k_{I,i})\). Using Equation~\eqref{eqn:decisionmerging_beta'_equalityk}, we can define $b^k_{I,i}(\mathbf{S}_{I,i}^k) \coloneqq \bar\beta_i(M'(I,\mathbf{S}_{I,i}^k))$.
When $M_k(I,\mathbf S^k_{I,i})=\varnothing$, we instead select $b^k_{I,i}(\mathbf S^k_{I,i})$ arbitrarily in $\Delta(S_{I,i})$.
Similarly to the ex-interim case, $\mathrm{supp}(b_{I,i}^k (\mathbf{S}_{I,i}^k)) \subseteq S_{I,i}$, where $S_{I,i}$ is the component of $\mathbf{S}_{I,i}^k$ corresponding to $I$. Finally, the induced strategy in the expanded game, denoted by $\hat{\beta}_i(I^\dagger) = b^k_{\phi(I^\dagger),i}(\mathrm{obs_k(I^\dagger)})$, satisfies $U_i(\beta_i,\beta_{-i}) = U_i(\hat\beta_i,\beta_{-i})$ for all $\beta_{-i}$ as required.
\end{proof}

Note that Proposition~\ref{prop:mergingkinterim} reduces to Proposition~\ref{prop:decisionmerging} when $k=0$. The proposition allows for expanded infosets with the same $I$ and observed action availability profile $\mathbf{S}_{I,i}^k$ to be merged. Analogous to the ex-interim case, we can define the resulting DAG coordinate index set as
\begin{equation}
    \Xi_i^k \coloneqq \{ (I, \mathbf{S}_{I,i}^k,a) : I\in\calI_i, \mathbf{S}_{I,i}^k \in \calS_{i}^k(I), a\in S_{I,i}\} \cup \{\varnothing\}.
\end{equation}

The resulting polytope is a DAG-plex. Next, we analyze the size of the DAG-plex in the $k$-ex-interim formalism. For each $I\in\calI_i$, we recall the constants defined earlier: $c_1 = \max_{I\in\calI_i}|\calS_{I,i}|$, $c_3 = \max_{I\in\calI_i} \sum_{S_{I,i}\in\calS_{I,i}} |S_{I,i}|$. Then, the following holds:
\begin{proposition}\label{prop:dagplexsizek}
    Fix a player $i$ in an EFGSAS. The size of their $k$-ex-interim DAG-plex is bounded by
    \[
    |\Xi^k_i| \le 1 + |\calI_i| \cdot c_3 \cdot c_1^{\sum_{r=1}^k B_i^r}.
    \]
    
    In particular, for $B_i >1$, $ | \Xi^k_i| \le 1 + |\calI_i| \cdot c_3 \cdot c_1^{\frac{B_i^{k+1}-B_i}{B_i-1}}$ and for $B_i=1$, $ | \Xi^k_i| \le 1 + |\calI_i| \cdot c_3 \cdot c_1^{k}$. 
\end{proposition}
\begin{proof}
    Consider an arbitrary infoset $I\in\calI_i$, and let $\Xi^k_i(I)$ denote the set of nonempty DAG coordinates associated with decisions at $I$. 
    By definition of $B_i$, the number of infosets belonging to player $i$ in the $r$-th decision layer after $I$ is at most $B_i^r$. Hence, the number of infosets whose availability sets are observed at $I$, excluding $I$ itself, is at most $\sum_{r=1}^k B_i^r$.

   The availability set at each infoset has at most $c_1$ possible realizations, by definition of $c_1$. Thus, the number of possible profiles of future availability sets observed at $I$ is at most $c_1^{\sum_{r=1}^k} B_i^r$.

    For each individual realization of the availability sets, the number of possible pairs consisting of the current availability set $S_{I,i}$ and an action $a\in S_{I,i}$ is $\sum_{S_{I,i}}\in\calS_{I,i} |S_{I,i}| \leq c_3$. By definition of the DAG-plex, each sequence in $\Xi^k_i(I)$ is indexed by a future availability set profile and a pair $(S_{I,i}, a)$. Consequently, it holds that
    \begin{equation}
        | \Xi^k_i (I)| \le  c_3 \cdot c_1^{\sum_{r=1}^k B_i^r}.
    \end{equation}
    
    Note that since the DAG coordinate index set is a disjoint union $\Xi^k_i = \{\varnothing\} \cup \bigcup_{I\in\calI_i} \Xi_i^k(I)$, we have $|\Xi^k_i| = 1 + \sum_{I\in\calI_i} |\Xi_i^k(I)|$. This gives
    \begin{equation}
        |\Xi^k_i| \le 1 + |\calI_i| \cdot c_3 \cdot c_1^{\sum_{r=1}^k B_i^r},
    \end{equation}
    as required. The bounds for $B_i > 1$ and $B_i=1$ in particular follow by taking the largest disclosure factor over all layers and summing the resulting geometric series.
\end{proof}

\section{Omitted Proofs and Details from Section~\ref{sec:efg_compute}}\label{appsecs:efgproofs5}
\subsection{Proof of Theorem~\ref{thm:catcfrregret}}

We first relate ranking regret at one infoset to standard external regret over that infoset. The following lemma shows that, over the set of implementable marginals, these two notions coincide.

\begin{lemma}
\label{lem:rankingexternal}
For every player $i$ and infoset $I\in\calI_i$,
\[
R^{\mathsf{EXT}}_{T,I,i}=R^{\mathsf{RANK}}_{T,I,i}.
\]
\end{lemma}

\begin{proof}
Let
\[
G(a)\coloneqq\sum_{t=1}^T u_{I,i}^t(a),\quad a\in A_{I,i}.
\]

Fix $\mu\in M_{I,i}$, and let $b_i(\cdot\mid S)\in\Delta(S)$, $S\in\calS_{I,i}$, implement $\mu$. Then
\begin{align*}
\langle G,\mu\rangle
&=
\sum_{S\in\calS_{I,i}}\rho_{I,i}(S)\sum_{a\in S}b_i(a\mid S)G(a)\\
&\leq
\sum_{S\in\calS_{I,i}}\rho_{I,i}(S)\max_{a\in S}G(a).
\end{align*}

Let $\kappa^*$ rank the actions in non-increasing order of $G(a)$, using a fixed tiebreaking rule.
For every $S\in\calS_{I,i}$, $\kappa^*(S)$ maximizes $G(a)$ over $a\in S$. The marginal $\mu_{I,i}^{\kappa^*}$ therefore attains the upper bound. Hence
\[
\max_{\mu\in M_{I,i}}\langle G,\mu\rangle=\max_{\kappa\text{ ranking of }A_{I,i}}\langle G,\mu_{I,i}^{\kappa}\rangle.
\]

Subtracting the common learner payoff $\sum_{t=1}^T\langle u_{I,i}^t,\mu_{I,i}^t\rangle$ proves the claim.
\end{proof}

The next corollary states the standard CFR decomposition with respect to $M_{I,i}$ at each infoset.

\begin{corollary}
\label{cor:rankingcfr}
Under Assumption~\ref{assumption:productefg}, player $i$'s external regret over $\Omega_i$ satisfies
\[
R^{\mathsf{EXT}}_{T,i}(\Omega_i)\leq\sum_{I\in\calI_i}\left[R^{\mathsf{RANK}}_{T,I,i}\right]^+.
\]
\end{corollary}

\begin{proof}
Fix \(\omega_i\in\Omega_i\). Let \(\mu_{I,i}^{\omega_i}\) be its local marginal at \(I\), and let \(q_{\omega_i}(I)\in[0,1]\) be the probability of following its own actions leading to \(I\). Replacing the learner's local marginals by those of \(\omega_i\) one infoset at a time in top-down order gives
\[
R^{\mathsf{EXT}}_{T,i}(\Omega_i)=\max_{\omega_i\in\Omega_i}\sum_{I\in\calI_i}q_{\omega_i}(I)\sum_{t=1}^T\left\langle u_{I,i}^t,\mu_{I,i}^{\omega_i}-\mu_{I,i}^t\right\rangle.
\]

The inner sum is at most \(R^{\mathsf{EXT}}_{T,I,i}\). Since \(q_{\omega_i}(I)\in[0,1]\),
\[
R^{\mathsf{EXT}}_{T,i}(\Omega_i)\le\sum_{I\in\calI_i}\left[R^{\mathsf{EXT}}_{T,I,i}\right]^+.
\]

Applying Lemma~\ref{lem:rankingexternal} completes the proof.
\end{proof}

Now we show the proof for~\cref{thm:catcfrregret}.
\catcfrregret*

\begin{proof}[Proof of Theorem~\ref{thm:catcfrregret}]
Fix $i\in\{1,2\}$. At every infoset $I\in\calI_i$, set the parameters
\begin{equation}
\eta_{I,i}=\zeta_{I,i}=\sqrt{\frac{\log|A_{I,i}|}{|A_{I,i}|T}}.
\label{eqn:catcfr_parameters}
\end{equation}

Lemma~\ref{lem:catcfr-feedback} shows that the complete counterfactual payoff vector whose available coordinates are supplied at $I$ is bounded, has conditional mean $u_{I,i}^t$, and is conditionally independent of the current local availability set and the auxiliary samples in the CAT-FTPL update. Consequently, the loss vector $(1-\tilde u_{I,i}^t)/2$ lies in $[0,1]^{|A_{I,i}|}$, has conditional mean $(1-u_{I,i}^t)/2$, and
hence we can apply Proposition~\ref{prop:catftplregret}, which gives
\[
\mathbb{E}\left[\left[R^{\mathsf{RANK}}_{T,I,i}\right]^+\right]\leq16\sqrt{|A_{I,i}|T\log|A_{I,i}|}.
\]

Corollary~\ref{cor:rankingcfr} implies
\[
\left[R^{\mathsf{EXT}}_{T,i}(\Omega_i)\right]^+\leq\sum_{I\in\calI_i}\left[R^{\mathsf{RANK}}_{T,I,i}\right]^+.
\]

Taking expectations and using the local bounds at each infoset gives
\[
\mathbb{E}\left[\left[R^{\mathsf{EXT}}_{T,i}(\Omega_i)\right]^+\right]\leq16\sum_{I\in\calI_i,|A_{I,i}|\geq2}\sqrt{|A_{I,i}|T\log|A_{I,i}|}
\le16\sum_{I\in\calI_i}|A_{I,i}|\sqrt{T}\le16\left|\Sigma^{\mathrm{orig}}_i
\right|\sqrt{T}.
\]

Finally, Markov's inequality immediately yields the second part of the theorem.
\end{proof}

\subsection{Proof of Proposition~\ref{prop:catNEconvergeEFG}}
\begin{proof}
Since $\Omega_i$ is convex, $\bar{\omega}_i^T\in\Omega_i$ for each $i\in\{1,2\}$. Recall that $U=U_1=-U_2$. By bilinearity of $U$,
\begin{align*}
\frac{R^{\mathsf{EXT}}_{T,1}(\Omega_1)}{T}&=\max_{\omega_1'\in\Omega_1}U(\omega_1',\bar{\omega}_2^T)-\frac{1}{T}\sum_{t=1}^T U(\omega_1^t,\omega_2^t),\\
\frac{R^{\mathsf{EXT}}_{T,2}(\Omega_2)}{T}&=\frac{1}{T}\sum_{t=1}^T U(\omega_1^t,\omega_2^t)-\min_{\omega_2'\in\Omega_2}U(\bar{\omega}_1^T,\omega_2').
\end{align*}

Define \(\epsilon_T\coloneqq\frac{R^{\mathsf{EXT}}_{T,1}(\Omega_1)+R^{\mathsf{EXT}}_{T,2}(\Omega_2)}{T}\). Adding these equalities gives
\begin{align*}
\epsilon_T&=\max_{\omega_1'\in\Omega_1}U(\omega_1',\bar{\omega}_2^T)-\min_{\omega_2'\in\Omega_2}U(\bar{\omega}_1^T,\omega_2')\\&=\left[\max_{\omega_1'\in\Omega_1}U(\omega_1',\bar{\omega}_2^T)-U(\bar{\omega}_1^T,\bar{\omega}_2^T)\right]+\left[U(\bar{\omega}_1^T,\bar{\omega}_2^T)-\min_{\omega_2'\in\Omega_2}U(\bar{\omega}_1^T,\omega_2')\right].
\end{align*}

Both bracketed terms are non-negative because $\bar{\omega}_i^T\in\Omega_i$, and hence each is at most $\epsilon_T$. Thus neither player can gain more than $\epsilon_T$ by a unilateral deviation from $(\bar{\omega}_1^T,\bar{\omega}_2^T)$. By \cref{prop:NE-equiv}, any profile $(\chi_1,\chi_2)$ implementing these two strategies is an $\epsilon_T$-Nash equilibrium.

Finally, applying~\cref{cor:rankingcfr} to both players yields
\[
\epsilon_T\leq\frac{1}{T}\sum_{i=1}^2\sum_{I\in\calI_i}\left[R^{\mathsf{RANK}}_{T,I,i}\right]^+.
\]
\end{proof}

\subsection{Stochastic Approximation Details}\label{appsecs:add_SA}

At timestep $t$, let $b_{I,i}^t(\cdot\mid S)\in\Delta(S)$ be the behavioral strategy prescribed by CAT-CFR at infoset $I$ when the available set is $S$, with $b_{I,i}^t(a\mid S)=0$ for $a\notin S$. These rules are determined before the availability sets at timestep $t$ are sampled. To compute a full sequence-form strategy during learning, we sample $S_i^t=(S_{I,i}^t)_{I\in\calI_i}$ at every infoset, including those not visited in a particular play of the game. Starting from $\chi_i^t(S_i^t)[\varnothing]=1$, we construct the sampled sequence-form strategy in topological order by
\[
 \chi_i^t(S_i^t)[(I,a)]
 =\chi_i^t(S_i^t)[\sigma_i(I)]b_{I,i}^t(a\mid S_{I,i}^t),
\]
where $\sigma_i(I)$ is the parent sequence of $I$. Each coordinate is therefore the product of player $i$'s behavioral action probabilities along that sequence.

Let $\mathcal H_t$ contain all previous observations and the current CAT-CFR behavioral strategies, before the availability sets at timestep $t$ are sampled. Under Assumption~\ref{assumption:productefg}, $\rho_i\coloneqq\prod_{I\in\calI_i}\rho_{I,i}$, and the induced marginal sequence-form strategy is
\[
 \omega_i^t=\mathbb E[\chi_i^t(S_i^t)\mid\mathcal H_t]\in\Omega_i.
\]

We distinguish the empirical average of sampled strategies from the average of their induced marginals:
\[
\tilde\omega_i^t\coloneqq\frac1t\sum_{\tau=1}^t\chi_i^\tau(S_i^\tau),\quad
\bar\omega_i^t\coloneqq\frac1t\sum_{\tau=1}^t\omega_i^\tau.
\]

Thus $\tilde\omega_i^t$ estimates $\bar\omega_i^t$. At infoset $I$, the empirical marginal behavioral strategy corresponding to $\tilde\omega_i^t$ is
\begin{equation}\label{eq:efg-sa-local-target}
 \tilde\mu_{I,i}^t(a)
 =\frac{\tilde\omega_i^t[(I,a)]}{\tilde\omega_i^t[\sigma_i(I)]}
\end{equation}
when the denominator is positive, and is uniform over $A_{I,i}$ otherwise.

The SA procedure seeks local weights whose induced marginal behavioral strategies approximate $\tilde\mu_{I,i}^t$. It maintains real parameters $\theta_{I,i}(a)$, called \emph{log-weights} because $\exp(\theta_{I,i}(a))$ are positive, unnormalized action weights. They prescribe the EFGSAS behavioral strategy
\begin{equation}\label{eq:efg-sa-policy}
 \hat b_i(a\mid S;\theta_{I,i})
 =\frac{\mathds{1}\{a\in S\}\exp(\theta_{I,i}(a))}
 {\sum_{a'\in S}\exp(\theta_{I,i}(a'))}.
\end{equation}

Algorithm~\ref{alg:compute_w_efg} starts from zero log-weights, which give uniform action probabilities on each available set. Its update compares the empirical marginal behavioral strategy with the behavioral strategy under the sampled availability set. It can run in tandem with CAT-CFR.

\begin{algorithm}[ht]
\caption{Computing a compact vector for player $i$}
\label{alg:compute_w_efg}
\begin{algorithmic}[1]
\Require $\{\eta_t\}$
\State $\theta_{I,i}^0\gets0\in\mathbb R^{|A_{I,i}|}$ for every $I\in\calI_i$; $\tilde\omega_i^0\gets0$
\For{$t=1,\ldots,T$}
 \State Observe $S_i^t$ and $\chi_i^t(S_i^t)$
 \State $\tilde\omega_i^t\gets\frac{t-1}{t}\tilde\omega_i^{t-1}+\frac1t\chi_i^t(S_i^t)$
 \For{$I\in\calI_i$}
  \State Compute $\tilde\mu_{I,i}^t$ using~\eqref{eq:efg-sa-local-target}
  \State $\theta_{I,i}^t\gets\theta_{I,i}^{t-1}+\eta_t[\tilde\mu_{I,i}^t-\hat b_i(\cdot\mid S_{I,i}^t;\theta_{I,i}^{t-1})]$
 \EndFor
\EndFor
\State $w_{I,i}^T(a)\gets\exp(\theta_{I,i}^T(a))/\sum_{a'\in A_{I,i}}\exp(\theta_{I,i}^T(a'))$ for every $I,a$
\State $W_i^T[\varnothing]\gets1$ and, in topological order, $W_i^T[(I,a)]\gets W_i^T[\sigma_i(I)]w_{I,i}^T(a)$
\State \Return $W_i^T$
\end{algorithmic}
\end{algorithm}

The output $W_i^T$ satisfies the sequence-form constraints of the base game and serves as a compact representation of an EFGSAS strategy. Given availability sets $S_i$, Algorithm~\ref{alg:renorm} renormalizes its weights to obtain a sequence-form strategy $\hat\chi_i(S_i,W_i^T)$ that uses only available actions. For any such positive compact vector $W_i$, define the induced marginal sequence-form strategy by
\[
 \hat\omega_i(W_i)\coloneqq\mathbb E_{S_i\sim\rho_i}[\hat\chi_i(S_i,W_i)].
\]

When $W_i$ is constructed from local weights $w_{I,i}$, the corresponding marginal behavioral strategies are
\[
 \hat\mu_{I,i}(w_{I,i})(a)
 \coloneqq\mathbb E_{S\sim\rho_{I,i}}
 \left[\frac{\mathds{1}\{a\in S\}w_{I,i}(a)}
 {\sum_{a'\in S}w_{I,i}(a')}\right].
\]

Independence across infosets gives
\begin{equation}\label{eq:efg-sa-marginal-recursion}
 \hat\omega_i(W_i)[(I,a)]
 =\hat\omega_i(W_i)[\sigma_i(I)]
 \hat\mu_{I,i}(w_{I,i})(a),
 \quad \hat\omega_i(W_i)[\varnothing]=1.
\end{equation}

All entries of $W_i^T$ are strictly positive by construction, so every normalization denominator in~\cref{alg:renorm} is positive. At each infoset, renormalization implements~\eqref{eq:efg-sa-policy} using only the currently observed availability set, then multiplies by the preceding sequence coordinate. Thus, the compact strategy can be played in the ex-interim regime even though, at each infoset, the player observes only the current availability set.

\begin{algorithm}[ht]
\caption{EFGSAS sequence-form renormalization}
\label{alg:renorm}
\begin{algorithmic}[1]
\Require A positive compact vector $W_i$ and nonempty availability sets $S_i$
\State $\hat\chi_i[\varnothing]\gets1$
\For{$I\in\calI_i$ in topological order}
 \State $\sigma\gets\sigma_i(I)$
 \If{$\hat\chi_i[\sigma]>0$}
  \State $Z_I\gets\sum_{a'\in S_{I,i}}W_i[(I,a')]$
  \For{$a\in A_{I,i}$}
   \State $\hat\chi_i[(I,a)]\gets\hat\chi_i[\sigma]W_i[(I,a)]\mathds{1}\{a\in S_{I,i}\}/Z_I$
  \EndFor
 \Else
  \State $\hat\chi_i[(I,a)]\gets0$ for all $a\in A_{I,i}$
 \EndIf
\EndFor
\State \Return $\hat\chi_i(S_i,W_i)$
\end{algorithmic}
\end{algorithm}

\subsection{Proof of Proposition~\ref{prop:stochasticapproxasymptotic}}
\stochasticapproximation*
\begin{proof}
Let $\calH_t$ contain the history and strategies before the action availability is disclosed at timestep $t$.
Since $\chi_i^t(S_i^t)-\omega_i^t$ are bounded martingale differences, the Azuma--Hoeffding inequality~\citep{azuma1967weighted,hoeffding1963probability} and a union bound give
\[
 \Pr\left(\|\tilde\omega_i^t-\bar\omega_i^t\|_\infty>\varepsilon\right)
 \le2|\Sigma_i^{\mathrm{orig}}|e^{-t\varepsilon^2/2}.
\]

These probabilities are summable, so $\tilde\omega_i^t-\bar\omega_i^t\to0$ almost surely.
As $\bar\omega_i^t\in\Omega_i$ and $\Omega_i$ is closed, the limit $\omega_i^*$ belongs to $\Omega_i$.
Hence, at every infoset with $r_{I,i}^*\coloneqq\omega_i^*[\sigma_i(I)]>0$,
\begin{equation}\label{eq:efg-sa-target-limit}
 \tilde\mu_{I,i}^t(a)\to
 \mu_{I,i}^*(a)\coloneqq
 \frac{\omega_i^*[(I,a)]}{r_{I,i}^*},
 \quad \mu_{I,i}^*\in M_{I,i}.
\end{equation}

Fix one such infoset and suppress its indices. Define
\[
 \Phi(\theta)\coloneqq\mathbb E_{S\sim\rho}\log\sum_{a\in S}e^{\theta(a)},
 \quad
 f(\theta)\coloneqq\Phi(\theta)-\langle\mu^*,\theta\rangle,
 \quad f_*\coloneqq\inf_\theta f(\theta).
\]

Then $\nabla\Phi(\theta)=\hat\mu(\operatorname{softmax}\theta)$.
Each log-sum-exp term has Hessian $\operatorname{diag}(\hat b(\cdot\mid S;\theta))-\hat b(\cdot\mid S;\theta)\hat b(\cdot\mid S;\theta)^\top$ with eigenvalues in $[0,1]$, so $f$ is convex with a $1$-Lipschitz gradient.
Moreover, $\|\nabla f\|_2\le\sqrt2$, since $\nabla\Phi$ and $\mu^*$ are probability vectors.
By definition of implementability for $\mu^*$, it holds that
\[
 \langle\mu^*,\theta\rangle
 \le\mathbb E_{S\sim\rho}\max_{a\in S}\theta(a)
 \le\Phi(\theta),
\]
so $f\ge0$.

We first show that, for every $\delta>0$, there is some $c_\delta>0$ such that
\begin{equation}\label{eq:efg-sa-small-gradient}
 f(\theta)\ge f_*+\delta
 \quad\implies\quad
 \|\nabla f(\theta)\|_2\ge c_\delta.
\end{equation}

Suppose for contradiction that there exists a sequence $(\theta_k)$ where $f(\theta_k)\ge f_*+\delta$ and $\nabla f(\theta_k)\to0$.
There are finitely many availability sets, and each probability simplex is compact.
Hence there are increasing indices $\{k_j\}$ along which all distributions $\hat b(\cdot\mid S;\theta_{k_j})$ converge simultaneously to distributions $q_S$.
Since $\nabla f(\theta_{k_j})\to0$, their marginals satisfy $\sum_S\rho(S)q_S=\mu^*$.
Thus $f(\vartheta)=-\sum_{S,a}\rho(S)q_S(a)\log \hat b(a\mid S;\vartheta)$ for every $\vartheta$, and
\begin{equation}\label{eq:efg-sa-entropy-identity}
 \lim_{j\to\infty}[f(\theta_{k_j})-f(\theta)]
 =\sum_{S,a}\rho(S)q_S(a)
 \log\frac{\hat b(a\mid S;\theta)}{q_S(a)}
 \le0.
\end{equation}

Here, the sums omit terms with $\rho(S)q_S(a)=0$, so the limit follows directly. The inequality follows from $\log x\le x-1$.
Since $f(\theta_{k_j})\ge f_*+\delta$, this implies $f(\theta)\ge f_*+\delta$ for every $\theta$, contradicting the definition of $f_*$.

We now write the SA update as a gradient step with noise:
\[
 \theta^t-\theta^{t-1}
 =\eta_t[-\nabla f(\theta^{t-1})+(\tilde\mu^t-\mu^*)+\xi_t],
 \quad
 \xi_t\coloneqq\nabla\Phi(\theta^{t-1})
 -\hat b(\cdot\mid S^t;\theta^{t-1}).
\]

The noise term satisfies $\mathbb E[\xi_t\mid\calH_t]=0$ and $\|\xi_t\|_2\le\sqrt2$.
Thus $\sum_t\eta_t\xi_t$ converges almost surely by the martingale convergence theorem~\citep{doob1953stochastic} and $\sum_t\eta_t^2<\infty$.

Fix an outcome on which $\sum_{t=1}^\infty\eta_t\xi_t$ converges and $\tilde\mu^t\to\mu^*$, and set
\[
 \vartheta^t\coloneqq\theta^t+\sum_{k>t}\eta_k\xi_k.
\]

Then $\vartheta^t-\theta^t\to0$,
\(
 \vartheta^t-\vartheta^{t-1}
 =\eta_t[\tilde\mu^t-\nabla\Phi(\theta^{t-1})]\), and \(
 \|\vartheta^t-\vartheta^{t-1}\|_2\le\sqrt2\eta_t
\). Writing $\varepsilon_t\coloneqq\|\tilde\mu^t-\mu^*\|_2+
\|\vartheta^{t-1}-\theta^{t-1}\|_2\to0$, smoothness and the gradient bound give
\begin{align}
 f(\vartheta^t)-f(\vartheta^{t-1})
 &\le-\eta_t\|\nabla f(\vartheta^{t-1})\|_2^2
 +\sqrt2\eta_t\varepsilon_t+\eta_t^2,
 \label{eq:efg-sa-descent}\\
 |f(\vartheta^t)-f(\vartheta^{t-1})|
 &\le2\eta_t\to0.
 \label{eq:efg-sa-step-change}
\end{align}

Consider some $\delta>0$. Eventually, $\sqrt2\varepsilon_t+\eta_t\le c_\delta^2/2$ and $2\eta_t\le\delta$.
Whenever $f(\vartheta^{t-1})\ge f_*+\delta$, Equations~\eqref{eq:efg-sa-small-gradient} and~\eqref{eq:efg-sa-descent} then give a decrease of at least $c_\delta^2\eta_t/2$.
Since $\sum_t\eta_t=\infty$ and $f\ge f_*$, the sequence must fall below $f_*+\delta$.
It then stays below $f_*+2\delta$: a step starting below $f_*+\delta$ increases the value by at most $\delta$, and one starting above decreases it.
Thus $f(\vartheta^t)\to f_*$.

Finally, smoothness gives
\[
 f_*\le f(\theta-\nabla f(\theta))
 \le f(\theta)-\tfrac12\|\nabla f(\theta)\|_2^2.
\]

Hence $\nabla f(\vartheta^t)\to0$ and, since $\vartheta^t-\theta^t\to0$,
$\hat\mu(w^t)=\nabla\Phi(\theta^t)\to\mu^*$ almost surely.
There are finitely many infosets, so this holds simultaneously whenever $\omega_i^*[\sigma_i(I)]>0$.

It remains to combine these marginal behavioral strategies into a sequence-form strategy and handle infosets with $\omega_i^*[\sigma_i(I)]=0$.
We use the recursion
\begin{equation}\label{eq:efg-sa-sequence-recursion}
 \hat\omega_i(W_i^t)[(I,a)]
 =\hat\omega_i(W_i^t)[\sigma_i(I)]
 \hat\mu_{I,i}(w_{I,i}^t)(a),
\end{equation}
which follows from independence across infosets and the decoding of $W_i^t$.

Consider induction on the number of player $i$'s actions in a sequence.
For the base case, note that the empty-sequence coordinates both equal $1$.
For the induction step, suppose the shorter parent coordinate $\hat\omega_i(W_i^t)[\sigma_i(I)]$ converges to $r_{I,i}^*$.
If $r_{I,i}^*>0$, local convergence gives
$\hat\omega_i(W_i^t)[(I,a)]\to r_{I,i}^*\mu_{I,i}^*(a)=\omega_i^*[(I,a)]$.
If $r_{I,i}^*=0$, the sequence-form constraint $\sum_a\omega_i^*[(I,a)]=r_{I,i}^*$ and nonnegativity give $\omega_i^*[(I,a)]=0$, while
\[
 0\le\hat\omega_i(W_i^t)[(I,a)]
 \le\hat\omega_i(W_i^t)[\sigma_i(I)]\to0.
\]

Since there are finitely many sequences, this induction proves $\|\hat\omega_i(W_i^T)-\omega_i^*\|_1\to0$ almost surely.


\end{proof}

\subsection{Proof of Theorem~\ref{thm:finitetimeRSAEFG}}\label{appsec:RSAdetails}


For the finite-time result, our algorithm modifies the standard RSA algorithm~\citep{nemirovski1978cezari,nemirovski2009robust} by fixing the empirical average after $T$ CAT-CFR rounds and running for $T$ further iterations using independent availability samples from \(\rho_i\) that are also independent of the CAT-CFR history. At each infoset, it adds the quadratic term $\lambda\|\theta_{I,i}\|_2^2/2$ to the SA objective, giving the additional term $-\lambda\theta_{I,i}^t$ in the update. Algorithm~\ref{alg:compute_w_efg_rsa} then uses suffix averaging~\citep{rakhlin2012making}, with weights proportional to $\eta_t$. Recall \(\hat b_i\) is defined in~\cref{eq:efg-sa-policy}.

\begin{algorithm}[ht]
\caption{Finite-time RSA for player $i$}
\label{alg:compute_w_efg_rsa}
\begin{algorithmic}[1]
\Require An empirical sequence-form strategy $\tilde\omega_i$
\State Compute empirical marginal behavioral strategies $\tilde\mu_{I,i}$ using~\eqref{eq:efg-sa-local-target}
\State $\theta_{I,i}^1\gets0$ for every $I\in\calI_i$; $\lambda\gets T^{-1/2}$; $s\gets\lfloor T/2\rfloor+1$
\For{$t=1,\ldots,T$}
 \State Draw $S_i^t\sim\rho_i$
 \State $\eta_t\gets2/[\lambda(t+1)]$
 \For{$I\in\calI_i$}
  \State $G_{I,i}^t\gets\tilde\mu_{I,i}-\hat b_i(\cdot\mid S_{I,i}^t;\theta_{I,i}^t)-\lambda\theta_{I,i}^t$
  \State $\theta_{I,i}^{t+1}\gets\theta_{I,i}^t+\eta_tG_{I,i}^t$
 \EndFor
\EndFor
\State $\tilde\theta_{I,i}\gets\sum_{t=s}^T\eta_t\theta_{I,i}^t/\sum_{t=s}^T\eta_t$ for every $I$
\State $\tilde w_{I,i}(a)\gets\exp(\tilde\theta_{I,i}(a))/\sum_{a'\in A_{I,i}}\exp(\tilde\theta_{I,i}(a'))$ for every $I,a$
\State $\tilde W_i[\varnothing]\gets1$ and, in topological order, $\tilde W_i[(I,a)]\gets\tilde W_i[\sigma_i(I)]\tilde w_{I,i}(a)$
\State \Return $\tilde W_i$
\end{algorithmic}
\end{algorithm}

We first bound the squared distance between the induced marginal $\hat\mu_{I,i}(\tilde w_{I,i})$ and an implementable marginal $\mu_{I,i}$. This bound accounts for both RSA approximation and the discrepancy between the empirical marginal $\tilde\mu_{I,i}$ and $\mu_{I,i}$. We then bound the resulting duality gap.

\begin{definition}[Duality gap]\label{def:dualitygap}
For $(\omega_1,\omega_2)\in\Omega_1\times\Omega_2$, let $U=U_1=-U_2$ and define
\[
 \gamma(\omega_1,\omega_2)
 \coloneqq\max_{\nu_1\in\Omega_1}U(\nu_1,\omega_2)
 -\min_{\nu_2\in\Omega_2}U(\omega_1,\nu_2).
\]
\end{definition}
The duality gap is non-negative and is zero if and only if $(\omega_1,\omega_2)$ is a Nash equilibrium. For compact strategies, we use the shorthand
$\gamma(\tilde W)=\gamma(\hat\omega_1(\tilde W_1),\hat\omega_2(\tilde W_2))$.

\begin{lemma}\label{lem:efg-rsa-local}
Fix an infoset $I\in\calI_i$ and a local empirical marginal $\tilde\mu_{I,i}\in\Delta(A_{I,i})$ in Algorithm~\ref{alg:compute_w_efg_rsa}. For every $\mu_{I,i}\in M_{I,i}$,
\begin{equation}\label{eq:efg-rsa-local-bound}
 \mathbb E\|\hat\mu_{I,i}(\tilde w_{I,i})-\mu_{I,i}\|_2^2
 \le \frac{10+4\log|A_{I,i}|}{\sqrt T}
 +2\|\tilde\mu_{I,i}-\mu_{I,i}\|_2^2.
\end{equation}
\end{lemma}

\begin{proof}
In the remainder of the proof, we suppress the player and infoset indices. 
As in the preceding proof,
$\Phi(\theta)=\mathbb E_{S\sim\rho}\log\sum_{a\in S}e^{\theta(a)}$
is convex and $\nabla\Phi(\theta)=\hat\mu(\exp\theta)$.
Its gradient is $1/2$-Lipschitz as each log-sum-exp Hessian is symmetric with absolute row sums $2\hat b(a\mid S;\theta)(1-\hat b(a\mid S;\theta))\le1/2$.
The regularized objective
\[
 F_\lambda(\theta)=\Phi(\theta)-\langle\tilde\mu,\theta\rangle
 +\frac\lambda2\|\theta\|_2^2
\]
is coercive and $\lambda$-strongly convex. Its unique minimizer $\theta^*$ satisfies
\begin{equation}\label{eq:efg-rsa-optimality}
 \nabla\Phi(\theta^*)-\tilde\mu+\lambda\theta^*=0.
\end{equation}

Since $\lambda\eta_t\in[0,1]$, the update
\[
 \tilde\mu-\lambda\theta^{t+1}
 =(1-\lambda\eta_t)(\tilde\mu-\lambda\theta^t)
 +\lambda\eta_t \hat b(\cdot\mid S^t;\theta^t),
\]
and $\theta^1=0$ imply that $\tilde\mu-\lambda\theta^t\in\Delta(A)$ by induction.
Thus $-G^t=\hat b(\cdot\mid S^t;\theta^t)-(\tilde\mu-\lambda\theta^t)$ is an unbiased stochastic gradient of $F_\lambda$, with squared norm of at most $2$.
Setting $D_t=\mathbb E\|\theta^t-\theta^*\|_2^2$, strong convexity gives
\begin{equation}\label{eq:efg-rsa-objective-step}
 2\eta_t\mathbb E[F_\lambda(\theta^t)-F_\lambda(\theta^*)]
 \le(1-\lambda\eta_t)D_t-D_{t+1}+2\eta_t^2.
\end{equation}

Since $F_\lambda(\theta^t)-F_\lambda(\theta^*)\ge\lambda\|\theta^t-\theta^*\|_2^2/2$, it also gives
\[
 D_{t+1}\le\left(1-\frac4{t+1}\right)D_t
 +\frac8{\lambda^2(t+1)^2}.
\]

At $t=3$, this yields $D_4\le1/(2\lambda^2)$, and induction gives $D_t\le8/[3\lambda^2(t+1)]$ for every $t\ge4$.
Dropping $-\lambda\eta_tD_t$ in Equation~\eqref{eq:efg-rsa-objective-step}, summing over $t=s,\ldots,T$, where $s=\lfloor T/2\rfloor+1$, and using convexity gives
\begin{equation}\label{eq:efg-rsa-objective-bound}
 \mathbb E[F_\lambda(\tilde\theta)-F_\lambda(\theta^*)]
 \le\frac{D_s+2\sum_{t=s}^T\eta_t^2}{2\sum_{t=s}^T\eta_t}
 \le\frac5{\lambda T}.
\end{equation}

Indeed, Jensen's inequality for $x\mapsto1/x$ gives
\begin{equation}
 \sum_{t=s}^T\eta_t=\frac2\lambda\sum_{t=s}^T\frac1{t+1}\ge\frac{2(T-s+1)^2}{\lambda\sum_{t=s}^T(t+1)}
 =\frac{4(T-s+1)}{\lambda(T+s+2)}.
\end{equation}

Also, $1/(t+1)^2\le1/t-1/(t+1)$, so the following sum telescopes:
\begin{align*}
 \sum_{t=s}^T\eta_t^2
 &=\frac4{\lambda^2}\sum_{t=s}^T\frac1{(t+1)^2}\\
 &\le\frac4{\lambda^2}\sum_{t=s}^T
 \left(\frac1t-\frac1{t+1}\right)\\
 &=\frac4{\lambda^2}\left(\frac1s-\frac1{T+1}\right)
 =\frac{4(T-s+1)}{\lambda^2s(T+1)}.
\end{align*}

Using $(T+1)/2\le s\le(T+2)/2$, the middle expression in~\eqref{eq:efg-rsa-objective-bound} is at most
\begin{equation}
 \frac{T+s+2}{3\lambda(s+1)(T-s+1)}
 +\frac{T+s+2}{\lambda s(T+1)}
 \le\frac5{\lambda T}.
\end{equation}



For any implementable $\mu$, it holds that $\Phi(\theta)-\langle\mu,\theta\rangle\ge0$.
Together with $F_\lambda(\theta^*)\le F_\lambda(0)$ and~\eqref{eq:efg-rsa-optimality}, this bounds the regularization error:
\begin{align*}
 \|\nabla\Phi(\theta^*)-\mu\|_2^2
 &=\|\tilde\mu-\mu\|_2^2+\lambda^2\|\theta^*\|_2^2
 -2\lambda\langle\tilde\mu-\mu,\theta^*\rangle\\
 &\le\|\tilde\mu-\mu\|_2^2
 +2\lambda[\Phi(0)-\Phi(\theta^*)+\langle\mu,\theta^*\rangle]\\
 &\le\|\tilde\mu-\mu\|_2^2+2\lambda\log|A|.
\end{align*}

Convexity and $1/2$-smoothness of $\Phi$ also give
\[
 \|\nabla\Phi(\theta)-\nabla\Phi(\theta^*)\|_2^2
 \le\Phi(\theta)-\Phi(\theta^*)
 -\langle\nabla\Phi(\theta^*),\theta-\theta^*\rangle\le F_\lambda(\theta)-F_\lambda(\theta^*).
\]

Combining these inequalities with~\eqref{eq:efg-rsa-objective-bound}, we get
\[
 \mathbb E\left[\|\nabla\Phi(\tilde\theta)-\mu\|_2^2\right]
 \le\frac{10}{\lambda T}+4\lambda\log|A|
 +2\|\tilde\mu-\mu\|_2^2.
\]

Substituting $\lambda=T^{-1/2}$ and $\nabla\Phi(\tilde\theta)=\hat\mu(\tilde w)$ proves the result.
\end{proof}

The next lemma bounds the payoff difference $|U_i(\hat\omega_i,\nu_{-i})-U_i(\omega_i,\nu_{-i})|$ using the marginal differences $\|\hat\mu_{I,i}-\mu_{I,i}\|_2$, including at infosets with zero reach probability.

\begin{lemma}\label{lem:efg-rsa-reconstruction}
For $\omega_i\in\Omega_i$, write $r_{I,i}=\omega_i[\sigma_i(I)]$ and
$\omega_i[(I,a)]=r_{I,i}\mu_{I,i}(a)$, choosing any $\mu_{I,i}\in M_{I,i}$ when $r_{I,i}=0$.
If $\hat\omega_i$ is induced by local marginals $\hat\mu_{I,i}\in M_{I,i}$, then, for every $\nu_{-i}\in\Omega_{-i}$,
\begin{equation}\label{eq:efg-rsa-reconstruction-bound}
 |U_i(\hat\omega_i,\nu_{-i})-U_i(\omega_i,\nu_{-i})|^2
 \le\left(|\Sigma_i^{\mathrm{orig}}|-1\right)\sum_{I\in\calI_i}r_{I,i}^2
 \|\hat\mu_{I,i}-\mu_{I,i}\|_2^2.
\end{equation}

For any sequence-form strategy $\tilde\omega_i$ in $\calG^{\mathrm{orig}}$, let $\tilde r_{I,i}=\tilde\omega_i[\sigma_i(I)]$ and define its marginal behavioral strategy $\tilde\mu_{I,i}$ by~\eqref{eq:efg-sa-local-target}. Then
\begin{equation}\label{eq:efg-rsa-ratio-bound}
 r_{I,i}^2\|\tilde\mu_{I,i}-\mu_{I,i}\|_2^2
 \le2|\tilde r_{I,i}-r_{I,i}|^2+2\sum_{a\in A_{I,i}}|\tilde\omega_i[(I,a)]-\omega_i[(I,a)]|^2.
\end{equation}
\end{lemma}

\begin{proof}
Suppress the player index and fix an opponent strategy. Replace the marginals $\mu_I$ by $\hat\mu_I$ in reverse topological order. When replacing the marginal at $I$, its ancestors still play according to $\mu$, so its own reach probability is $r_I$. By perfect recall, the total chance and opponent reach probability of histories in $I$ is at most $1$. Since terminal payoffs are in $[-1,1]$, this replacement changes the expected payoff by at most $r_I\|\hat\mu_I-\mu_I\|_1$.
Summing and applying Cauchy--Schwarz over the $|\Sigma_i^{\mathrm{orig}}|-1$ infoset-action pairs gives
\[
 |U_i(\hat\omega,\nu_{-i})-U_i(\omega,\nu_{-i})|
 \le\sum_{I,a}r_I|\hat\mu_I(a)-\mu_I(a)|
 \le\sqrt{|\Sigma_i^{\mathrm{orig}}|-1}
 \left(\sum_I r_I^2\|\hat\mu_I-\mu_I\|_2^2\right)^{1/2}.
\]

Now recall the sequence-form constraints:
\[
 \tilde\omega[\varnothing]=1,\quad
 \tilde\omega[(I,a)]\ge0,\quad
 \sum_{a\in A_{I,i}}\tilde\omega[(I,a)]
 =\tilde\omega[\sigma_i(I)]=\tilde r_I,
\]
and likewise for $\omega$. In particular, $\tilde r_I=0$ implies $\tilde\omega[(I,a)]=0$ for every $a$, so $\tilde\omega[(I,a)]=\tilde r_I\tilde\mu_I(a)$ holds also in this case. Hence
\[
 r_I(\tilde\mu_I-\mu_I)
 =(r_I-\tilde r_I)\tilde\mu_I
 +(\tilde\omega[(I,a)]-\omega[(I,a)])_{a\in A_{I,i}}.
\]

Since $\|\tilde\mu_I\|_2\le1$ and
$\|x+y\|_2^2\le2\|x\|_2^2+2\|y\|_2^2$, this gives
\[
 r_I^2\|\tilde\mu_I-\mu_I\|_2^2
 \le2|r_I-\tilde r_I|^2
 +2\sum_{a\in A_{I,i}}|\tilde\omega[(I,a)]-\omega[(I,a)]|^2,
\]
as required.
\end{proof}

Now we are ready to prove Theorem~\ref{thm:finitetimeRSAEFG}.
\mainthm*
\begin{proof}[Proof of Theorem~\ref{thm:finitetimeRSAEFG}]
Recall that
\[
 \tilde\omega_i^T=\frac1T\sum_{t=1}^T\chi_i^t(S_i^t),
 \quad
 \bar\omega_i^T=\frac1T\sum_{t=1}^T\omega_i^t.
\]


We first bound the sampling error in the sequence weights used to compute each local marginal. The differences $\chi_i^t(S_i^t)-\omega_i^t$ have conditional mean zero.
For each infoset $I$, the sampled weights
$\chi_i^t(S_i^t)[(I,a)]$ correspond to the sequences ending with action $a$ at $I$.
They are non-negative and sum to $\chi_i^t(S_i^t)[\sigma_i(I)]$, the sampled weight of the sequence preceding $I$, which lies in $[0,1]$.
In particular,
\[
 \sum_{a\in A_{I,i}}\bigl(\chi_i^t(S_i^t)[(I,a)]\bigr)^2
 \le\bigl(\chi_i^t(S_i^t)[\sigma_i(I)]\bigr)^2
 \le\chi_i^t(S_i^t)[\sigma_i(I)].
\]


To control the sampling errors in these weights jointly, we bound the sum of their conditional variances.
This will yield a $1/T$ bound on the corresponding mean-squared errors after averaging over $T$ rounds, since errors from different timesteps have zero expected product, as shown below.
Expanding the conditional variances, using the preceding inequality and $\mathbb E[\chi_i^t(S_i^t)\mid\mathcal H_t]=\omega_i^t$ gives
\[
 \operatorname{Var}(\chi_i^t(S_i^t)[\sigma_i(I)]\mid\mathcal H_t)
 +\sum_{a\in A_{I,i}}\operatorname{Var}(\chi_i^t(S_i^t)[(I,a)]\mid\mathcal H_t)\le2\omega_i^t[\sigma_i(I)]-(\omega_i^t[\sigma_i(I)])^2\le1.
\]

For every sequence $\sigma$ and $\tau<t$,
\(
 \mathbb E\left[
 \bigl(\chi_i^\tau(S_i^\tau)[\sigma]-\omega_i^\tau[\sigma]\bigr)
 \bigl(\chi_i^t(S_i^t)[\sigma]-\omega_i^t[\sigma]\bigr)
 \right]=0
\),
because the first factor is $\mathcal H_t$-measurable and the second has conditional mean zero given $\mathcal H_t$.
Thus, when expanding the squared averages, the products of errors from different timesteps have zero expectation. Using the variance bound above gives
\begin{equation}\label{eq:efg-rsa-empirical-coordinate}
 \mathbb E\left[
 |\tilde\omega_i^T[\sigma_i(I)]-\bar\omega_i^T[\sigma_i(I)]|^2
 +\sum_{a\in A_{I,i}}|\tilde\omega_i^T[(I,a)]-\bar\omega_i^T[(I,a)]|^2
 \right]\le\frac1T.
\end{equation}

Using analogous notation as before, let $\bar r_{I,i}^T\coloneqq\bar\omega_i^T[\sigma_i(I)]$. Since $\bar\omega_i^T\in\Omega_i$, its marginal behavioral strategies $\bar\mu_{I,i}^T$ are implementable when $\bar r_{I,i}^T>0$. If $\bar r_{I,i}^T=0$, the preceding sequence has probability zero under $\bar\omega_i^T$, and we choose any $\bar\mu_{I,i}^T\in M_{I,i}$.
Conditionally on the CAT-CFR history, RSA uses independent samples. As such, we can apply Lemma~\ref{lem:efg-rsa-local} with input $\tilde\mu_{I,i}^T$ and comparison marginal $\bar\mu_{I,i}^T$. Then, multiplying the inequality in the lemma by $(\bar r_{I,i}^T)^2$, summing over infosets, and applying Equations~\eqref{eq:efg-rsa-ratio-bound} and~\eqref{eq:efg-rsa-empirical-coordinate} give
\[
 \mathbb E\left[\sum_{I\in\calI_i}(\bar r_{I,i}^T)^2
 \|\hat\mu_{I,i}(\tilde w_{I,i})-\bar\mu_{I,i}^T\|_2^2\right]
 \le
 \frac{10|\calI_i|+4\sum_{I\in\calI_i}\log|A_{I,i}|}{\sqrt T}
 +\frac{4|\calI_i|}{T}.
\]

Let $E_i$ denote the largest absolute change in player $i$'s expected payoff, over all opponent strategies, when $\bar\omega_i^T$ is replaced by the RSA output $\hat\omega_i(\tilde W_i)$: 
\[
 E_i\coloneqq\sup_{\nu_{-i}\in\Omega_{-i}}
 |U_i(\hat\omega_i(\tilde W_i),\nu_{-i})-U_i(\bar\omega_i^T,\nu_{-i})|.
\]

By Lemma~\ref{lem:efg-rsa-reconstruction} and since $T^{-1}\le T^{-1/2}$, we have $\mathbb E\left[E_i^2\right]\le C_i/\sqrt T$, where
\[
 C_i\coloneqq(|\Sigma_i^{\mathrm{orig}}|-1)\left(14|\calI_i|+4\sum_{I\in\calI_i}\log|A_{I,i}|\right).
\]

The expected regret bound in the proof of Theorem~\ref{thm:catcfrregret} and Proposition~\ref{prop:catNEconvergeEFG} give
\[
 \mathbb E\left[\gamma(\bar\omega_1^T,\bar\omega_2^T)\right]
 \le\frac{16\left(|\Sigma_1^{\mathrm{orig}}|+|\Sigma_2^{\mathrm{orig}}|\right)}{\sqrt T}.
\]

Since payoffs lie in $[-1,1]$, we also have $0\le\gamma(\bar\omega_1^T,\bar\omega_2^T)\le2$.
By the definition of $E_i$, maximizing and minimizing over the respective opponent strategies gives
\(
 \gamma(\tilde W)
 \le\gamma(\bar\omega_1^T,\bar\omega_2^T)+E_1+E_2
\). Using $(a+b+c)^2\le3(a^2+b^2+c^2)$ and $\gamma(\bar\omega_1^T,\bar\omega_2^T)^2\le2\gamma(\bar\omega_1^T,\bar\omega_2^T)$, we thus obtain
\[
 \mathbb E\left[\gamma(\tilde W)^2\right]
 \le\frac{96(|\Sigma_1^{\mathrm{orig}}|+|\Sigma_2^{\mathrm{orig}}|)+3(C_1+C_2)}{\sqrt T}.
\]

Finally, Markov's inequality gives, with probability at least $1-p$,
\begin{equation}
 \gamma(\tilde W)^2
 \le\frac{96(|\Sigma_1^{\mathrm{orig}}|+|\Sigma_2^{\mathrm{orig}}|)+3(C_1+C_2)}{p\sqrt T}
 \in O\left(\frac1{p\sqrt T}\right).
\end{equation}
\end{proof}

\section{Additional Experimental Details}\label{appsecs:add_example}
\subsection{Game definitions}
\label{appsubsec:game_definitions}
In this section we introduce and formally define a version of Kuhn poker~\citep{kuhn1950simplified} with stochastic action sets and Binary Pursuit games with stochastic action sets.
\paragraph{SAS-Kuhn poker EFGSAS}
Consider an EFGSAS where the base game \(\calG^\mathrm{orig}\) is Kuhn poker. 
At player 1's second infoset after receiving a Queen, at which they have played \(\mathsf{Check}\) and observed player 2 has played \(\mathsf{Bet}\), they receive the action availability \(\{\mathsf{Fold}\}\) with probability $2/3$ and otherwise receive \(\{\mathsf{Fold, Call}\}\).
At player 2's infoset for receiving a Jack and observing player 1 to have played \(\mathsf{Check}\), they receive the action availability \(\{\mathsf{Bet}\}\) with probability \(1/3\) and otherwise receive \(\{\mathsf{Check, Bet}\}\). We refer to this game as \textit{SAS-Kuhn poker}. 
\paragraph{\(n \times n\) Stochastic Binary Pursuit (SBP) EFGSAS}
Consider an EFGSAS where the base game \(\calG^\mathrm{orig}\) is a pursuit-evasion game \citep{parsons_pursuit-evasion_1978} played on an \(n \times n\) grid graph. In a pursuit-evasion game there are two players, an \emph{attacker} and a \emph{defender}, who both start at some different node in the graph. A game consists of several rounds in which both players traverse an edge to another node. Players only observe their own position and their history. The defender wins (i.e., payoff \(1\)) if in any round they traverse the same edge in the same direction as the attacker, and loses (i.e., payoff \(-1\)) otherwise, with the attacker receiving opposite payoff to the defender. In the stochastic variant, an edge is available for player \(i\) by sampling a Bernoulli variable with parameter \(p_{e,i}\), resampling all edges if none are available. We refer to this game as \textit{\(n \times n\) SBP EFGSAS}. 

We randomly generate \(n \times n\) SBP EFGSAS by sampling \(p_{e,i} \sim \mathrm{Uniform}[0.3, 0.8]\). We also choose starting vertices for the attacker and defender. In order to have equal theoretical bounds across runs we restrict the starting vertices for the defender to a random corner vertex, and restrict the attacker to a random interior vertex an even number of edges away from the defender's starting vertex. In \(5 \times 5\) SBP EFGSAS we must additionally restrict the attacker to a random interior vertex which is not the center vertex. We use 7 rounds for \(4 \times 4\) SBP EFGSAS and 9 rounds for \(5 \times 5\) SBP EFGSAS.

\subsection{Additional details}
All experiments are conducted on an Intel Core Ultra 9 285K, 5.4 GHz on a Linux 64-bit platform using Python version 3.14.7 with NumPy version 2.5.3 and SciPy version 1.18.1. Central 95\% intervals, where reported, were computed using \texttt{numpy.quantile}. To ensure reproducibility, all randomness used in the algorithms was generated by a NumPy random generator seeded with the current run number, i.e., if an experiment performed 100 runs, the \(r\)th run seeded the random number generator with \(r\).

All experiments except the \(5 \times 5\) SBP EFGSAS in~\cref{fig:binary_pursuit_5x5} were run on a single thread. We run CAT-CFR with the parameters \(\eta_{I,i}=\zeta_{I,i}=\sqrt{\log|A_{I,i}|/(|A_{I,i}|T)}\)
as required for the theoretical bounds in~\cref{eqn:catcfr_parameters}. In~\cref{fig:spr_compact} and~\cref{fig:app_spr_compact} the same time-averaged marginal output by CAT-CFR was used for all~\cref{alg:compute_w_efg_rsa} runs in order to allow comparisons across runs, with CAT-CFR run for a time horizon of \(T=80,000\) for all games.
All experiments done on~\cref{ex:updown} are run with \(\alpha=0\) and \(\lambda=0.5\).
In~\cref{fig:updown_compact} at an intermediate iterate \(t\) we perform suffix averaging from \(\lfloor t/2 \rfloor+1\) to \(t\).
We give the approximate wallclock running time of CAT-CFR on the EFGSAS used in our experiments in~\cref{tab:exp-runtime-catcfr}, and for~\cref{alg:compute_w_efg_rsa} in~\cref{tab:exp-runtime-rsa}. For~\cref{tab:exp-runtime-catcfr}, note the size of \(\calG^\mathrm{orig}\) for an \(n \times n\) SBP EFGSAS grows very quickly with \(n\) due to the exponential growth of the treeplex as the number of rounds and nodes increases.

We calculate the duality gap of a pair of marginal sequence form strategies by, for each player, traversing the treeplex bottom-up and calculating the expected value of the best response to that marginal sequence with regards to the stochastic action availability of the responding player, using the known exact probability of the action availability. We calculate the duality gap of a compact vector \(\tilde W\) by converting into into the equivalent marginal behavioral form strategy.

\begin{table}[t]
\centering
\caption{Approximate runtime of CAT-CFR for various games.}
\label{tab:exp-runtime-catcfr}
\begin{tabular}{@{}lrrrrrr@{}}
\toprule
Game & \multicolumn{1}{l}{Iterations} & \multicolumn{1}{l}{\(\lvert \Sigma_1^\mathrm{orig} \rvert\)} & \multicolumn{1}{l}{\(\lvert \Sigma_2^\mathrm{orig} \rvert\)} & \multicolumn{1}{l}{Runs} & \multicolumn{1}{l}{Wallclock time per run (s)} & \multicolumn{1}{l}{Stddev (s)} \\ \midrule
\cref{ex:updown} & 80,000 & 3 & 7 & 100 & 8.2 & 0.08 \\
SAS-Kuhn & 80,000 & 13 & 13 & 100 & 9.28 & 0.05 \\
\(4 \times 4\) SBP & 80,000 & 2,275 & 873 & 100 & 351 & 81 \\
\(5 \times 5\) SBP & 80,000 & 33,931 & 11,391 & 5 & 53,120 & 148 \\ \bottomrule
\end{tabular}
\end{table}
\begin{table}[t]
\centering
\caption{Approximate runtime of~\cref{alg:compute_w_efg_rsa} for various games.
}
\label{tab:exp-runtime-rsa}
\begin{tabular}{@{}lrrrrrr@{}}
\toprule
Game & \multicolumn{1}{l}{Iterations} & \multicolumn{1}{l}{\(\lvert \Sigma_1^\mathrm{orig} \rvert\)} & \multicolumn{1}{l}{\(\lvert \Sigma_2^\mathrm{orig} \rvert\)} & \multicolumn{1}{l}{Runs} & \multicolumn{1}{l}{Wallclock time per run (s)} & \multicolumn{1}{l}{Stddev (s)} \\ \midrule
\cref{ex:updown} & 6,000 & 3 & 7 & 100 & 0.101 & 0.000378 \\
SAS-Kuhn & 6,000 & 13 & 13 & 100 & 0.104 & 0.00204 \\
\(4 \times 4\) SBP & 6,000 & 2,275 & 873 & 100 & 0.695 & 0.00517 \\
\(5 \times 5\) SBP & 10,000 & 33,931 & 11,391 & 100 & 16.9 & 0.106 \\
\bottomrule
\end{tabular}
\end{table}

\subsection{Additional figures}\label{appsubsec:additional_figures}
In~\cref{fig:updown_compact_full}, we show the compact vector \(W^t_i\) entries obtained analytically in~\cref{example:lsg_compact} as in~\cref{fig:updown_compact}, for all RSA iterates to the time horizon \(T=10,000\).
In~\cref{fig:regret_plots_intermediate}, we show the ranking regret of CAT-CFR intermediate.
In~\cref{fig:regret_plots_final}, we show the ranking regret of CAT-CFR at the final iterate for various time horizons \(T\), and compare with the theoretical bounds from~\cref{thm:catcfrregret}.
In~\cref{fig:app_spr_marginal}, we show the duality gap of the time-averaged empirical marginals as shown in~\cref{fig:spr_marginal}, split into the individual games. In~\cref{fig:app_spr_compact} we show the duality gap of the strategy induced by the compact vector produced by~\cref{alg:compute_w_efg_rsa} for various time horizons \(T\), as shown in~\cref{fig:spr_compact}, split into the individual games and with the theoretical bounds from~\cref{thm:finitetimeRSAEFG}.
Finally, in~\cref{fig:binary_pursuit_5x5} we show the ranking regret at intermediate iterates (as in~\cref{fig:regret_plots_intermediate}) and the duality gap of the time-averaged marginal (as in~\cref{fig:app_spr_marginal}) for larger \(5 \times 5\) SBP EFGSAS.
All plots show the mean and central 95\% interval over 100 runs, except for the \(5 \times 5\) SBP EFGSAS in~\cref{fig:regret_plots_intermediate_5x5,fig:spr_marginal_5x5} which show the same over 5 runs. In the legends, $\tilde R$ denotes empirical regret and H.P. bound stands for high-probability bound.

\begin{figure}[ht]
  \centering
    \includegraphics[width=0.49\linewidth]{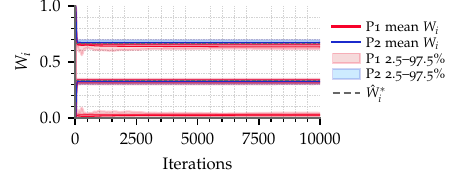}
    \caption{\(W^t_i\) entries obtained via CAT-CFR and~\cref{alg:compute_w_efg_rsa}, 
    showing all RSA iterates to the time horizon \(T=10,000\). }
    \label{fig:updown_compact_full}
\end{figure}

\begin{figure}[ht]
    \centering
    \includegraphics[width=0.9\linewidth]{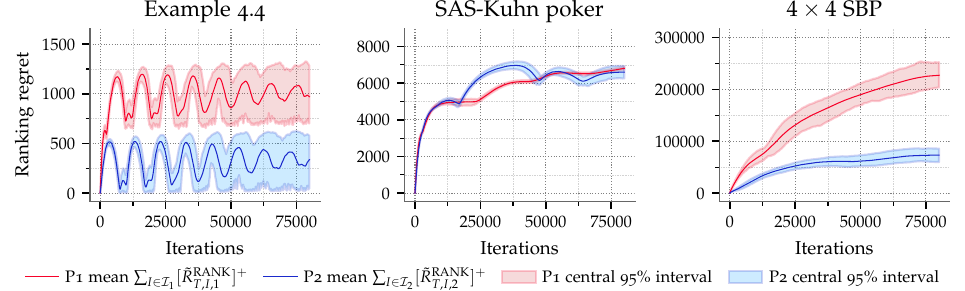}
    \caption{Ranking regret at intermediate iterates for fixed time horizon. Mean and 95\% central interval shown over 100 runs for each game shown.
    }
    \label{fig:regret_plots_intermediate}
\end{figure}

\begin{figure}[ht]
    \centering
    \includegraphics[width=0.9\linewidth]{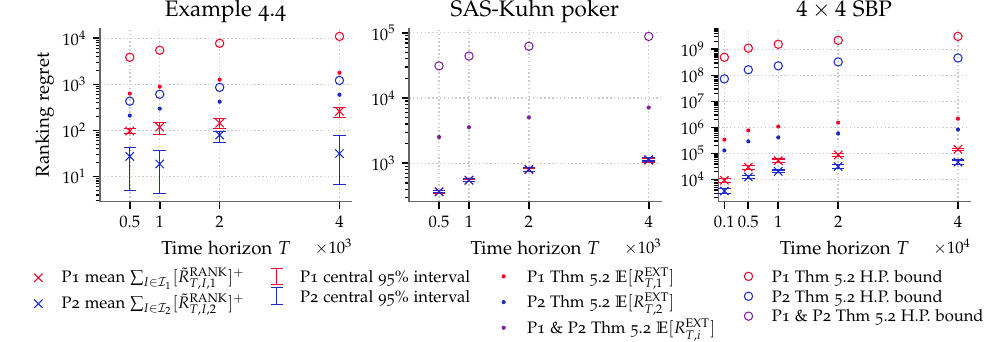}
    \caption{Ranking regret at final iterate for various time horizons \(T\).
    Mean and 95\% central interval shown over 100 runs for each game shown. Theoretical bounds in expectation and with high probability from~\cref{thm:catcfrregret}  also shown.
    }
    \label{fig:regret_plots_final}
\end{figure}
\begin{figure}[ht]
    \centering
    \includegraphics[width=\linewidth]{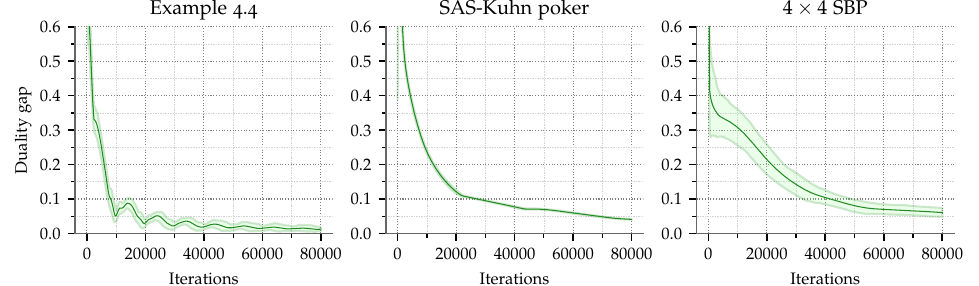}
    \caption{Duality gap of the time-averaged empirical marginals \(\tilde \omega_i^t\) played by CAT-CFR at intermediate iterates for time horizon \(T=80,000\). Mean and 95\% central interval shown over 100 runs for each game shown.}
    \label{fig:app_spr_marginal}
\end{figure}
\begin{figure}[ht]
    \centering
    \includegraphics[width=\linewidth]{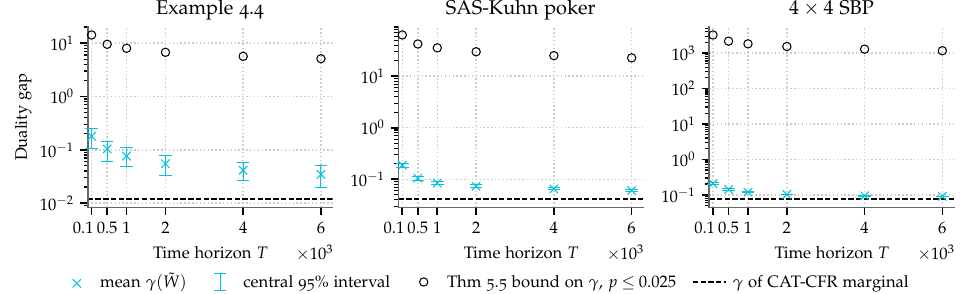}
    \caption{Duality gap of strategies induced by the $\tilde W_i^T$ vectors computed by \cref{alg:compute_w_efg_rsa} for various time horizons \(T\). Mean and 95\% central interval shown over 100 runs for each game shown. \cref{thm:finitetimeRSAEFG} bound is shown for \(p\leq0.025\), to match the one-sided high-probability bound with the two-sided central 95\% interval.}
    \label{fig:app_spr_compact}
\end{figure}

\begin{figure}[ht]
    \centering
    \begin{subfigure}[t]{0.32\linewidth}
        \vspace{0pt}
        \includegraphics[width=\linewidth]{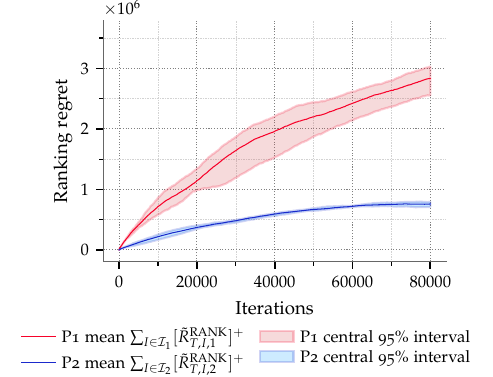}
        \caption{Ranking regret at intermediate iterates, as in~\cref{fig:regret_plots_intermediate}, over 5 runs.}
        \label{fig:regret_plots_intermediate_5x5}
    \end{subfigure}
    \hfill
    \begin{subfigure}[t]{0.32\linewidth}
        \vspace{0pt}
        \includegraphics[width=\linewidth]{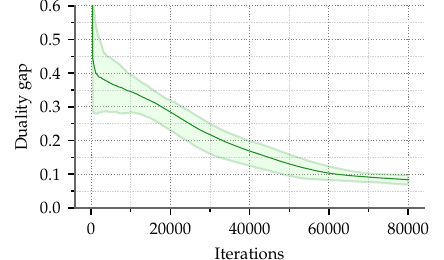}
        \caption{Duality gap of the time-averaged empirical marginals, as in~\cref{fig:spr_marginal,fig:app_spr_marginal}, over 5 runs.}
        \label{fig:spr_marginal_5x5}
    \end{subfigure}
    \begin{subfigure}[t]{0.32\linewidth}
        \vspace{0pt}
        \includegraphics[width=\linewidth]{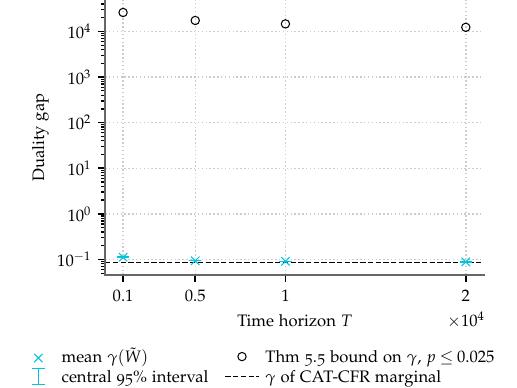}
        \caption{Duality gap of the strategies induced by the \(\tilde W_i^T\) vectors computed by~\cref{alg:compute_w_efg_rsa} for various time horizons \(T\), as in \cref{fig:spr_compact,fig:app_spr_compact}. We run CAT-CFR for a time horizon of \(80,000\) and show the mean and central 95\% interval of ~\cref{alg:compute_w_efg_rsa} over 100 runs for each \(T\).}
        \label{fig:spr_compact_5x5}
    \end{subfigure}
    \caption{Results for \(5 \times 5\) SBP for fixed time horizon \(T=80,000\).}
    \label{fig:binary_pursuit_5x5}
\end{figure}

\clearpage
\section{Notation Table}\label{appsecs:table}
\begin{table}[H]\small
\caption{Notation table for symbols used throughout the main text}
    \label{tab:symbols}
	\begin{tabular}{
			p{\dimexpr 0.14\linewidth - 2\tabcolsep\relax}
			p{\dimexpr 0.12\linewidth - 2\tabcolsep\relax}
			p{\dimexpr 0.7\linewidth - 2\tabcolsep\relax}
		}
		\toprule
		Section & Symbol & Meaning\\
		\midrule
		EFG & \(\calH\) & States in the game tree\\
		& \(\calZ\) & Set of terminal nodes\\
		& \(A_h\) & Possible actions at state \(h\)\\
		& \(\calN\) & Set of players\\
		& \(u_i\), \(U_i\) & Terminal and expected payoffs for player \(i\), respectively\\
		& \(\calI_i\) & Information sets (infosets) for player \(i\)\\
		& \(I\) & An infoset \(I\in\calI_i\)\\
		& \(A_{I,i}\) & Action set at infoset \(I\) for player \(i\)\\
		& \(\beta_i\) & Behavioral strategy for player \(i\)\\
		& \(\Sigma_i\) & Set of sequences of player \(i\)\\
		& \(\sigma\) & A sequence \(\sigma\in\Sigma_i\)\\
		& \(\sigma(I)\) & Parent sequence of infoset \(I\)\\
		\midrule
		EFGSAS & \(\calG\) & An EFGSAS\\
		& \(\calG^{\mathrm{orig}}\) & Base EFG for an EFGSAS\\
		& \(\calG^\dagger\) & Expanded EFG for an EFGSAS\\
		& \(\calI^\dagger\) & Infosets in the expanded EFG\\
		& \(\calS_{I,i}\) & Set of all availability sets at infoset \(I\) for player \(i\)\\
		& \(S_{I,i}\) & An availability set from \(\calS_{I,i}\)\\
		& \(\rho_{I,i}\) & Probability distribution over \(\calS_{I,i}\)\\
		& \(\calS\), \(\rho\) & Ensemble of joint availability-profiles and corresponding probability distributions\\
		& \(S_i\) & Availability realization \((S_{I,i})_{I\in\calI_i}\) for player \(i\)\\
		& \(\mathrm{obs}(I^\dagger)\) & Availability set observed at expanded infoset \(I^\dagger\)\\
		\midrule
		EFGSAS & \(\Sigma_i^{\mathrm{orig}}\) & Sequences of player \(i\) in the base EFG\\
		strategies & \(\Sigma_i^\dagger\) & Sequences of player \(i\) in the expanded EFG\\
		& \(S_{I,i}^\dagger\) & DAG-plex node corresponding to availability set \(S_{I,i}\)\\
		& \(\Xi_i\) & Coordinate index set for player \(i\)'s DAG-plex\\
		& \(\xi_i(S_{I,i}^\dagger)\) & Set of predecessor pairs of \(S_{I,i}^\dagger\)\\
		& \(\chi_i\) & DAG realization vector for player \(i\)\\
		& \(b_i\) & Behavioral strategy for player \(i\)\\
		& \(\mu_{I,i}(a)\) & Marginal probability of action \(a\) at infoset \(I\) for player \(i\)\\
		& \(\mu_i\) & Ensemble of marginal behavioral strategies over \(I\in\calI_i\) for player \(i\)\\
		& \(\omega_i\) & Implementable sequence-form strategy for player \(i\)\\
		& \(\Omega_i\) & Set of implementable sequence-form strategies for player \(i\)\\
		& \(\chi_i(S_i)\) & Base-game realization vector induced by \(\chi_i\) under \(S_i\)\\
		& \(W_i\) & Positive compact weight vector indexed by \(\Sigma_i^{\mathrm{orig}}\)\\
		\midrule
		Equilibrium & \(M_{I,i}\) & Set of implementable marginal behavioral strategies at \(I\) for player \(i\) \\
		computation & \(\kappa\), \(\kappa(S)\) & A ranking over actions and the highest-ranked action in \(S\)\\
		& \(\mu_{I,i}^{\kappa}\) & Marginal behavioral strategy induced by ranking \(\kappa\) for player \(i\)\\
		& \(u_{I,i}^t\) & Expected counterfactual payoff vector at \(I\) and timestep \(t\) for player \(i\)\\
		& \(R_{T,i}^{\mathsf{EXT}}\) & Cumulative external regret of player \(i\) over \(T\) timesteps\\
		& \(R_{T,I,i}^{\mathsf{RANK}}\) & Ranking regret at infoset \(I\) for player \(i\) over \(T\) timesteps\\
		& \(R_{T,I,i}^{\mathsf{EXT}}\) & External regret at infoset \(I\) for player \(i\) over \(T\) timesteps\\
		& \(\bar\omega_i^T\) & Time-averaged marginal sequence-form strategy\\
		& \(\tilde\omega_i^T\) & Empirical average of sampled realization vectors\\
		& \(\hat\omega_i(W_i)\) & Implementable sequence-form strategy induced by \(W_i\)\\
        & \(\tilde W\) & Profile of compact vectors returned by the finite-time RSA procedure\\
		& \(\gamma(\tilde W)\) & Duality gap of the strategy profile induced by \(\tilde W\)\\
		\bottomrule
	\end{tabular}
\end{table}
\end{document}